%% file: avg_therm.tex
\documentclass[aps,pra,superscriptaddress,12pt,longbibliography]{revtex4-2}

\renewcommand{\thesection}{\arabic{section}}
\renewcommand{\thesubsection}{\thesection.\arabic{subsection}}
\renewcommand{\thesubsubsection}{\thesubsection.\arabic{subsubsection}}
\makeatletter
\renewcommand{\p@subsection}{}
\renewcommand{\p@subsubsection}{}
\makeatother

\usepackage{amsmath,amsfonts,amssymb,amsthm}
\usepackage{thmtools, thm-restate}
\usepackage{mathtools}
\usepackage{setspace}
\usepackage[normalem]{ulem} 
\usepackage{stmaryrd}       
\usepackage{dsfont}         
\usepackage{expdlist}
\usepackage{physics}
\usepackage{braket}
\usepackage{verbatim}
\usepackage{xspace}
\usepackage{bm}
\usepackage{tikz}
\usetikzlibrary{arrows,automata}
\usetikzlibrary{shapes.geometric}
\usetikzlibrary{arrows.meta}
\usepackage[latin1]{inputenc}
\usepackage{verbatim}
\usepackage{textcomp} 

\usepackage{graphicx}
\usepackage{xcolor}
\usepackage{hyperref}
\hypersetup{pdfstartview=}
\usepackage{url}
\usepackage{ytableau}

\usepackage{alltt}
\usepackage{moreverb}

\usepackage{xr}

\input{defs.tex}

\begin{document}
\title{Measuring Multiparticle Indistinguishability with the Generalized Bunching Probability}
\begin{abstract}
  The indistinguishability of many bosons undergoing passive linear transformations followed by number basis measurements is fully characterized by the visible state of the bosons.
  However, measuring all the parameters in the visible state is experimentally demanding.
  In this work, we seek to perform partial characterization of the visible state by measuring properties of it that are available after randomization.
  First we study the case where the occupied visible modes are randomly permuted, and second we study the case where Haar random linear optical unitaries are applied.
  In each case, we find that the generalized bunching probability---which is the probability that all the input bosons arrive in a given subset of the output modes---obeys monotonicity with respect to some partial order of distinguishability of the input bosons.
  As an intermediate result, we show that Lieb's permanental-dominance conjecture for immanants is equivalent to the following statement: for states that are invariant under permutations of the occupied visible modes, the generalized bunching probability is maximized when the bosons are perfectly indistinguishable.
  We also prove that a consequence of the monotonicity of the generalized bunching probability after Haar averaging is that this average is maximized when the bosons are perfectly indistinguishable.
  Finally, we discuss applications of our results to thermometry of cold-atom systems.
\end{abstract}
\author{Shawn Geller}
\affiliation{National Institute of Standards and Technology, Boulder, Colorado 80305, USA}
\date{\today}
\author{Emanuel Knill}
\affiliation{National Institute of Standards and Technology, Boulder, Colorado 80305, USA}
\affiliation{Center for Theory of Quantum Matter, University of Colorado, Boulder, Colorado 80309, USA}
\maketitle

\section{Introduction}

Bosons are fundamentally indistinguishable particles, but the presence of uncontrolled degrees of freedom can make them behave as though they are not.
The indistinguishability of two bosons can be quantified with the Hong-Ou-Mandel visibility~\cite{hongMeasurementSubpicosecondTime1987}, which is directly related to the probability that the input bosons arrive in the same output mode.
The indistinguishability of bosons determines the extent to which they interfere in a linear-optics experiment, and this interference is typically the quantum resource that we wish to harness in such experiments.
For example, the computational hardness of boson sampling~\cite{aaronsonComputationalComplexityLinear2011} is lost when the bosons are distinguishable~\cite{renema2019classicalsimulabilitynoisyboson}.
As another example, linear optical quantum computation~\cite{Knill2001} relies on the interference of many bosons to perform quantum computations.
In this work, we are interested in quantifying the $n$-particle indistinguishability.
To illustrate the concept, we discuss the example of many photons interfering in a linear-optics experiment.

Consider an $m$-mode interferometer, $n$ modes of which are populated with one photon each.
The photons can have any frequency, polarization, and temporal mode, but have definite spatial modes.
The photons propagate through an interferometer that acts only on the spatial degree of freedom of the photons, then are measured in the number basis, resolving only which spatial mode the photons arrived in.
Because the interferometer acts only on the spatial degree of freedom (DOF), it is called the visible DOF, while all other DOFs (frequency, time, polarization) are called hidden.
A schematic of this setup is shown in Fig.~\ref{fig:setupfig}.
\begin{figure}[h]
  \centering
  \includegraphics[width=.4\columnwidth]{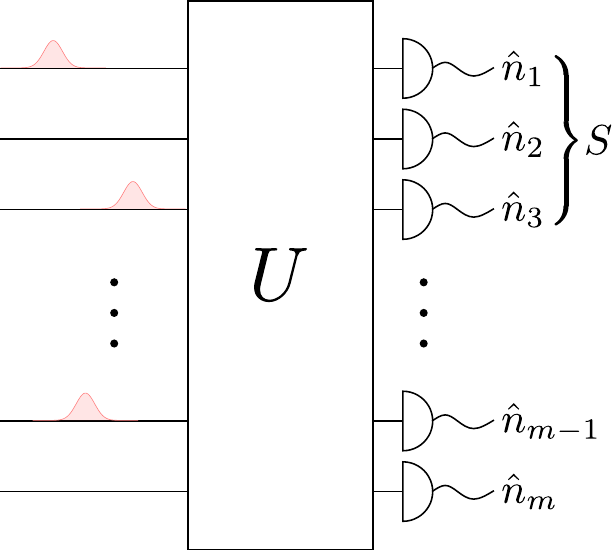}
  \caption{Interferometer with $n$ input bosons, with at most one in each of $m$ spatial modes. The interferometer acts only on the spatial degree of freedom, but the bosons may have some other hidden DOFs, such as time. We perform number-basis measurements at the end, which do not resolve the hidden DOFs. The generalized bunching probability~\cite{shchesnovichUniversalityGeneralizedBunching2016} is the probability that all the input photons arrive in a given subset of the visible output modes.}
  \label{fig:setupfig}
\end{figure}

There have been many efforts to
describe~\cite{englbrecht2024indistinguishability,tichyInterferenceIdenticalParticles2014,dufour2024fourieranalysismanybodytransition,tichySamplingPartiallyDistinguishable2015,PhysRevA.91.013844,PhysRevA.89.022333,Seron:2023yyu,DittelPhdThesis,stanisic2020quantum,dittel2021wave,spivak2020immanants,spivak2022generalized}
and
detect~\cite{Ou2006,Seron2024efficientvalidation,anguita2025experimentalvalidationbosonsampling,PhysRevA.103.023722,PRXQuantum.6.020340,stanisic2018discriminating,PhysRevA.108.053701,PhysRevA.91.013811,PRXQuantum.2.020326,PRXQuantum.2.020326,PhysRevLett.122.063602,doi:10.1073/pnas.1206910110,PhysRevLett.125.123603,PhysRevA.83.062111,Giordani_2020,PhysRevA.100.053829,PhysRevA.83.062307,PhysRevLett.120.240404,PhysRevA.97.062116,Giordani2018,Walschaers_2016,PhysRevLett.104.220405,PhysRevX.12.031033,Crespi2016,youngAtomicBosonSampler2024,brunner2023many,PhysRevA.91.063842,tillmann2015generalized}
the effects of multiparticle indistinguishability on these
interference experiments.  
The complete description of indistinguishability effects on multiboson interference experiments can be expressed in terms of the distinguishability function of the bosons, introduced in Ref.~\citenum{PhysRevA.91.013844} as the $J$-matrix, then later in Ref.~\citenum{shchesnovichUniversalityGeneralizedBunching2016} as the $J$-function.
An equivalent description is given by the visible
state of the bosons (Eq.~\ref{eq:visiblepart}) ~\cite{geller2024characterization,englbrecht2024indistinguishability,stanisic2018discriminating}
(called the ``external state'' in~\cite{englbrecht2024indistinguishability}).  
Multiboson interference experiments are tomographically complete for the visible state~\cite{englbrecht2024indistinguishability}, so in principle one can measure the visible state from such interference experiments to obtain a characterization of the effect of indistinguishability on such experiments.
However, these measurements are experimentally challenging, both because of the exponentially many parameters needed to describe the visible state, and because of the high sensitivity to interferometer calibration errors.
Furthermore, the perfectly indistinguishable component of a state is typically exponentially small~\cite{dittel2024distinguishability}, so it is statistically demanding to measure precisely.  
Here we argue that the generalized bunching probability~\cite{shchesnovichUniversalityGeneralizedBunching2016}, which is the probability that all the input photons arrive in a given subset of visible output modes, provides useful partial information about the multiparticle indistinguishability.  
The generalized bunching probability is easier to measure than the visible state  because it remains large even when the number of particles is large~\cite{shchesnovich_distinguishing_2021}, and in this work we show it obeys monotonicity relative to some partial orders of distinguishability.

Analogous to the two-particle case, one might conjecture that the generalized bunching probability is maximized when all the input photons are simultaneously indistinguishable, meaning that they all have the same polarization, temporal mode, and frequency mode.
This conjecture is equivalent to the permanent-on-top conjecture~\cite{soules1966matrix}, which is false~\cite{shchesnovich2016permanent}.
The conjecture that, among states for which the auxiliary state (Def.~\ref{def:auxiliarystate}) is separable (called ``classically correlated states'' in Ref.~\citenum{shchesnovichUniversalityGeneralizedBunching2016}), the generalized bunching probability is maximized when the bosons are perfectly indistinguishable, is equivalent to a conjecture due to Bapat and Sunder~\cite{bapat1985majorization}.
This conjecture was shown to be false~\cite{drury2016counterexample}, and the connection to the bunching behavior of bosons was observed in Ref.~\citenum{seronBosonBunchingNot2023}.
We show in this work that if one can randomize the input state, then the generalized bunching probability does serve as a witness of maximal indistinguishability of the input bosons.
We study two forms of randomization: one where the occupied input modes are randomly permuted, the other where Haar random linear optical unitaries are applied.

To study the case that random permutations are an available experimental control, we introduce the weak generalized-bunching conjecture (Conj.~\ref{conj:weakgenbunch}): Among states that are invariant under permutations of the occupied visible modes, the generalized bunching probability is maximized when the bosons are perfectly indistinguishable.
We emphasize that our conjecture applies to any choice of applied linear optical and any choice of subset of visible modes.
We prove in Thm.~\ref{thm:weakgenbunchequiv} that the weak generalized-bunching conjecture is equivalent to an existing mathematical conjecture, Lieb's permanental-dominance conjecture for immanants (Conj.~\ref{conj:liebpermanentaldominance})~\cite{liebProofsConjecturesPermanents2002,wanlessLiebPermanentalDominance2022}.
Lieb's conjecture states that the permanent is the largest among the normalized immanants for any positive-semidefinite argument.
Despite significant efforts to resolve it~\cite{Tran_2022,sym13101782,tabata2010sharp,grone1986hadamard,johnson1988inequalities,doi:10.1137/0609016,pate1992descending,RODTES202494,pate1991partitions,shchesnovich2016permanent,merris1985inequalities,pate1999tensor,rodtes2024some,jing2025immanantinequalitiesweightspaces}, Lieb's conjecture has remained open for over fifty years.
Assuming Lieb's conjecture, this establishes that, given access to random permutations of the occupied visible modes, one can use the generalized bunching probability to detect that the input bosons are perfectly indistinguishable.

Analogous to the Hong-Ou-Mandel visibility providing a quantitative notion of the indistinguishability of two bosons, we ask: Does the deviation of the generalized bunching probability from its maximum value give a measure of the extent to which the input bosons are ``more distinguishable''?
To formalize this question, we need a rigorous notion of ``more distinguishable''.
We now describe one such notion, and make use of a few definitions to do so.

The visible DOF $\mathcal{H}_V$ and hidden DOF $\mathcal{H}_H$ are tensor factors of the single-particle Hilbert space, $\mathcal{H}_1 = \mathcal{H}_V \otimes \mathcal{H}_H$.
We assume that there is a preferred basis $\left\{ \ket{i}_V \right\}_{i=1}^m$ of $\mathcal{H}_V$, which we call sites, and a preferred basis $\left\{ \ket{j}_H \right\}_{j=1}^L$ of $\mathcal{H}_H$, which we call labels.
For a partition $\lambda$ of the number of populated modes $n$ and a list of distinct site indices $\bm{i} = (i_1, \ldots, i_n)$, consider the state that (1) has $\lambda_j$ many particles with label $j$, (2) is invariant under permutations of the occupied modes, and (3) has one boson in each of the distinct sites $\bm{i}$.
This state is called the partially labelled state with label pattern $\lambda$ occupying $\bm{i}$ (Def.~\ref{def:partiallylabelledstate}).
We show in Thm.~\ref{thm:refinementmonotonicity} that the generalized bunching probability for partially labelled states is monotonic with respect to the refinement partial order on partitions (Def.~\ref{def:refinementofpartitions}), provided Lieb's permanental-dominance conjecture holds.

We now discuss our results relating the generalized bunching probability to the indistinguishability of the input bosons in the case that Haar random linear optical unitaries are an available experimental control.
Define the mean generalized bunching probability to be the average over the Haar measure on unitaries \(U\) as in Fig.~\ref{fig:setupfig} of the generalized bunching probability.
We show in Thm.~\ref{thm:meangenbunchsuremajorization} that the mean generalized bunching probability is monotonic with respect to a certain preorder (Def.~\ref{def:suremajorization}) on all $n$-boson states.
Because the definition of this preorder is technical, we defer its discussion to App.~\ref{sec:schurconvexity}, and present some consequences of Thm.~\ref{thm:meangenbunchsuremajorization} here.
First we show in Cor.~\ref{cor:avglieb} that as a consequence of Thm.~\ref{thm:meangenbunchsuremajorization}, the mean generalized bunching probability is maximized when the bosons are perfectly indistinguishable.
Define a uniform state to be one where the input bosons all have the same single-particle density matrix on $\mathcal{H}_H$ (Def.~\ref{def:uniformstate}).
Then second, we prove in Cor.~\ref{cor:schurconvexity} that, for uniform states, the mean generalized bunching probability is Schur convex with respect to the eigenvalues of the single-particle density matrix.
Third, when the single-particle density matrix is a Gibbs state, we show in Cor.~\ref{cor:thermcor} that such a mean generalized bunching probability can be used to infer the temperature of this Gibbs state.

The generalized bunching probability can be efficiently measured in many scenarios; see Ref.~\citenum{shchesnovich_distinguishing_2021} for a detailed analysis of the number of measurements required.
Combined with our monotonicity results, this suggests that the generalized bunching probability is a good measure of the joint indistinguishability of the bosons.

The rest of the paper is structured as follows.
In Sec.~\ref{sec:mainprelim} we present preliminaries.
In Sec.~\ref{sec:weakgeneralizedbunching} we introduce the weak generalized-bunching conjecture and state our theorem showing its equivalence to Lieb's conjecture.
In Sec.~\ref{sec:schurconvexityofmeangenbunch} we state corollaries to Thm.~\ref{thm:meangenbunchsuremajorization}, relating the mean generalized bunching probability to the indistinguishability of the input bosons.
In Sec.~\ref{sec:discussion} we conclude and offer remarks on future work, including applications of our results to thermometry of cold-atom systems.

\section{Preliminaries}
\label{sec:mainprelim}

Lists of indices are written in boldface, e.g. $\bm{i} = (i_1, \ldots, i_n)$.
For $N \in \mathbb{N}$, we use the abbreviation $[N] := \left\{ 1, \ldots, N \right\}$.
Let 
\begin{align}
\Delta^{N-1} = \left\{ \alpha \in \R^N \middle| \alpha_i \ge 0~\forall i \in [N], \sum_i \alpha_i = 1 \right\}
\end{align}
be the set of probability distributions on $N$ elements.

Let $W$ be a finite dimensional Hilbert space.
We write $\mathcal{B}_+(W)$ for the set of positive-semidefinite matrices on $W$.
The set of density matrices on $W$ is denoted $\mathcal{D}(W) = \left\{ \rho \in \mathcal{B}_+(W)|\Tr(\rho) = 1 \right\}$.
Density matrices are sometimes referred to as states, and vectors $\ket{\psi}\in W$ are also referred to as states.

\subsection{Linear optics}
Let $W$ be a finite dimensional Hilbert space.
Let $\mathrm{U}(W)$ be the group of unitaries on $W$.
We write $\mathrm{Sym}^n(W)$ for the $n$th symmetric power of $W$.
The Fock space over $W$ is defined to be
  \begin{align}
    \mathcal{F}(W) &= \bigoplus_{n=0}^{\infty}\mathrm{Sym}^n(W).
  \end{align}
The vacuum state is written $\ket{0}$.
The annihilation map $a:W \rightarrow \mathrm{End}(\mathcal{F}(W))$ is a conjugate-linear map that assigns an operator $a_\psi$ to each $\ket{\psi} \in W$.
The operator $a_\psi$ is called an annihilation operator, and its Hermitian conjugate $a^\dagger_\psi$ is called a creation operator.
The annihilation operators satisfy
\begin{align}
  a_\psi \ket{0} = 0,
\end{align}
and if $\ket{\psi'}\in W$, we have the commutation relations
\begin{align}
  [a_\psi, a^\dagger_{\psi'}] &= \braket{\psi|\psi'}.
  \label{eq:wfncommutations}
\end{align}
If $B = \left\{ \ket{x} \right\}_x$ is an orthonormal basis of $W$, an orthonormal basis of $\mathrm{Sym}^n(W)$ can be obtained from
\begin{align}
  \ket{g} = \prod_{\ket{x}\in B}\frac{(a^\dagger_x)^{g(x)}}{\sqrt{g(x)!}}\ket{0},
  \label{eq:ketg}
\end{align}
where $g:[\dim(W)] \rightarrow \left\{ 0, \ldots, n \right\}$ is a function such that $\sum_x g(x) = n$.
Such states $\ket{g}$ are called states of definite occupation.
Accordingly, we define the set of occupations with total particle number $n$ on the set $S$,
\begin{align}
  \omega^n_S &= \left\{ g:S\rightarrow \{0, \ldots, n\}\middle| \sum_x g(x) = n \right\}.
  \label{eq:occupations}
\end{align}

The unitary group $\mathrm{U}(W)$ acts on $\mathrm{End}(\mathcal{F}(W))$ according to the linear-optical action:
\begin{definition}[Linear-optical action]
The linear-optical action on $\mathrm{Sym}^n(W)$ by unitaries $U \in \mathrm{U}(W)$ is given by
\begin{align}
  U \cdot \ket{g} = \prod_{x=1}^{\mathrm{dim}(W)}\frac{(a^\dagger_{Ux})^{g(x)}}{\sqrt{g(x)!}}\ket{0}.
\end{align}
This gives rise to an action on $\mathcal{F}(W)$, and a corresponding adjoint action on $\mathrm{End}(\mathcal{F}(W))$, which are both also called the linear-optical action.
\label{def:linearopticalaction}
\end{definition}

\subsection{Integer partitions}
\label{sec:partitionsmaintext}
We follow conventions in Ref.~\citenum{fultonYoungTableauxApplications1996} for denoting partitions.
A partition of a positive integer $n$ is a nonincreasing list of nonnegative integers $(\lambda_1, \ldots, \lambda_k)$ such that $\sum_{j=1}^k \lambda_j = n$.
If $\lambda$ is a partition of $n$, we write $\lambda \vdash n$.
If two partitions differ only by their trailing zeros, they are considered the same partition.
The number of nonzero elements of $\lambda\vdash n$ is called its length, written $\mathrm{len}(\lambda)$.
Suppose $\lambda\vdash n$ has $m_i$ many elements that are equal to $i$ for $i \in [n]$.
Then $\lambda$ is also denoted by the string $(n^{m_n}\ldots 1^{m_1})$, where $i^m$ denotes the string containing $m$ copies of $i$, $i^0$ is the empty string, and the juxtaposition of strings indicates the concatenation of strings.

\section{Generalized bunching for permutation invariant states}
\label{sec:weakgeneralizedbunching}
In this section we study the problem of characterizing the indistinguishability of bosons given the ability to perform arbitrary permutations of the visible modes and some nontrivial linear optical unitary.
To that end, we present our conjecture that, under certain conditions, the generalized bunching probability is maximized by perfectly indistinguishable particles.
We consider a linear-optics experiment where the bosons have a hidden DOF.
The single-particle Hilbert space $\mathcal{H}_1$ is a tensor product, $\mathcal{H}_1 = \mathcal{H}_V \otimes \mathcal{H}_H$.
Let $m = \dim(\mathcal{H}_V)$ and $L = \dim(\mathcal{H}_H)$.
We assume there is a preferred basis $\left\{ \ket{i}_V \right\}_{i=1}^m$ of the visible DOF $\mathcal{H}_V$ called sites, and the hidden DOF has a preferred basis $\left\{ \ket{l}_H \right\}_{l=1}^L$ called labels.
We start in a state $\rho \in \mathcal{D}(\mathrm{Sym}^n(\mathcal{H}_V \otimes \mathcal{H}_H))$ of $n$ particles, evolve with linear-optical dynamics under $U \otimes \mathds{1}_{\mathcal{H}_H}$ for a unitary $U \in \mathrm{U}(\mathcal{H}_V)$, then measure in the number basis, not resolving the hidden DOF.
To express the resulting probability distribution in terms of $\rho$ and $U$, we introduce the following definition.
Let $v \in \omega_{[m]}^n$ (Eq.~\ref{eq:occupations}).
We define the set of all occupations for which the marginal visible occupation is equal to $v$ by
\begin{align}
  N_v = \left\{ g \in \omega_{[m]\times [L]}^n\middle|\sum_{j \in [L]}g(i, j) = v(i) \, \forall i \in [m] \right\}.
  \label{eq:nvdef}
\end{align}
Then the probability that, for all $j\in [m]$, we observe $v(j)$ many particles in the visible site $j$ is
\begin{align}
  p(v|U, \rho) &= \sum_{g \in N_v}\left( \left((U \otimes \mathds{1}_{\mathcal{H}_H})\cdot \rho \cdot (U^\dagger \otimes \mathds{1}_{\mathcal{H}_H})\right)\ketbra{g}{g}\right),
  \label{eq:visibleoccupationdistributions}
\end{align}
where the state of definite occupation $\ket{g}$ is defined in Eq.~\ref{eq:ketg}.
In the above expression, the dot is the linear-optical action (Def.~\ref{def:linearopticalaction}).
The probability distribution $p(v|U, \rho)$ is called the visible occupation distribution of $\rho$ under $U$.

We are interested in states that have at most one boson per visible site.
If $\bm{i} \in [m]^{n}$ satisfies $i_x = i_y \implies x = y$, we say that $\bm{i}$ is a list of $n$ distinct indices from $[m]$. 
If $\bm{i} \in [m]^{n}$ and $\bm{j} \in [L]^{n}$, we define $a^\dagger(\bm{i}, \bm{j}) := a^\dagger_{i_1, j_1}\cdots a^\dagger_{i_n, j_n}$.
\begin{definition}[Singly-occupied subspace]
  Let $\bm{i}$ be a list of $n$ distinct indices from $[m]$.
The space of states that singly occupies the sites \(\bm{i}\) is defined by 
  \begin{align}
    K_{\bm{i}} = \operatorname{span}(\left\{a^\dagger(\bm{i}, \bm{j})\ket{0}\middle|\,\bm{j} \in [L]^n\right\}) \subseteq \mathrm{Sym}^n(\mathcal{H}_V \otimes \mathcal{H}_H).
  \end{align}
Correspondingly, we say that a state $\rho \in \mathcal{D}(\mathrm{Sym}^n(\mathcal{H}_1))$ singly occupies the sites $\bm{i}$ if it can be written in the form
\begin{align}
  \rho = \sum_{\bm{j}, \bm{j}'\in [L]^{n}}\rho_{\bm{j}, \bm{j}'}a^{\dagger}(\bm{i}, \bm{j})\ket{0}\!\!\bra{0}a(\bm{i}, \bm{j}').
  \label{eq:singlyoccupiesdef}
\end{align}
for some numbers $\rho_{\bm{j},\bm{j}'}\in\mathbb{C}$.
    \label{def:kisubspacedef}
\end{definition}
Our conjecture concerns the subspace of $K_{\bm{i}}$ that is invariant under exchange of the elements of $\bm{i}$. 
To define this subspace, we introduce an action of the symmetric group $\mathcal{S}_n$ on $K_{\bm{i}}$.
\begin{definition}[Permutation matrix on $\bm{i}$]
  Let $\bm{i}$ be a list of $n$ distinct indices from $[m]$.
For $\sigma \in \mathcal{S}_n$, let $F_{\bm{i}}(\sigma) \in \mathrm{U}(\mathcal{H}_V)$ be the permutation matrix whose action is given by $
F_{\bm{i}}(\sigma)\ket{i_x} = \ket{i_{\sigma(x)}}$ for all $x \in [n]$, and $F_{\bm{i}}(\sigma)\ket{j} = \ket{j}$ for any $j$ not in the list $\bm{i}$.
\label{def:frepdef}
\end{definition}
The states we consider are a special class of states that are invariant under the action of $F_{\bm{i}}$, in the sense defined next. 
\begin{definition}[Permutation-invariant state]
  Let $\rho$ be a state that singly occupies the $n$ distinct sites $\bm{i}$ (Def.~\ref{def:kisubspacedef}). We say $\rho$ is permutation invariant if $(F_{\bm{i}}(\sigma)\otimes \mathds{1}_{\mathcal{H}_H})\cdot \rho\cdot (F_{\bm{i}}(\sigma)\otimes \mathds{1}_{\mathcal{H}_H})^{-1} = \rho$ for all $\sigma \in \mathcal{S}_n$.
  \label{def:permutationinvariantstates}
\end{definition}
Permutation-invariant states are commonly considered in the literature~\cite{dittel2024distinguishability,shchesnovich_distinguishing_2021}.
Any state that singly occupies the sites $\bm{i}$ can be made permutation invariant by passing it through a linear-optical device that applies a uniformly random permutation to the sites $\bm{i}$.
We note that some permutation-invariant states are entangled in a certain sense that we define in App.~\ref{sec:commentsonentanglement}; we refer the reader to this appendix for further details.

The perfectly indistinguishable state with single-particle hidden state $\ket{\psi} \in \mathcal{H}_H$ is defined to be
\begin{align}
    \phi_{\bm{i}, \psi} &= a^\dagger_{i_1, \psi}\cdots a^\dagger_{i_n, \psi}\ket{0}\!\!\bra{0}a_{i_1, \psi}\cdots a_{i_n, \psi}\,.
    \label{eq:perfectlyindsitinguishable}
\end{align}
These states span the set of states that behave bosonically under visible unitaries, see App.~\ref{sec:commentsonentanglement} for further details.
We write $\phi_{\bm{i}} = \phi_{\bm{i}, 1}$, where $\ket{1}_{H}$ is the first label, and refer to $\phi_{\bm{i}}$ as the perfectly indistinguishable state, for short.
The perfectly indistinguishable state describes bosons that all have the same label, and in particular, is permutation invariant.
Observe that the definition of the perfectly indistinguishable state does not depend on the order of the elements of $\bm{i}$.

We want to compare the generalized bunching probability for an arbitrary permutation-invariant state to that for the perfectly indistinguishable state.
Let $\rho\in \mathcal{D}(\mathrm{Sym}^n(\mathcal{H}_V\otimes \mathcal{H}_H))$ be a state, and let $U \in \mathrm{U}(\mathcal{H}_V)$ be a unitary.
We write $b(S|U, \rho)$ for the probability that the particles start in the state $\rho$, evolve under the linear-optical action determined by $U$, then all arrive in $S\subseteq [m]$.
This probability is known as the generalized bunching probability of $\rho$ into $S$ under $U$.
To express $b(S|U, \rho)$ explicitly in terms of the quantum state and the unitary, we introduce the following definition.
We define the set of occupations (Eq.~\ref{eq:occupations}) that are nonzero only on a subset of modes $S \subseteq S'$ as
\begin{align}
  \omega_{S'}^n(S) := \left\{ g \in \omega^n_{S'}\middle|\,x \notin S \implies g(x) = 0  \right\}.
  \label{eq:occupationsnonzero}
\end{align}
Then the generalized bunching probability is
\begin{align}
  b(S|U, \rho) &= \sum_{v \in \omega_{[m]}^n(S)} p(v|U, \rho)\\
  &= \sum_{g\in \omega_{[m] \times [L]}^n(S\times [L])}\Tr\left( \left((U \otimes \mathds{1}_{\mathcal{H}_H})\cdot \rho \cdot (U^\dagger \otimes \mathds{1}_{\mathcal{H}_H})\right)\ketbra{g}{g}\right).
  \label{eq:genbunchfirstexpression}
\end{align}
We are now prepared to state our conjecture.
\begin{conjecture}[Weak Generalized Bunching]
  Let $\rho$ be a permutation-invariant state that singly occupies the $n$ distinct sites $\bm{i}$ (Def.~\ref{def:permutationinvariantstates}).
  Let $\phi_{\bm{i}}$ be the perfectly indistinguishable state defined in Eq.~\ref{eq:perfectlyindsitinguishable}.
  Let $U \in \mathrm{U}(\mathcal{H}_V)$, and let $S \subseteq [m]$.
  Then
  \begin{align}
    b(S|U, \rho) \le b(S|U, \phi_{\bm{i}}).
  \end{align}
  \label{conj:weakgenbunch}
\end{conjecture}
We show in Thm.~\ref{thm:weakgenbunchequiv} that the weak generalized-bunching conjecture is equivalent to an existing mathematical conjecture, known as Lieb's permanental-dominance conjecture for immanants.
To state Lieb's conjecture, we need a few more definitions.
Let $M_{n \times n}(\mathbb{C})$ be the set of $n$ by $n$ complex matrices.
If $\lambda\vdash n$, the irreducible character $\chi_\lambda$ of $\mathcal{S}_n$ is defined to be the character of the Specht module of shape $\lambda$, as defined in Thm.~7.18.5 of Ref.~\citenum{Stanley_Fomin_1999}, and Eq.~\ref{eq:irrechdef} 
(see also Refs.~\citenum{fultonRepresentationTheory2004,fultonRepresentationTheory2004,vershikNewApproachRepresentation2005,goodmanSymmetryRepresentationsInvariants2009,Etingof2011,Procesi2006-on,kondorGroupTheoreticalMethods2008,harrowApplicationsCoherentClassical2005,wright2016learn}).
Let $e\in \mathcal{S}_n$ be the identity element.
\begin{definition}
  Let $\lambda \vdash n$, and let $\chi_\lambda$ be the corresponding irreducible character of $\mathcal{S}_n$.
  The normalized immanant of shape $\lambda$ is defined to be
  \begin{align}
    \overline{\operatorname{Imm}}_\lambda:M_{n \times n}(\mathbb{C})\rightarrow \mathbb{C},\qquad   \overline{\operatorname{Imm}}_\lambda(A) &= \frac{1}{\chi_\lambda(e)}\sum_\sigma\chi_\lambda(\sigma)\prod_{i=1}^n A_{i, \sigma(i)}.
  \end{align}
  \label{def:immanants}
\end{definition}

When $\lambda = (1^n)$,  $\overline{\operatorname{Imm}}_{(1^n)}$ is equal to the determinant.
When $\lambda = (n)$, $\overline{\operatorname{Imm}}_{(n)}$ is equal to the permanent.
Thus, the normalized immanants interpolate between the determinant and permanent.
When the argument $A$ is positive semidefinite, the normalized immanant $\overline{\operatorname{Imm}}_\lambda(A)$ is nonnegative~\cite{schur1918endliche}.
It is known that, for a positive-semidefinite argument, the determinant is the smallest among the normalized immanants~\cite{schur1918endliche}.
Lieb's conjecture~\cite{liebProofsConjecturesPermanents2002} is that the permanent is the largest.
\begin{conjecture}[Lieb's Permanental-Dominance Conjecture for Immanants]
  Let $A\in \mathcal{B}_+(\mathbb{C}^n)$ be an $n \times n$ positive-semidefinite matrix.
  Let $\lambda \vdash n$. Then
  \begin{align}
    \overline{\operatorname{Imm}}_\lambda(A) \le \overline{\operatorname{Imm}}_{(n)}(A)\,.
  \end{align}
    \label{conj:liebpermanentaldominance}
\end{conjecture}
Lieb's conjecture has been open for over 50 years, despite many attempts to resolve it.
Here we relate it to the physical scenario at hand:
\begin{restatable}{theorem}{weakgenequivlieb}
  The weak generalized-bunching conjecture (Conj.~\ref{conj:weakgenbunch}) is equivalent to Lieb's permanental-dominance conjecture for immanants (Conj.~\ref{conj:liebpermanentaldominance}).
  \label{thm:weakgenbunchequiv}
\end{restatable}
This theorem, proven in App.~\ref{sec:weakgenbunchproof}, provides strong evidence supporting the weak generalized-bunching conjecture.
It also implies that any partial progress towards Lieb's conjecture has a direct corollary for the bunching behavior of bosons.

While our conjecture states that the perfectly indistinguishable state achieves the largest generalized bunching probability among permutation-invariant states, it does not make any direct statements about the relationship between different values of the generalized bunching probability.
Specifically, we would like a statement that captures the intuition that bosons that are more distinguishable have a lower generalized bunching probability.
We show that, assuming the weak generalized-bunching conjecture, there is a partial order on partially labelled states (defined below) with respect to which the generalized bunching probabilities are monotonic.

For $\mu \vdash n$, we write $\overline{\mu}$ for the nondecreasing list of $n$ numbers containing $\mu_i$ many $i$s.
 For a list of indices $\bm{j} \in [d]^{n}$ and a permutation $\sigma \in \mathcal{S}_n$, define the action of $\sigma$ on $\bm{j}$ by
\begin{align}
  \sigma \cdot (j_1, \ldots, j_n) &= (j_{\sigma^{-1}(1)}, \ldots, j_{\sigma^{-1}(n)}).
\end{align}
We are now prepared to define partially labelled states.
\begin{definition}[Partially labelled state]
  Let $\bm{i}$ be a list of $n$ distinct indices from $[m]$, and let $\mu \vdash n$ be a partition of $n$ such that $\mathrm{len}(\mu) \le L = \mathrm{dim}(\mathcal{H}_H)$. 
  Then the partially labelled state with label pattern \(\mu\) occupying \(\bm{i}\) is defined to be
  \begin{align}
    \rho_{\bm{i}, \mu} = \frac{1}{n!}\sum_{\sigma\in \mathcal{S}_n} a^\dagger(\sigma\cdot\bm{i}, \overline{\mu})\ketbra{0}{0}a(\sigma\cdot\bm{i}, \overline{\mu}).
  \end{align}
  \label{def:partiallylabelledstate}
\end{definition}
Partially labelled states have a definite number of particles in each label, and the assignment of labels to particles is symmetrized.
In particular, partially labelled states are permutation invariant.
They are a complete set of permutation-invariant states in the following sense.
In Cor.~\ref{cor:perminvpdistlinind}, we show that if $\rho$ is permutation invariant then there exist real coefficients $\left\{ c_\mu \right\}_{\mu\vdash n}$  such that, for all $U\in \mathrm{U}(\mathcal{H}_V)$ and for all $v \in \omega_{[m]}^n$, the probability of the visible occupation $v$ is given by $p(v|U, \rho) = \sum_{\mu \vdash n}c_\mu p(v|U, \rho_{\bm{i}, \mu})$.

We can order partially labelled states occupying \(\bm{i}\) according to the refinement order of their label patterns, which we now define.
Loosely speaking, a partition $\mu \vdash n$ is a refinement of $\lambda\vdash n$ if it can be obtained by partitioning the parts of $\lambda$, as depicted in Fig.~\ref{fig:refinement}.
\begin{figure}[h]
  \centering
  \includegraphics[width=.5\columnwidth]{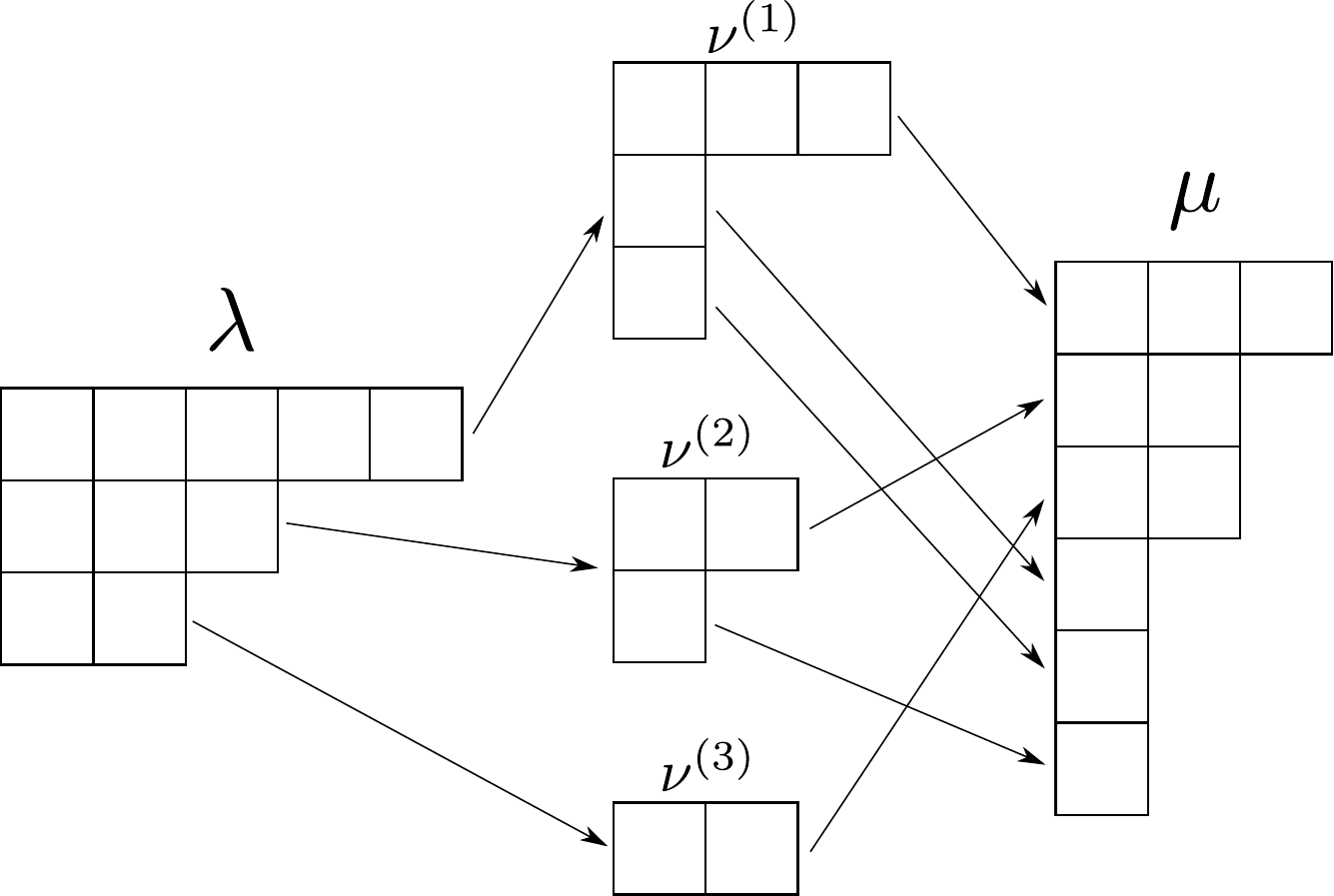}
  \caption{A schematic description of refinement of partitions. $\mu$ is obtained by partitioning the parts of $\lambda$.}
  \label{fig:refinement}
\end{figure}
To define refinement formally, we need the following definition.
Let $\kappa\vdash n$ and $\pi\vdash n'$.
Let $k = \max(n, n')$.
We can write $\kappa = (k^{m_k}\ldots 1^{m_1})$, $\pi = (k^{m'_k}\ldots 1^{m'_1})$, where for all $j \in [k]$, $m_j, m'_j$ are nonnegative integers. 
Then define the sum of partitions to be
\begin{align}
  \kappa+\pi &= (k^{m_k+m'_k}\ldots 1^{m_1+m'_1}).
\end{align}
\begin{definition}[Refinement partial order]
  Let $\lambda, \mu \vdash n$ be partitions. We say that $\mu$ is a refinement of $\lambda$, and write $\lambda \trianglerighteq \mu$, if $\mu$ can be expressed as 
  \begin{align}
    \mu &= \sum_{i \in [n]}\nu^{(i)},
  \end{align}
  for some partitions $\nu^{(i)}$ such that $\nu^{(i)} \vdash \lambda_i$.
  \label{def:refinementofpartitions}
\end{definition}
With these definitions in hand, we now state a monotonicity relation between the generalized bunching probabilities of partially labelled states.
\begin{restatable}{theorem}{refinementmonotonicity}
  Denote $m = \mathrm{dim}(\mathcal{H}_V), L = \mathrm{dim}(\mathcal{H}_H)$, let $\bm{i}$ be a list of $n$ distinct indices from $[m]$, and let $\lambda, \mu \vdash n$ such that $\mathrm{len}(\lambda) \le L$ and $\mathrm{len}(\mu) \le L$. Let $U \in \mathrm{U}(\mathcal{H}_V)$, and let $S \subseteq [m]$.
  Suppose that Conj.~\ref{conj:weakgenbunch} holds. 
  Then,
  \begin{align}
    \lambda \trianglerighteq \mu \implies b(S|U, \rho_{\bm{i}, \lambda}) \ge b(S|U, \rho_{\bm{i}, \mu}).
  \end{align}
  \label{thm:refinementmonotonicity}
\end{restatable}
This theorem is proven in App.~\ref{sec:proofofpartiallylabelled}.
It is the statement that, if we take a partially labelled state, then add more labels to it, the generalized bunching probability cannot increase.
Refinement is only one of many ways to ``add distinguishability'' to a state.
In the following section, we study a different partial order of  distinguishability.

\section{Schur convexity of mean generalized bunching}
\label{sec:schurconvexityofmeangenbunch}

In the previous section, we analyzed the problem of characterizing the distinguishability of particles when random permutations of the input modes are an available experimental control.
In this section, we consider the case that we can perform Haar random linear optical unitaries on the visible DOF.
Let $U\sim \mathrm{Haar}(\mathrm{U}(\mathcal{H}_V))$ denote that the unitary $U$ is distributed according to the Haar measure over $\mathrm{U}(\mathcal{H}_V)$.
Let $k \in [m]$, and let $S \subseteq [m]$ be a set of size $k$.
Then define the mean generalized bunching probability into a subset of size $k$,
\begin{align}
  w(k|\rho) = \mathbb{E}_{U \sim \mathrm{Haar}(\mathrm{U}(\mathcal{H}_V))}\left(b(S|U, \rho)\right).
  \label{eq:meangenbunchdef}
\end{align}
Due to the Haar averaging, the left-hand side depends only on $k$, and not on $S$.
In App.~\ref{sec:schurconvexity}, we prove Thm.~\ref{thm:meangenbunchsuremajorization}, which states that the mean generalized bunching probability is monotonic with respect to a certain preorder (Def.~\ref{def:suremajorization}) on the set of all $n$-particle states $\mathcal{D}(\mathrm{Sym}^n(\mathcal{H}_V \otimes \mathcal{H}_H))$.
Because the definition of this preorder is technical, we choose to discuss consequences of Thm.~\ref{thm:meangenbunchsuremajorization} in this section, and refer the interested reader to App.~\ref{sec:schurconvexity} for further detail.
As a part of the proof of Thm.~\ref{thm:meangenbunchsuremajorization}, we derive an explicit expression for the mean generalized bunching probability, which is given in Lem.~\ref{lem:meangenbunchgeneral}.

First we show that a natural generalized bunching conjecture is true on average over visible unitaries. 
\begin{restatable}{corollary}{avglieb}
  Let $\bm{i}$ be a list of $n$ distinct sites from $[m]$, and let $\phi_{\bm{i}}$ be the perfectly indistinguishable state defined in Eq.~\ref{eq:perfectlyindsitinguishable}.
  Let $k \in [m]$, and let $\rho \in \mathcal{D}(\mathrm{Sym}^n(\mathcal{H}_V \otimes \mathcal{H}_H))$ be any $n$-particle state.
  Then,
  \begin{align}
    w(k|\phi_{\bm{i}}) \ge w(k|\rho).
  \end{align}
  \label{cor:avglieb}
\end{restatable}
The proof of this corollary is given in App.~\ref{sec:schurconvexity}, immediately following Thm.~\ref{thm:meangenbunchsuremajorization}.

We now show two more consequences (Cor.~\ref{cor:schurconvexity}, Cor.~\ref{cor:thermcor}) of Thm.~\ref{thm:meangenbunchsuremajorization}, which provide a method to perform thermometry of cold-atom systems such as in Ref.~\citenum{youngAtomicBosonSampler2024}, where the atoms are in thermal equilibrium.
Such systems are described by singly occupied states where each particle's hidden density matrix is the same; we call these uniform states, as defined next.
\begin{definition}[Uniform states]
  Let $\rho^{(1)} \in \mathcal{D}(\mathcal{H}_H)$ be a single-particle hidden density matrix.
  Without loss of generality, write $\rho^{(1)} = \sum_k \alpha_k \ketbra{k}{k}$ for some basis $\left\{ \ket{k} \right\}_{k=1}^L$ and some probability distribution $\alpha \in \Delta^{L-1}$.
Define $\alpha(\bm{j}) = \alpha_{j_1}\cdots \alpha_{j_n}$.  
Let $\bm{i}$ be a list of $n$ distinct indices from $[m]$. 
Then the uniform state on $\bm{i}$ with single-particle density matrix $\rho^{(1)}$ is
\begin{align}
  \rho(\alpha) &= \sum_{\bm{j}}\alpha(\bm{j})a^{\dagger}(\bm{i}, \bm{j})\ketbra{0}{0}a(\bm{i}, \bm{j}).
  \label{eq:alphastate}
\end{align}
\label{def:uniformstate}
\end{definition}
The uniform state is the state where, for each $x \in [n]$, we independently populate the site $i_x$ with a boson of label $j$ with probability $\alpha_j$.
For uniform states, majorization provides a way to compare distinguishability.
\begin{definition}[Majorization]
For a distribution $\alpha \in \Delta^{N-1}$, write $\alpha^{\downarrow}$ for the distribution sorted in nonincreasing order, so $\alpha^\downarrow_{\; 1} \ge \cdots \ge \alpha^\downarrow_{\; N}$.
For distributions $\alpha, \alpha' \in \Delta^{N-1}$, we write $\alpha \succ \alpha'$ if for all $k \in [N]$,
\begin{align}
  \sum_{i=1}^k \alpha_{\;i}^{\downarrow} \ge \sum_{i=1}^k (\alpha')^\downarrow_{\;i},
\end{align}
and say that $\alpha$ majorizes $\alpha'$.
\label{def:majorization}
\end{definition}

We note that $\alpha \succ\alpha'$ implies that the entropy of $\alpha'$ is larger than that of $\alpha$.
Consequently, for uniform states, the probability that particles have the same label decreases as we go down in this order.
A function $f:\Delta^{N-1}\rightarrow \mathbb{R}$, that obeys $\alpha \succ \alpha' \implies f(\alpha) \ge f(\alpha')$, is said to be Schur convex.
  One might conjecture that, for all $U\in \mathrm{U}(\mathcal{H}_V)$ and all $S\subseteq [m]$, the generalized bunching probability of uniform states is Schur convex, so $\alpha\succ\alpha' \implies b(S|U,\rho(\alpha)) \ge b(S|U, \rho(\alpha'))$.
We were not able to prove this, even assuming the weak generalized-bunching conjecture.
Instead, we show that this implication holds on average, unconditionally.
\begin{restatable}{corollary}{schurconvexity}
  Let $n \le m = \mathrm{dim}(\mathcal{H}_V)$ and $n \le L = \mathrm{dim}(\mathcal{H}_H)$. Let $k \leq m$.  Let $\alpha, \alpha' \in \Delta^{L-1}$ be such that $\alpha \succ \alpha'$. Then 
  \begin{align}
    w(k|\rho(\alpha)) \ge w(k|\rho(\alpha')).
  \end{align}
  \label{cor:schurconvexity}
\end{restatable}
This corollary of Thm.~\ref{thm:meangenbunchsuremajorization} is proven in App.~\ref{sec:schurconvexity}, following Lem.~\ref{lem:uniformstateirrepdistribution}.
As an application of this corollary, we show that when the single-particle state is a Gibbs state, the mean generalized bunching probability is decreasing as a function of temperature.
\begin{restatable}{corollary}{thermcor}
  We take the same hypotheses as in Cor.~\ref{cor:schurconvexity}.
  Let $0 = \epsilon_0 \le \cdots \le \epsilon_{L-1}$ be chosen so that at least two of the $\epsilon_i$ differ.
  Let $\alpha_\beta$ be the Gibbs distribution,
  \begin{align}
    (\alpha_\beta)_k = \frac{e^{-\epsilon_k\beta}}{\sum_{j=0}^{L-1}e^{-\epsilon_j\beta}}.
  \end{align}
  Then $w(k|\rho(\alpha_\beta))$ is strictly monotonically increasing as a function of $\beta$.
  \label{cor:thermcor}
\end{restatable}
This corollary is proven in App.~\ref{sec:schurconvexity}, following Lem.~\ref{lem:gibbsmaj}.
Therefore, if we measure the mean generalized bunching probability, we can infer $\beta$ by inverting the relationship derived in Lem.~\ref{lem:meangenbunchgeneral}.
In this sense, the mean generalized bunching probability serves as a thermometer for uniform states with a Gibbs distribution.

\section{Discussion}
\label{sec:discussion}

We have shown various monotonicity properties of the generalized bunching probability.
Through these monotonicity properties, we argue that the generalized bunching probability provides useful partial information about the indistinguishability of many bosons.
We leave open the problem of whether the generalized bunching probability for uniform states is Schur convex for a given visible unitary, rather than averaged over all visible unitaries.
We also leave open the following problem: Let $k \le m$, let $\bm{i}$ be a list of $n$ distinct indices from $[m]$, let $\lambda, \mu \vdash n$, and let $L=\mathrm{dim}(\mathcal{H}_H)$ be such that $L \ge \mathrm{len}(\lambda), \mathrm{len}(\mu)$.
Does $\lambda \trianglerighteq \mu \implies w(k|\rho_{\bm{i}, \lambda})\ge w(k|\rho_{\bm{i}, \mu})$ (Def.~\ref{def:refinementofpartitions}, Def.~\ref{def:partiallylabelledstate}, Eq.~\ref{eq:meangenbunchdef}) follow from Thm.~\ref{thm:meangenbunchsuremajorization}?

In cold-atom systems, such as in Ref.~\cite{youngAtomicBosonSampler2024}, it is reasonable to assume that the state is a uniform state (Def.~\ref{def:uniformstate}), because the atoms have thermalized with a bath.
In this scenario, Cor.~\ref{cor:thermcor} provides a method of performing thermometry for these very low temperature systems, where it may be difficult to infer the temperature in other ways.
The experiment we propose is to (1) prepare a uniform state of thermal atoms in distinct sites, (2) perform a Haar-random linear optical unitary, and (3) measure the generalized bunching probability for a subset $S$ of sites that has size $k$.
Lem.~\ref{lem:meangenbunchgeneral} and Lem.~\ref{lem:uniformstateirrepdistribution} then provide a conversion of this measured generalized bunching probability to a measurement of temperature.
There are a few obstructions to implementing this directly.
The first is that in Ref.~\cite{youngAtomicBosonSampler2024}, the measurements are of parity of the number of atoms on the site, rather than of the atom number.
Secondly, the bosons in Ref.~\cite{youngAtomicBosonSampler2024} are slightly interacting, which may corrupt this measurement of temperature.
Lastly, measuring the mean generalized bunching probability involves sampling over the Haar measure of visible unitaries, which may be difficult in practice.
A possible extension of this work is to calculate how many samples from the Haar measure are necessary to obtain a good estimate of the mean generalized bunching probability.
It is possible that this number of samples is equal to one, using a method we now describe.

For a state $\rho$ that singly occupies the distinct sites $\bm{i}\in [m]^n$, a visible unitary $U\in \mathrm{U}(\mathcal{H}_V)$, and a number $k \in [m]$ of visible sites, define $\overline{b}(k|U, \rho)$ to be the uniform average of the generalized bunching probability over all subsets of $[m]$ that have $k$ elements.
By a simple combinatorial argument, given in Ref.~\citenum{youngAtomicBosonSampler2024} Eq.~17, there is a simple and efficient estimator for $\overline{b}(k|U, \rho)$ that can be calculated from samples of the visible occupation distribution.
The mean of $\overline{b}(k|U, \rho)$ over the Haar measure of visible unitaries is equal to the mean generalized bunching probability $w(k|\rho)$.
Then by the log-Sobolev inequality and Herbst's argument (Ref.~\citenum{aubrun2017alice}, Thm.~5.39), $\overline{b}(k|U, \rho)$ approximates $w(k|\rho)$ with high probability over the Haar measure of visible unitaries.
The estimate of the mean generalized bunching probability obtained from this method of averaging over subsets is subject to approximation error introduced from sampling the Haar measure. 
The bound on the approximation error obtained by direct application of the log-Sobolev inequality appears to be unfavorable when compared to the range of the mean generalized bunching probability over input states.
It is possible that the bounds can be significantly improved.
Alternatively, one can seek to determine a lower bound on the number of samples of $\overline{b}(k|U,\rho)$ over the Haar distribution of visible unitaries that are necessary to obtain a good estimate of the mean generalized bunching probability.
For a typical unitary $U$, the bunching probability \(b(S|U,\rho)\) is not constant as a function over subsets \(S\) of size \(k\).
This implies that the distribution over \(U\) of \(\overline{b}(k|U,\rho)\) is narrower than the distribution over \(U\) of \(b(S|U, \rho)\) for fixed \(S\). 
As a result it is possible that tighter concentration inequalities can be obtained for $\overline{b}(k|U, \rho)$ than those given in Ref.~\citenum{shchesnovich_distinguishing_2021} for $b(S|U, \rho)$.

\section*{Acknowledgements}
  This work includes contributions of the National Institute of
  Standards and Technology, which are not subject to U.S. copyright.
  The use of trade, product and software names is for informational
  purposes only and does not imply endorsement or recommendation by
  the U.S. government.
  S.G. acknowledges support from the Professional Research Experience Program (PREP) operated
jointly by NIST and the University of Colorado.
Parts of this work appear in the Ph.D. thesis of S.G.

We thank John Wilson, Aaron Friedman, Aaron Young, and Michael Grayson for comments on the manuscript.
We thank Scott Glancy for discussions about boson bunching.
We thank Aaron Young and Adam Kaufman for the motivation for this work.

\bibliography{thermometry}
\newpage
\begin{table}[h!]
  \caption{Glossary of symbols}
  \begin{tabular}{cc}
    Symbol & Definition\\
    \hline
    $b(S|U, \rho)$ & Generalized bunching probability (Eq.~\ref{eq:genbunchfirstexpression})\\
    $n$ & Number of particles \\
    $\mathcal{H}_V$ & Visible Hilbert space\\
    $\mathcal{H}_H$ & Hidden Hilbert space\\
    $m$ & $\mathrm{dim}(\mathcal{H}_V)$, number of visible modes\\
    $L$ & $\mathrm{dim}(\mathcal{H}_H)$, number of hidden modes\\
    $\mathrm{Part}_{n, d}$ & Partitions of $n$ with length at most $d$ (Sec.~\ref{sec:appendixprelim}.\ref{sec:partitionsyoungdia})\\
    $P(\sigma)$ & Permutation representation (Eq.~\ref{eq:permutationtensorfactoraction})\\
    $(\pi_\lambda, \mathrm{V}_\lambda^d)$, $(\pi_\lambda, \mathbb{S}_\lambda(\mathcal{H}))$ & $\lambda$-irrep of the unitary group $\mathrm{U}(\mathcal{H})$ (Eq.~\ref{eq:pilambdadef})\\
    $\mathrm{dim}(\mathrm{V}_\lambda^d)$ & Number of semistandard tableaux (Eq.~\ref{eq:nsst})\\
    $(\kappa_\lambda, \mathrm{Sp}_\lambda)$ & $\lambda$-irrep of the symmetric group\\
    $\mathrm{dim}(\lambda)$ & Number of standard tableaux (Def.~\ref{def:hooklength})\\
    $\omega_S^n$ & Set of occupations (Eq.~\ref{eq:occupations})\\
    $\omega_S^n(T)$ & Occupations that are nonzero on $T$ (Eq.~\ref{eq:occupationsnonzero})\\
    $K_{\bm{i}}$ & Singly occupied subspace (Eq.~\ref{eq:singlyoccupiesdef})\\
    $f_1$ & First quantization function (Lem.~\ref{lem:firstquantization})\\
    $h(\rho)$ & Auxiliary state (Def.~\ref{def:auxiliarystate})\\
    $F_{\bm{i}}(\sigma)$ & Permutation matrix (Def.~\ref{def:frepdef})\\
    $q_\lambda$ & Auxiliary irrep distribution (Cor.~\ref{cor:permutationinvariantauxiliarystates})\\
    $G(S|U, \bm{i})$ & Gram matrix (Eq.~\ref{eq:grammatrixdef})\\
      $c(\square)$ & Content of a box (Def.~\ref{def:contentofabox})\\
      $\chi_\lambda(\sigma)$ & Irreducible character of a permutation (Eq.~\ref{eq:irrechdef})\\
      $\overline{\mathrm{Imm}}_\lambda$ & Normalized $\lambda$-immanant (Def.~\ref{def:immanants})\\
      $\phi_{\bm{i}}$ & Perfectly indistinguishable state (Eq.~\ref{eq:perfectlyindsitinguishable})\\
      $\phi_{\bm{i}}^\lambda$ & Permutation-invariant state supported only on $\mathcal{H}_\lambda$ (Prop.~\ref{prop:deltastate})\\
      $\Tr_H(\rho)$ & Visible state (Def.~\ref{eq:visiblepart})\\
      $\rho_{\bm{i}, \mu}$ & Partially labelled state (Def.~\ref{def:partiallylabelledstate})\\
      $s_\lambda(\alpha)$ & Schur polynomial (Def.~\ref{def:schurpolynomial})\\
      $p_\lambda(\alpha)$ & Power-sum symmetric polynomial (Def.~\ref{def:powersum})\\
      $[\lambda]$ & Set of boxes of $\lambda$ (Sec.~\ref{sec:appendixprelim}.\ref{sec:partitionsyoungdia})\\
      $\mathrm{len}(\lambda)$ & Length of a partition (Sec.~\ref{sec:appendixprelim}.\ref{sec:partitionsyoungdia})\\
      $k^{\uparrow \lambda}$ & Rising factorial (Eq.~\ref{eq:risingfactorial})\\
      $\mathrm{SW}^n(\alpha)(\lambda)$ & Schur-Weyl distribution (Def.~\ref{def:swdistn})\\
      $\mathbb{S}_\lambda(\mathcal{H})_{\mathrm{wt} = w}$ & Weight space (Def.~\ref{def:weightspace})\\
      $\xi$ & Conversion from locations to occupations (Eq.~\ref{eq:xidef})\\
      $\zeta$ & Conversion from occupations to locations (Eq.~\ref{eq:zetadef})\\
      $\alpha\succ \alpha'$, $\lambda\succ\mu$ & Majorization (Def.~\ref{def:majorization}, Def.~\ref{def:majorizationforpartitions})\\
      $\lambda\trianglerighteq \mu$ & Refinement (Def.~\ref{def:refinementofpartitions})\\
  \end{tabular}
  \label{tab:glossary}
\end{table}
\renewcommand{\thesubsection}{\thesection.\arabic{subsection}}
\renewcommand{\thesubsubsection}{\thesubsection.\arabic{subsubsection}}
\renewcommand{\theequation}{\thesection.\arabic{equation}}
\section*{Appendixes}

A glossary of symbols is provided in Tab.~\ref{tab:glossary}.
The appendix is structured as follows.
In App.~\ref{sec:appendixprelim} we present preliminaries.
In App.~\ref{sec:derivationoffirstquantized} we present a derivation of the lemma that allows us to calculate the generalized bunching probability for different scenarios.
In App.~\ref{sec:expressingintermsofimmanants} we show that, for permutation-invariant states, the generalized bunching probability is equal to a convex combination of immanants of a certain Gram matrix.
In App.~\ref{sec:weakgenbunchproof} we prove Thm.~\ref{thm:weakgenbunchequiv}.
In App.~\ref{sec:proofofpartiallylabelled} we prove Thm.~\ref{thm:refinementmonotonicity}.
In App.~\ref{sec:moreprelims} we introduce further preliminaries necessary for proving Cor.~\ref{cor:schurconvexity}.
In App.~\ref{sec:derivationintermsofvisstate} we express the generalized bunching probability in terms of the visible state.
In App.~\ref{sec:derivationforperminv} we derive the visible state for permutation-invariant states.
In App.~\ref{sec:schurconvexity} we prove Cor.~\ref{cor:schurconvexity} and Cor.~\ref{cor:thermcor}.
In App.~\ref{sec:completeness} we show that the visible occupation distributions for partially labelled states spans the set of visible occupation distributions for all permutation-invariant states.
In App.~\ref{sec:commentsonentanglement} we offer some comments on the auxiliary state space, and the role of entanglement in our work.

\appendix
\section{Preliminaries}
\label{sec:appendixprelim}
\subsection{Conventions for occupations}
\begin{definition}
  For all $i \in [d]$, define $\xi_i:[d]^{n} \rightarrow \left\{ 0, \ldots, n \right\}$ by
  \begin{align}
    \xi_i(\bm{j}) &= \sum_{l} \delta_{i, j_l},
    \label{eq:xidef}
  \end{align} 
  \(\xi_{i}(\bm{j})\) counts the number of times that $i$ appears in $\bm{j}$.
The function \(i\mapsto \xi_i(\bm{j})\) is an occupation (Eq.~\ref{eq:occupations}). We define the function \(\xi\) from \([d]^{n}\) to the set of occupations \(\omega^{n}_{[d]}\) by
  \begin{align}
    \xi(\bm{j})(i) &= \xi_i(\bm{j}).
  \end{align}
\end{definition}
For an occupation $g \in \omega_S^n$, we write $g! = \prod_{s \in S}g(s)!$.

We need the following function $\zeta$ that switches from occupations (Eq.~\ref{eq:occupations}) to a list of locations,
\begin{align}
  \zeta:\omega^n_{[d]}\rightarrow [d]^{n}, \qquad  \zeta(g) = (\underbrace{1, \ldots, 1}_{g(1) \text{ times}}, \ldots, \underbrace{d, \ldots, d}_{g(d) \text{ times}}).
  \label{eq:zetadef}
\end{align} 
$\zeta$ is defined so that
\begin{align}
  \xi(\zeta(g)) = g.
\end{align}

\subsection{Probability theory}
If $\alpha$ is a probability measure, we write $X\sim \alpha$ to denote that the random variable $X$ is distributed according to $\alpha$.
We write $\mathbb{E}_{X \sim \alpha}(f(X))$ or $\mathbb{E}_{X}(f(X))$ to denote the expected value of the function $f$.

\subsection{Linear algebra}
Let $V, W$ be finite dimensional complex Hilbert spaces.
\begin{itemize}
  \item The algebra of linear maps on $V$ is written $\operatorname{End}(V)$.
  \item The trace of a linear operator $T$ on $V$ is written $\Tr(T)$.
  \item If $M\in \mathrm{End}(V \otimes W)$ is an operator on the tensor product of the Hilbert spaces $V$ and $W$, its partial trace over $V$ is denoted $\Tr_V(\rho)$.
  \item The identity operator on $V$ is denoted $\mathds{1}_V$.
  \item The support of a linear operator $T:V \rightarrow W$, written $\mathrm{supp}(T)$, is the orthogonal complement to its null space.
  \end{itemize}

\subsection{Partitions, Young diagrams, and Young tableaux}
\label{sec:partitionsyoungdia}
We mostly follow the conventions in~\cite{wright2016learn} and \cite{fultonYoungTableauxApplications1996}.
Given a partition $\lambda$, we write $|\lambda|$ for the sum of its entries, so that $\lambda \vdash |\lambda|$.
We write $\lambda! = \prod_{i \in [\mathrm{len}(\lambda)]}\lambda_i!$.
The set of partitions of $n$ is written $\mathrm{Part}_n$, and the set of partitions with length at most equal to $d$ is written $\mathrm{Part}_{n, d}$.

We make use of Young diagrams and tableaux.
For $\lambda$ a partition of $n$, a Young diagram of shape $\lambda$ is a left-justified array of $n$ boxes, with $\lambda_i$ boxes in the $i$th row.
For example, the Young diagram of shape $(4, 3, 1, 1)$ is
\begin{align}
  \ydiagram{4, 3, 1, 1}
\end{align}
\begin{itemize}
  \item We associate a partition to its Young diagram, and use the same symbol for the Young diagram as its partition. 
  \item For a partition $\lambda \vdash n$, we write $[\lambda]$ to mean the set of ``boxes'' of the Young diagram associated to $\lambda$. 
A box can be filled with a positive integer, or empty.
A generic box $\square$ is uniquely identified by its coordinates $\operatorname{coord}(\square) = (i, j)$, where $i$ is the row that $\square$ appears in $\lambda$, counting the rows from top to bottom, and $j$ is the column that it appears in $\lambda$, counting columns from left to right.
\end{itemize}
\begin{definition}
  The content of a box $\square \in [\lambda]$ with coordinates $\mathrm{coord}(\square) = (i, j)$ is defined to be $c(\square) = j - i$.
\label{def:contentofabox}
\end{definition}

A Young tableau, simply referred to as a tableau, is a Young diagram whose boxes are all filled with
positive integers. 
The plural of tableau is tableaux.  The Young diagram associated with a tableau is called its shape.  A semistandard tableau is a tableau such that, for each row, the entries of the row are nondecreasing and for each column, the entries of the column are strictly increasing.
For example, the tableau
\begin{align}
  \begin{ytableau}
    1 & 1 & 3 & 4\\
    2 & 3 & 4 \\
    4 \\
    5
  \end{ytableau}
\end{align}
is semistandard.

A tableau of shape $\lambda \vdash n$ is called standard if it is semistandard with distinct entries from $[n]$.
For example, the tableau
\begin{align}
  \begin{ytableau}
    1 & 4 & 5 & 7\\
    2 & 6 & 8 \\
    3 \\
    9
  \end{ytableau}
\end{align}
is standard.

\subsection{Representation theory}
See~\cite{goodmanSymmetryRepresentationsInvariants2009,fultonRepresentationTheory2004} for textbooks on representation theory.
Let $G$ be a group.
A representation $(\rho, V)$ of $G$ is a complex inner product space $V$ and a group homomorphism $\rho:G \rightarrow \mathrm{U}(V)$.
We will refer to $\rho$ or $V$ as being the representation when the other component is clear from context.
When the representation being used is clear from context, we write $g \cdot v$ to mean $\rho(g) v$.
The word ``irrep'' means irreducible representation.

For a representation $(\rho, R)$ of a group $G$, we write $(\mathrm{Res}^{G}_H(\rho), \mathrm{Res}^{G}_H(R))$ for the restriction of the representation to the subgroup $H \subseteq G$. Note that $\mathrm{Res}^{G}_H(R)$ is equal to $R$ as a vector space.

If $(\rho, V), (\pi, W)$ are representations of a group $G$, their direct sum is written $(\rho\oplus \pi, V \oplus W)$.
Similarly, their tensor product is written $(\rho\otimes \pi, V \otimes W)$.
If $(\rho, V)$ is isomorphic to $(\pi, W)$, we write $V \cong W$.

We review Schur's lemma.
\begin{lemma}[Schur (\cite{goodmanSymmetryRepresentationsInvariants2009} Lem. 4.1.4)]
Let $G$ be a group, and let $(\rho, V), (\pi, W)$ be finite dimensional complex irreps of $G$.
\begin{enumerate}
  \item Let $T:V\rightarrow W$ be a $G$-intertwiner.
    Then either $V \cong W$ and $T$ is an isomorphism, or $T$ is the zero map.
  \item Let $T: V \rightarrow V$ be a $G$-intertwiner.
    Then $T = \lambda \mathds{1}_V$ for some $\lambda \in \C$.
\end{enumerate}
\end{lemma}
We also list a few standard results that we need.
\begin{corollary}
  Let $G$ be a group, and $(\rho, V)$ be a semisimple finite dimensional complex representation of $G$, such that any irrep of $G$ has multiplicity equal to zero or one in $V$.
  If $T:V \rightarrow V$ is a $G$-intertwiner, it is equal to
  \begin{align}
    T &= \bigoplus_{s}\lambda_s \mathds{1}_s
  \end{align}
  for some $\lambda_s \in \C$, and the direct sum is over the irreducible subspaces of $V$.
  \label{cor:schurcor}
\end{corollary}

\begin{lemma}[\cite{fultonRepresentationTheory2004} Thm.~2.12]
  If $G$ is a finite group, and $\chi_1, \ldots, \chi_d$ are its distinct irreducible characters, then
\begin{align}
  \frac{1}{|G|}\sum_{g \in G}\chi_i(g)\chi_j(g)^* = \delta_{ij}.
\end{align}
\label{lem:irrechorthonormality}
\end{lemma}

\begin{definition}[\cite{goodmanSymmetryRepresentationsInvariants2009} Sec. 4.1.6]
  Let $G$ be a finite group, and let $(\rho, V)$ be a representation of $G$.
  Let $\mu$ be an equivalence class of representations of $G$.
  The $\mu$-isotype of $G$ in $V$ is the sum of subspaces
  \begin{align}
    V_{(\mu)} = \sum_{U \subset V|[U] = \mu}U,
  \end{align}
  where $[U]$ is the equivalence class of $U$.
\end{definition}
\begin{lemma}[\cite{fultonRepresentationTheory2004} Eq. 2.32]
Let $G$ be a finite group, and let $(\rho, V)$ be a representation of $G$.
Let $\mu$ be an irrep of $G$, and let $\chi_\mu = \mathrm{Tr}\circ \mu$ be the corresponding irreducible character.
Then 
\begin{align}
  \Theta_\mu := \frac{\chi_\mu(e)}{|G|}\sum_{g\in G}\chi_\mu(g)\rho(g^{-1})
\end{align}
is a projector onto the $\mu$-isotype of $G$ in $V$.
\label{lem:isotypicprojection}
\end{lemma}

\subsection{Representation theory of the symmetric group}

See~\cite{fultonRepresentationTheory2004,fultonRepresentationTheory2004,vershikNewApproachRepresentation2005,goodmanSymmetryRepresentationsInvariants2009,Etingof2011,Procesi2006-on,kondorGroupTheoreticalMethods2008,harrowApplicationsCoherentClassical2005,wright2016learn} for descriptions of the representation theory of the symmetric group.
The group of permutations of $n$ letters, also called the symmetric group on $n$ letters, is written $\mathcal{S}_n$.
Every permutation $\sigma \in \mathcal{S}_n$ can be written as a product of disjoint cycles, and furthermore this decomposition is unique up to reordering of the cycles.
The lengths of the cycles of $\sigma$ form a partition of $n$, known as the cycle type of $\sigma$.
Since each cycle is conjugate to every other of the same length, the cycle types enumerate the conjugacy classes of $\mathcal{S}_n$.
Since the number of conjugacy classes is equal to the number of irreps of a finite group, the irreps are enumerated by partitions as well.

The irreps of $\mathcal{S}_n$ are called the Specht modules, and they are indexed by partitions of $n$.
We use the Young orthogonal form of the Specht modules, which comes equipped with the Young-Yamanouchi basis, defined, for example, in Ref.~\citenum{vershikNewApproachRepresentation2005}.
This basis is an orthonormal basis, and the representation matrices are orthogonal.
We denote the irrep of shape $\lambda \vdash n$ by $(\kappa_\lambda, \mathrm{Sp}_\lambda)$.
The character of the irrep $\mathrm{Sp}_\lambda$ is written $\chi_\lambda:\mathcal{S}_n \rightarrow \C$, so
\begin{align}
  \chi_\lambda(\sigma) = \Tr(\kappa_\lambda(\sigma)).
  \label{eq:irrechdef}
\end{align}
The dimension of the Specht module $\mathrm{Sp}_\lambda$ is the number of standard tableaux of shape $\lambda$.
To state a formula for this number, we make use of the following definition.
\begin{definition}
  Let $\lambda \vdash n$, and let $\square \in [\lambda]$.
  Let $(i, j) = \mathrm{coord}(\square)$.
The hook length $h(\square)$ is defined to be $\lambda_i - j + (\lambda'_j - i) + 1$, where $\lambda'$ is the transpose of the Young diagram $\lambda$.
\end{definition}
\begin{definition}
  We denote the number of standard tableaux of shape $\lambda \vdash n$ by $\dim(\lambda)$. 
  It can be calculated from the hook length formula, which is
\begin{align}
  \dim(\lambda) &= \frac{n!}{\prod_{\square \in [\lambda]}h(\square)}.
\end{align}
\label{def:hooklength}
\end{definition}
The dimension of $\mathrm{Sp}_\lambda$ is $\mathrm{dim}(\lambda)$.

\subsection{Schur-Weyl duality}
Let $\mathcal{K} = \left( \C^d \right)^{\otimes n}$ be the $n$-fold tensor product of $d$ dimensional systems.
The unitary group $\mathrm{U}(d)$ acts on this Hilbert space under the tensor power action $\otimes^n:\mathrm{U}(d) \rightarrow \mathrm{U}(\mathcal{K})$
\begin{align}
  \otimes^n(U) \ket{\psi} &= U^{\otimes n}\ket{\psi}.
  \label{eq:tensorpoweraction}
\end{align}
The symmetric group $\mathcal{S}_n$ acts under the action $P:\mathcal{S}_n \rightarrow \mathrm{U}(\mathcal{K})$ that exchanges the tensor factors. 
That is, for computational basis states,
\begin{align}
  P(\sigma) \ket{i_1, \ldots, i_n} &= \ket{i_{\sigma^{-1}(1)}, \ldots, i_{\sigma^{-1}(n)}}.
  \label{eq:permutationtensorfactoraction}
\end{align}
The actions by unitaries commute with those by the permutations. Furthermore~\cite{Etingof2011}, writing $\langle S \rangle$ for the algebra generated by $S \subseteq \mathrm{End}(\mathcal{K})$, the algebras $\langle \{\otimes^n(U)| U \in \mathrm{U}(d)\}\rangle$ and $\langle \{P(\sigma) |\sigma \in \mathcal{S}_n \}\rangle$ centralize each other in $\mathrm{End}(\mathcal{K})$.
Accordingly, the space
\begin{align}
  \mathrm{V}_\lambda^d = \mathrm{Hom}_{\mathcal{S}_n}(\mathrm{Sp}_\lambda, (\mathbb{C}^d)^{\otimes n}) &= \left\{ f:\mathrm{Sp}_\lambda\rightarrow (\mathbb{C}^d)^{\otimes n}~\mathrm{linear}\;\middle| f(\sigma\cdot v) = \sigma \cdot f(v) \;\forall \sigma \in \mathcal{S}_n\right\}\
  \label{eq:pilambdadef}
\end{align}
is an irrep of $\mathrm{U}(d)$, where the action is defined as $(U\cdot f)(v) = U \cdot f(v)$.
We write $(\pi_\lambda, \mathrm{V}_\lambda^d)$ for this irrep of $\mathrm{U}(d)$.
A consequence of these results is Schur-Weyl duality (\cite{Etingof2011}, Cor. 4.59), which gives the following decomposition of $(\mathbb{C}^d)^{\otimes n}$ as an $\mathcal{S}_n \times \mathrm{U}(d)$ representation,
\begin{align}
  (\mathbb{C}^d)^{\otimes n} \cong \bigoplus_{\lambda\in\mathrm{Part}_{n, d}}\mathrm{Sp}_\lambda\otimes \mathrm{V}_\lambda^d,
  \label{eq:swdualitymodules}
\end{align}
where we recall the definition $\mathrm{Part}_{n, d} = \left\{ \lambda\vdash n|\mathrm{len}(\lambda) \le d \right\}$.
In this decomposition, the matrices representing permutations and unitaries are given by
\begin{align}
  P(\sigma) &= \bigoplus_{\lambda \in \mathrm{Part}_{n, d}}\kappa_\lambda(\sigma)\otimes \mathds{1}_{\mathrm{V}_\lambda^d},
  \label{eq:swpermdecomp}
\end{align}
and
\begin{align}
  U^{\otimes n} &= \bigoplus_{\lambda \in \mathrm{Part}_{n, d}}\mathds{1}_{\mathrm{Sp}_\lambda}\otimes \pi_\lambda(U).
\end{align}

For a complex inner product space $W$ of dimension $d$, we write $\mathbb{S}_\lambda(W)$ for the irrep of $\mathrm{U}(W)$ isomorphic to $\mathrm{V}_\lambda^d$. $\mathbb{S}_\lambda$ is called the Schur functor of shape $\lambda$.

\subsection{Physics background}

Let $\psi_1, \ldots, \psi_k, \phi_1, \ldots, \phi_k \in \mathcal{H}$ be states.
Let $P(\sigma)$ be as defined in Eq.~\ref{eq:permutationtensorfactoraction}.
Then we have
\begin{lemma}[Expectation of products of creation operators]
  \label{lem:productinnerproduct}
\begin{align}
  \bra{0}a_{\psi_1}\cdots a_{\psi_k}a_{\phi_k}^\dagger \cdots a_{\phi_1}^\dagger\ket{0} &= \braket{\psi_1 , \cdots , \psi_k|k!\Pi_{\mathrm{Sym}^k}|\phi_1 , \cdots , \phi_k},
  \label{eq:productinnerproduct}
\end{align}
where $\Pi_{\mathrm{Sym}^k}$ is the projector onto the symmetric subspace of $\mathcal{H}^{\otimes k}$, given by $\Pi_{\mathrm{Sym}^k} = \frac{1}{k!}\sum_{\sigma \in \mathcal{S}_k}P(\sigma)$.
\end{lemma}
See, for example,~\cite{geller2024characterization} Lem. 2.9.1 for a proof.
\begin{lemma}[First quantization]
  The map $f_1:\mathrm{Sym}^k(\mathcal{H})\rightarrow \mathcal{H}^{\otimes k}$
  \begin{align}
    f_1(a_{i_1}^\dagger\cdots a_{i_k}^\dagger\ket{0} ) &=  \sqrt{k!}\Pi_{\mathrm{Sym}^k}\ket{i_1, \ldots, i_k}\\
    &= \frac{1}{\sqrt{k!}}\sum_{\sigma \in \mathcal{S}_k}P(\sigma) \ket{i_1, \ldots, i_k}
    \label{eq:isometricintertwiner}
  \end{align}
  is an isometric intertwiner between the linear-optical action and the tensor power action.
  \label{lem:firstquantization}
\end{lemma}
See, for example, \cite{geller2024characterization}~Cor.~2.9.2 for a proof.
We also extend the first quantization map to density matrices, so $f_1(\ket{i}\bra{j}) = f_1(\ket{i})\otimes (f_1(\ket{j}))^\dagger$.
Then $f_1$ is an isometric intertwiner between the adjoint action on $\mathcal{B}_+(\mathrm{Sym}^k(\mathcal{H}))$ and that on $\mathcal{B}_+(\mathcal{H}^{\otimes k})$, with respect to the trace inner product.

We consider a linear-optics experiment where the bosons have a hidden DOF, so the single-particle Hilbert space $\mathcal{H}_1$ is a tensor product, $\mathcal{H}_1 = \mathcal{H}_V \otimes \mathcal{H}_H$.
We use the notation $m = \mathrm{dim}(\mathcal{H}_V)$ and $L = \mathrm{dim}(\mathcal{H}_H)$ throughout the appendices.

\section{Derivation of first quantized expression of generalized bunching probability}
\label{sec:derivationoffirstquantized}

Let $\mathcal{H}$ be a Hilbert space of dimension $T$.
For $g\in \omega_{[T]}^n$, let $\ket{\zeta(g)} = \ket{\zeta(g)_1}\otimes\cdots\otimes\ket{\zeta(g)_n} \in \mathcal{H}^{\otimes n}$ (Eq.~\ref{eq:zetadef}).
Let $\Pi_{\mathrm{Sym}^n}\in\mathcal{B}_+(\mathcal{H}^{\otimes n})$ be the projector onto the symmetric subspace of $\mathcal{H}^{\otimes n}$.
We make use of the definition of occupations $\omega_{[T]}^n(S)$ that are nonzero on a subset $S \subseteq [T]$ from Eq.~\ref{eq:occupationsnonzero}.
Define for $S \subseteq [T]$, define $\Pi_S = \sum_{s\in S}\ketbra{s}{s}$.
\begin{lemma}
Let $\mathcal{H}$ be a Hilbert space of dimension $T$.
  Let $S \subseteq [T]$. Then
  \begin{align}
    \Pi_{\mathrm{Sym}^n}\Pi_S^{\otimes n}\Pi_{\mathrm{Sym}^n} &= n!\sum_{g\in \omega_{[T]}^n(S)}\Pi_{\mathrm{Sym}^n}\frac{\ketbra{\zeta(g)}}{g!}\Pi_{\mathrm{Sym}^n}.
  \end{align}
  \label{lem:sumoversubset}
\end{lemma}
\begin{proof}
Observe that 
\begin{align}
P(\sigma)\Pi_{\mathrm{Sym}^n} =  P(\sigma)\frac{1}{n!}\sum_{\tau\in\mathcal{S}_n}P(\tau) =  \frac{1}{n!}\sum_{\tau\in\mathcal{S}_n}P(\sigma\tau) =  \frac{1}{n!}\sum_{\tau\in\mathcal{S}_n}P(\tau) =  \Pi_{\mathrm{Sym}^n}.
\end{align}
Therefore, 
\begin{align}
 n!\sum_{g\in \omega_{[T]}^n(S)}\Pi_{\mathrm{Sym}^n}\frac{\ketbra{\zeta(g)}}{g!}\Pi_{\mathrm{Sym}^n}&= 
 \sum_{\sigma\in\mathcal{S}_n}\sum_{g\in \omega_{[T]}^n(S)}\Pi_{\mathrm{Sym}^n}P(\sigma)\frac{\ketbra{\zeta(g)}}{g!}P(\sigma^{-1})\Pi_{\mathrm{Sym}^n}.
\end{align}
Now observe that, for any element $\sigma$ in the stabilizer of $\zeta(g)$, $P(\sigma)\ket{\zeta(g)} = \ket{\zeta(g)}$.
This stabilizer has cardinality $g!$.
Observe also that $S^{n}$ is the union of the orbits of the $\zeta(g)$ for $g\in \omega_{[T]}^n(S)$ under the permutation group.
Therefore,
\begin{align}
   n!\sum_{g\in \omega_{[T]}^n(S)}\Pi_{\mathrm{Sym}^n}\frac{\ketbra{\zeta(g)}}{g!}\Pi_{\mathrm{Sym}^n}&=\Pi_{\mathrm{Sym}^n}\Pi_S^{\otimes n}\Pi_{\mathrm{Sym}^n}.
\end{align}
\end{proof}

According to Def.~\ref{def:kisubspacedef}, the singly occupied pure states are of the form \(\sum_{\bm{j}\in[L]^{n}} \psi_{\bm{j}}a(\bm{i},\bm{j})^{\dagger}\ket{0}\), where the states
  \(a(\bm{i},\bm{j})^{\dagger}\ket{0}\) are orthonormal.
  This suggests an isometry between \(K_{\bm{i}}\) and \(\mathcal{H}_{H}^{\otimes n}\) that maps \(a(\bm{i},\bm{j})^{\dagger}\ket{0}\) to \(\ket{\bm{j}}\).
  This isometry is formalized by the next definition.

\begin{definition}
  Let $\bm{i}$ be a list of $n$ distinct indices from $[m]$.
  We define the function $h:K_{\bm{i}}\rightarrow \mathcal{H}_H^{\otimes n}$ (Def.~\ref{def:kisubspacedef}), by 
  \begin{align}
    h(\ket{\psi}) &=  \sum_{\bm{j}\in [L]^{n}}\ket{\bm{j}}\bra{0}a(\bm{i}, \bm{j})\ket{\psi}
  \end{align}
for \(\ket{\psi}\in K_{\bm{i}}\).
We also overload notation and extend the function \(h\) to density matrices.
For density matrices $\rho \in\mathcal{D}(K_{\bm{i}})$, $h(\rho)$ is given by
  \begin{align}
    h(\rho) = \sum_{\bm{j}, \bm{j}' \in [L]^{n}}\Tr(\rho a^\dagger(\bm{i}, \bm{j})\ketbra{0}a(\bm{i}, \bm{j}'))\ketbra{\bm{j}}{\bm{j}'}.
  \end{align}
  The state $h(\rho)$ is called the auxiliary state of $\rho$.
  \label{def:auxiliarystate}
\end{definition}
\begin{proposition}
  The function $h$ is a linear isomorphism between $K_{\bm{i}}$ and $\mathcal{H}_H^{\otimes n}$.
  \label{prop:hliniso}
\end{proposition}
\begin{proof}
We claim that the inverse of $h$ is $h^{-1}:\mathcal{H}_H^{\otimes n}\rightarrow K_{\bm{i}}$, given by
  \begin{align}
    h^{-1}(\ket{\phi}) &= \sum_{\bm{j}}a^\dagger(\bm{i}, \bm{j})\ket{0}\braket{\bm{j}|\phi}.
  \end{align}
  To prove the claim, we calculate
  \begin{align}
    h^{-1}(h(a^{\dagger}(\bm{i}, \bm{k})\ket{0})) &= \sum_{\bm{j}}a^\dagger(\bm{i}, \bm{j})\ket{0}\bra{\bm{j}}
\sum_{\bm{j}'}\ket{\bm{j}'}\bra{0}a(\bm{i}, \bm{j}') a^{\dagger}(\bm{i}, \bm{k})\ket{0}\\
&= \sum_{\bm{j},\bm{j}'}\delta_{\bm{j},\bm{j}'}\delta_{\bm{j}',\bm{k}}a^{\dagger}(\bm{i},\bm{j})\ket{0}\\
&= a^\dagger(\bm{i},\bm{k})\ket{0},
  \end{align}
  so $h^{-1}$ is a left inverse.
  To show it is a right inverse, we calculate
  \begin{align}
    h(h^{-1}(\ket{\bm{k}})) &= \sum_{\bm{j}'}\ket{\bm{j}'}\bra{0}a(\bm{i},\bm{j}')\sum_{\bm{j}}a^{\dagger}(\bm{i},\bm{j})\ket{0}\braket{\bm{j}|\bm{k}}\\
    &= \sum_{\bm{j},\bm{j}'}\delta_{\bm{j},\bm{j}'}\delta_{\bm{j},\bm{k}}\ket{\bm{j}'}\\
    &= \ket{\bm{k}}.
  \end{align}
\end{proof}
The next lemma provides an expression for the generalized bunching probability \(b(S|U,\rho)\) (Eq.~\ref{eq:genbunchfirstexpression}) when the initial state is in \(\mathcal{D}(K_{\bm{i}})\).
\begin{lemma}
  Let $\rho$ be a state that singly occupies the $n$ distinct sites $\bm{i}$~(Eq.~\ref{eq:singlyoccupiesdef}). 
  Then, with $h(\rho)$ the auxiliary state of $\rho$,
  \begin{align}
    b(S|U, \rho) &= \Tr\left( n! \Pi_{\mathrm{Sym}^n} \left(\Pi_S^{\otimes n}U^{\otimes n}\ketbra{\bm{i}}(U^\dagger)^{\otimes n}\Pi_S^{\otimes n}\otimes h(\rho)\right) \right).
  \end{align}
  \label{lem:genbunchfirstquantized}
\end{lemma}
\begin{proof}
  We start with Eq.~\ref{eq:genbunchfirstexpression},
  \begin{align}
    b(S|U, \rho) &= \sum_{g\in \omega_{[m]\times[L]}^n(S\times [L])}\Tr\left( \left((U \otimes \mathds{1})\cdot \rho \cdot (U^\dagger \otimes \mathds{1})\right)\ketbra{g}{g}\right).
  \label{eq:beginninggenbunchcalc}
  \end{align}
  We first overload notation so that
  \begin{align}
    \zeta(g) &= \Bigg( \underbrace{(1, 1), \ldots, (1, 1)}_{g(1, 1) \text{ times }},\ldots, \underbrace{(m, L), \ldots, (m, L)}_{g(m, L) \text{ times }}\Bigg)
  \end{align}
  where the ordering of the elements of $[m]\times [L]$ is arbitrary but fixed.
  We also overload
  \begin{align}
    a^\dagger(k_1, \ldots, k_n) = a^\dagger_{k_1}\cdots a^\dagger_{k_n}.
  \end{align}
  Then, we can express Eq.~\ref{eq:beginninggenbunchcalc} as
  \begin{align}
    b(S|U, \rho) &= \sum_{g\in \omega_{[m]\times[L]}^n(S\times [L])}\frac{1}{g!}\Tr\left( \left((U \otimes \mathds{1})\cdot \rho \cdot (U^\dagger \otimes \mathds{1})\right)a^\dagger(\zeta(g))\ketbra{0}{0}a(\zeta(g))\right).
  \end{align}
  By Lem.~\ref{lem:firstquantization}, this is equal to
  \begin{align}
  b(S|U, \rho)  &= \sum_{ g\in \omega_{[m]\times[L]}^n(S\times [L])}\frac{1}{g!}\Tr(n! (U \otimes \mathds{1})^{\otimes n}f_1(\rho)(U^\dagger \otimes \mathds{1})^{\otimes n}\Pi_{\operatorname{Sym}^n}\ketbra{\zeta(g)}\Pi_{\operatorname{Sym}^n})
  \end{align}
  By Lem.~\ref{lem:sumoversubset} applied to the set of modes $S \times [L]$, this is
  \begin{align}
b(S|U, \rho) &= \Tr((U \otimes \mathds{1})^{\otimes n}f_1(\rho) (U^\dagger \otimes \mathds{1})^{\otimes n}\Pi_{\operatorname{Sym}^n}(\Pi_S\otimes \mathds{1})^{\otimes n}\Pi_{\operatorname{Sym}^n})
\label{eq:middleoffirstquantizedcalc}
  \end{align}
  Since $\rho$ singly occupies the sites $\bm{i}$ (Eq.~\ref{eq:singlyoccupiesdef}), $\rho$ can be expressed as 
  \begin{align}
  \rho = \sum_{\bm{j}, \bm{j}'\in [L]^{n}}\rho_{\bm{j}, \bm{j}'}a^{\dagger}(\bm{i}, \bm{j})\ket{0}\!\!\bra{0}a(\bm{i}, \bm{j}').
  \end{align}
Then we calculate
  \begin{align}
    f_1(\rho) &= \sum_{\bm{j}, \bm{j}'\in [L]^{n}}\rho_{\bm{j}, \bm{j}'}f_1(a^{\dagger}(\bm{i}, \bm{j})\ket{0}\!\!\bra{0}a(\bm{i}, \bm{j}'))\\
    &= \sum_{\bm{j}, \bm{j}'\in [L]^{n}}\rho_{\bm{j}, \bm{j}'}n!\Pi_{\operatorname{Sym}^n}\ketbra{\bm{i},\bm{j}}{\bm{i},\bm{j}'}\Pi_{\operatorname{Sym}^n}.
  \end{align}
  Now we calculate the auxiliary state $h(\rho)$ (Def.~\ref{def:auxiliarystate}),
  \begin{align}
    h(\rho) &=  \sum_{\bm{k}, \bm{k}'}\Tr(\rho a^\dagger(\bm{i}, \bm{k})\ketbra{0}a(\bm{i}, \bm{k}'))\ketbra{\bm{k}}{\bm{k}'}\\
    &=  \sum_{\bm{k}, \bm{k}', \bm{j}, \bm{j}'}\rho_{\bm{j},\bm{j}'}\Tr( a^\dagger(\bm{i}, \bm{j})\ketbra{0}a(\bm{i}, \bm{j}') a^\dagger(\bm{i}, \bm{k})\ketbra{0}a(\bm{i}, \bm{k}'))\ketbra{\bm{k}}{\bm{k}'}\\
    &=  \sum_{\bm{k}, \bm{k}', \bm{j}, \bm{j}'}\rho_{\bm{j},\bm{j}'}\delta_{\bm{j},\bm{k}}\delta_{\bm{j}',\bm{k}'}\ketbra{\bm{k}}{\bm{k}'},
    \end{align}
    since $\bm{i}$ is a list of distinct indices.
    Doing the sums on $\bm{k}, \bm{k}'$ then yields
    \begin{align}
 h(\rho) &= \sum_{\bm{j}, \bm{j}'}\rho_{\bm{j},\bm{j}'}\ketbra{\bm{j}}{\bm{j}'}.
  \end{align}
  Therefore, after appropriate reordering of tensor factors, we have 
  \begin{align}
    f_1(\rho) &=  n!\Pi_{\operatorname{Sym}^n} (\ketbra{\bm{i}}\otimes h(\rho))\Pi_{\operatorname{Sym}^n}.
  \end{align}
  Therefore, we can plug into Eq.~\ref{eq:middleoffirstquantizedcalc} and use the cyclicity of trace to obtain
  \begin{align}
b(S|U, \rho) &= \Tr((U \otimes \mathds{1})^{\otimes n}n!\Pi_{\operatorname{Sym}^n} (\ketbra{\bm{i}}\otimes h(\rho))\Pi_{\operatorname{Sym}^n} (U^\dagger \otimes \mathds{1})^{\otimes n}\Pi_{\operatorname{Sym}^n}(\Pi_S\otimes \mathds{1})^{\otimes n}\Pi_{\operatorname{Sym}^n})\\
&= \Tr\left( n! \Pi_{\mathrm{Sym}^n} \left(U^{\otimes n}\ketbra{\bm{i}}(U^\dagger)^{\otimes n}\otimes h(\rho)\right) (\Pi_S\otimes \mathds{1})^{\otimes n}\right)\\
   &= \Tr\left( n! \Pi_{\mathrm{Sym}^n} \left(\Pi_S^{\otimes n}U^{\otimes n}\ketbra{\bm{i}}(U^\dagger)^{\otimes n}\Pi_S^{\otimes n}\otimes h(\rho)\right) \right).
  \end{align}
\end{proof}

  The permutation action on the sites is given by the representation $F_{\bm{i}}$ of $\mathcal{S}_n$ (Def.~\ref{def:frepdef}).
Loosely speaking, the subspace $K_{\bm{i}}$ defined in Def.~\ref{def:kisubspacedef} is special because we can permute the entries of $\bm{j}$ by permuting the modes $\bm{i}$.
  This notion is made precise in the following proposition.
  \begin{proposition}
    Let $\bm{i}$ be a list of $n$ distinct indices from $[m]$.
  Let $\sigma\in \mathcal{S}_n.$
  Let $\bm{j} \in [L]^{n}$.
  Then\begin{align}
(F_{\bm{i}}(\sigma)\otimes \mathds{1}_{\mathcal{H}_H})\cdot a^\dagger(\bm{i}, \bm{j})\ket{0} &= a^\dagger(\bm{i}, \sigma\cdot\bm{j})\ket{0}.
  \end{align}
  \label{prop:fintertwines}
\end{proposition}
\begin{proof}
The proposition is established by the following calculation:
  \begin{align}
    (F_{\bm{i}}(\sigma)\otimes \mathds{1}_{\mathcal{H}_H})\cdot a^\dagger(\bm{i}, \bm{j})\ket{0} &= a^\dagger_{i_{\sigma(1)}, j_1}\cdots a^\dagger_{i_{\sigma(n)}, j_n}\ket{0}\\
    &= a^\dagger_{i_1, j_{\sigma^{-1}(1)}}\cdots a^\dagger_{i_n, j_{\sigma^{-1}(n)}}\ket{0}\\
    &= a^\dagger(\bm{i}, \sigma\cdot\bm{j})\ket{0}.
  \end{align}
\end{proof}
  We define 
  \begin{align}
    \mathcal{S}_n(\bm{i}) = F_{\bm{i}}(\mathcal{S}_n)\subset \mathrm{U}(\mathcal{H}_V)
    \label{eq:snidef}
  \end{align}
   to be the image of the map $F_{\bm{i}}$ defined in Def.~\ref{def:frepdef}, which is the set of permutation matrices on $\mathcal{H}_V$ that act trivially on the orthogonal complement of the subspace spanned by $\left\{ \ket{i_x} \right\}_{x \in [n]}$.
  Observe that the restriction of the linear-optical action of $\mathrm{U}(\mathcal{H}_V \otimes \mathcal{H}_H)$ to the subgroup $\mathcal{S}_n(\bm{i})\times \mathrm{U}(\mathcal{H}_H)$ gives rise to an action of $\mathcal{S}_n\times \mathrm{U}(\mathcal{H}_H)$ on $K_{\bm{i}}$.
  We also have the action of $\mathcal{S}_n\times \mathrm{U}(\mathcal{H}_H)$ on $\mathcal{H}_H^{\otimes n}$ from Eqs.~\ref{eq:tensorpoweraction},\ref{eq:permutationtensorfactoraction}.
  With respect to these actions, we have
\begin{lemma}
  The function $h$ (Def.~\ref{def:auxiliarystate}) is an $\mathcal{S}_n\times \mathrm{U}(\mathcal{H}_H)$-isomorphism between $K_{\bm{i}}$ and $\mathcal{H}_H^{\otimes n}$.
  \label{lem:hintertwinesfandp}
\end{lemma}

\begin{proof}
Define the representation $G_{\bm{i}}:\mathcal{S}_n \rightarrow \mathrm{U}(K_{\bm{i}})$ given by
\begin{align}
  G_{\bm{i}}(\sigma)v &= (F_{\bm{i}}(\sigma)\otimes \mathds{1}_{\mathcal{H}_H})\cdot v,
\end{align}
where the action on the right-hand side is the linear-optical action.
  We claim that $h$ is an $\mathcal{S}_n$-intertwiner between the representations $G_{\bm{i}}$ and $P$ (Eq.~\ref{eq:permutationtensorfactoraction}).
  To verify the claim, let $\ket{\psi}\in K_{\bm{i}}$ and calculate 
  \begin{align}
    h( (F_{\bm{i}}(\sigma)\otimes \mathds{1}_{\mathcal{H}_H}) \cdot \ket{\psi}) &= \sum_{\bm{j}}\ket{\bm{j}}\bra{0}a(\bm{i}, \bm{j})(F_{\bm{i}}(\sigma)\otimes \mathds{1}_{\mathcal{H}_H})\cdot \ket{\psi}\\
    &= \sum_{\bm{j}}\ket{\bm{j}}\left((F_{\bm{i}}(\sigma)\otimes \mathds{1}_{\mathcal{H}_H})^\dagger\cdot a^\dagger(\bm{i}, \bm{j})\ket{0}\right)^\dagger \ket{\psi}\\
    &= \sum_{\bm{j}}\ket{\bm{j}}\bra{0}a(\bm{i}, \sigma^{-1}\cdot\bm{j})\ket{\psi}
    \end{align}
    by Prop.~\ref{prop:fintertwines}.  
    Continuing,
    \begin{align}
     h( (F_{\bm{i}}(\sigma)\otimes \mathds{1}_{\mathcal{H}_H}) \cdot \ket{\psi})   &= \sum_{\bm{j}}\ket{\sigma\cdot\bm{j}}\bra{0}a(\bm{i}, \bm{j})\ket{\psi}\\
    &= P(\sigma)\sum_{\bm{j}}\ket{\bm{j}}\bra{0}a(\bm{i}, \bm{j})\ket{\psi}\\
    &= P(\sigma)h(\ket{\psi}).
  \end{align}
  Therefore, $h$ intertwines $P$ and $G_{\bm{i}}$.

  Define the representation $J:\mathrm{U}(\mathcal{H}_H) \rightarrow \mathrm{U}(K_{\bm{i}})$ given by
\begin{align}
  J(U)v &= (\mathds{1}_{\mathcal{H}_V}\otimes U)\cdot v,
\end{align}
where the action on the right-hand side is the linear-optical action.
Define the representation $\otimes^n:\mathrm{U}(\mathcal{H}_H) \rightarrow \mathrm{U}(\mathcal{H}_H^{\otimes n})$ by $\otimes^n(U)v = U^{\otimes n}v$.
We claim that $h$ is a $\mathrm{U}(\mathcal{H}_H)$-intertwiner between the representations $J$ and $\otimes^n$.
  \label{prop:hintertwinesfandpunitary}
  Let $\bm{k}\in [L]^{n}$.
  Let $U \in \mathrm{U}(\mathcal{H}_H)$.
  If $\bm{k}' \in [L]^{n}$, let $U(\bm{k}'|\bm{k}) = U_{k'_1, k_1}\cdots U_{k'_n, k_n}$.
  Then, 
  \begin{align}
    h(J(U)a^{\dagger}(\bm{i}, \bm{k})\ket{0}) &= h((\mathds{1}_{\mathcal{H}_V}\otimes U)\cdot a^{\dagger}(\bm{i}, \bm{k})\ket{0})\\
    &= \sum_{\bm{j}}\ket{\bm{j}}\bra{0}a(\bm{i}, \bm{j}) (\mathds{1}_{\mathcal{H}_V}\otimes U)\cdot a^{\dagger}(\bm{i}, \bm{k})\ket{0}\\
    &= \sum_{\bm{j}}\ket{\bm{j}}\bra{0}a(\bm{i}, \bm{j}) \sum_{\bm{k}'}U(\bm{k}'|\bm{k})a^{\dagger}(\bm{i}, \bm{k}')\ket{0}\\
    &= \sum_{\bm{j}}\ket{\bm{j}} \sum_{\bm{k}'}U(\bm{k}'|\bm{k})\delta_{\bm{j},\bm{k}'}\\
    &= \sum_{\bm{j}}\ket{\bm{j}} U(\bm{j}|\bm{k})\\
    &= U^{\otimes n}\ket{\bm{k}}\\
    &= U^{\otimes n}h(a^{\dagger}(\bm{i}, \bm{k})\ket{0})
  \end{align}
  Now we construct the representation $G_{\bm{i}}\times J:\mathcal{S}_n\times \mathrm{U}(\mathcal{H}_H)\rightarrow U(K_{\bm{i}}), (G_{\bm{i}}\times J)(\sigma, U) = G_{\bm{i}}(\sigma)J(U)$.
  The actions defined by $G_{\bm{i}}$ and $J$ commute because the visible unitaries commute with the hidden unitaries.
  Define the representation $P \times \otimes^{n}:\mathcal{S}_n\times \mathrm{U}(\mathcal{H}_H)\rightarrow U(\mathcal{H}_H^{\otimes n})$ similarly.
  Therefore, $h$ is an $\mathcal{S}_n \times \mathrm{U}(\mathcal{H}_H)$-intertwiner between $G_{\bm{i}}\times J$ and $P \times \otimes^{n}$.
  By Prop.~\ref{prop:hliniso}, $h$ is a linear isomorphism, and therefore it is an isomorphism of representations.
\end{proof}
By applying Schur-Weyl duality (Eq.~\ref{eq:swdualitymodules}) we have
\begin{corollary}[\cite{goodmanSymmetryRepresentationsInvariants2009}, Thm.~9.2.8]
As an ${\mathcal{S}_n \times \mathrm{U}(\mathcal{H}_H)}$ representation,
\begin{align}
  K_{\bm{i}}
  &\cong \bigoplus_{\lambda\in\mathrm{Part}_{n, d}}\mathrm{Sp}_\lambda\otimes \mathbb{S}_\lambda(\mathcal{H}_H).
\end{align}
\label{cor:kiswdecomp}
\end{corollary}
\begin{corollary}
  Suppose $\rho$ singly occupies the $n$ distinct sites $\bm{i} \in [m]^n$, and is permutation invariant (Eq.~\ref{eq:perminvstatedecomposition}).
  Then there exists a probability distribution $q_\lambda$ over partitions of $n$, and states $\tau_\lambda \in \mathcal{D}(\mathbb{S}_\lambda(\mathcal{H}_H))$ such that
  \begin{align}
    h(\rho) &= \bigoplus_{\lambda\in\mathrm{Part}_{n, d}}q_\lambda \frac{\mathds{1}_{\mathrm{Sp}_\lambda}}{\mathrm{dim}(\lambda)}\otimes \tau_\lambda
    \label{eq:perminvstatedecomposition}
  \end{align}
  where the decomposition is with respect to the one induced by the $\mathcal{S}_n \times \mathrm{U}(\mathcal{H}_H)$-action on $\mathcal{H}_H^{\otimes n}$.
  The distribution $q_\lambda$ is called the auxiliary irrep distribution of $\rho$.
  \label{cor:permutationinvariantauxiliarystates}
\end{corollary}
\begin{proof}
  This is immediate from Lem.~\ref{lem:hintertwinesfandp} and Schur's lemma.
\end{proof}

\section{Expressing the generalized bunching probability in terms of immanants}
\label{sec:expressingintermsofimmanants}

Let $S \subseteq [m]$.
The Gram matrix associated with $S \subseteq [m]$, the unitary $U \in \mathrm{U}(\mathcal{H}_V)$ and the $n$ distinct sites $\bm{i}$ from $[m]$  is defined by
  \begin{align}
  G(S|U, \bm{i})_{xy} &= \sum_{s\in S}U_{s,i_x}^* U_{s,i_y}.
  \label{eq:grammatrixdef}
  \end{align}
  for $x, y \in [n]$.
\begin{lemma}
  Let $\rho$ be a permutation-invariant state that singly occupies the $n$ distinct sites $\bm{i} \in [m]^n$ (Def.~\ref{def:permutationinvariantstates}).
  Then, with $q_\lambda$ given by the auxiliary irrep distribution of $\rho$ (Cor.~\ref{cor:permutationinvariantauxiliarystates}),
  \begin{align}
    b(S|U, \rho) &= \sum_{\lambda\in\mathrm{Part}_{n, d}}q_\lambda \overline{\operatorname{Imm}}_\lambda(G(S|U, \bm{i})).
  \end{align}
  \label{lem:perminvgbisimm}
\end{lemma}
\begin{proof}
  We define the representation $R:\mathcal{S}_n \to \mathrm{U}(\mathcal{H}_H^{\otimes n})$ in analogy with the definition of $P$ (Eq.~\ref{eq:permutationtensorfactoraction}) by its action on computational basis vectors,
  \begin{align}
    R(\sigma)\ket{k_1, \ldots, k_n} &= \ket{k_{\sigma^{-1}(1)}, \ldots, k_{\sigma^{-1}(n)}}
  \end{align}
  By Lem.~\ref{lem:genbunchfirstquantized}, the generalized bunching probability is
  \begin{align}
    b(S|U, \rho) &= \Tr\left( n! \Pi_{\mathrm{Sym}^n} (\Pi_S^{\otimes n}U^{\otimes n}\ketbra{\bm{i}}(U^\dagger)^{\otimes n}\Pi_S^{\otimes n}\otimes h(\rho)) \right) \\
    &= \sum_{\sigma \in \mathcal{S}_n}\Tr\left(  (\Pi_S^{\otimes n}U^{\otimes n}P(\sigma)\ketbra{\bm{i}}(U^\dagger)^{\otimes n}\Pi_S^{\otimes n}\otimes R(\sigma) h(\rho)) \right)\\
    &= \sum_{\sigma \in \mathcal{S}_n}\bra{\bm{i}}(U^{\dagger})^{\otimes n}\Pi_S^{\otimes n}U^{\otimes n}P(\sigma)\ket{\bm{i}}\Tr(R(\sigma) h(\rho)).
    \end{align}
    By Schur-Weyl duality (Eq.~\ref{eq:swpermdecomp}), we have $R(\sigma) = \bigoplus_{\lambda\in\mathrm{Part}_{n, d}}\kappa_\lambda(\sigma)\otimes \mathds{1}_{\mathbb{S}_\lambda(\mathcal{H}_H)}$, so
    \begin{align}
      b(S|U, \rho)&= \sum_{\sigma \in \mathcal{S}_n}\bra{\bm{i}}(U^{\dagger})^{\otimes n}\Pi_S^{\otimes n}U^{\otimes n}P(\sigma)\ket{\bm{i}}\Tr(\bigoplus_{\lambda\in\mathrm{Part}_{n, d}}\kappa_\lambda(\sigma)\otimes \mathds{1}_{\mathbb{S}_\lambda(\mathcal{H}_H)} h(\rho)).
    \end{align}
    By Cor.~\ref{cor:permutationinvariantauxiliarystates}, this is
    \begin{align}
      b(S|U, \rho) &= \sum_{\sigma \in \mathcal{S}_n}\bra{\bm{i}}(U^{\dagger})^{\otimes n}\Pi_S^{\otimes n}U^{\otimes n}P(\sigma)\ket{\bm{i}}\Tr\left(\bigoplus_{\lambda\in\mathrm{Part}_{n, d}}\left(\kappa_\lambda(\sigma)\otimes \mathds{1}_{\mathbb{S}_\lambda(\mathcal{H}_H)}\right)\left( q_\lambda\frac{\mathds{1}_{\mathrm{Sp}_\lambda}}{\mathrm{dim}(\lambda)}\otimes \tau_\lambda\right)
 \right)\\
 &= \sum_{\sigma \in \mathcal{S}_n}\bra{\bm{i}}(U^{\dagger})^{\otimes n}\Pi_S^{\otimes n}U^{\otimes n}P(\sigma)\ket{\bm{i}}\sum_{\lambda\in\mathrm{Part}_{n, d}} q_\lambda \frac{\chi_\lambda(\sigma)}{\operatorname{dim}(\lambda)}\\
 &= \sum_{\lambda\in\mathrm{Part}_{n, d}}q_\lambda \overline{\operatorname{Imm}}_\lambda(G(S|U, \bm{i})),
  \end{align}
where \(G\) is the Gram matrix of Eq.~\ref{eq:grammatrixdef}.
\end{proof}

\begin{proposition}
Let $\bm{i}$ be a list of $n$ distinct indices from $[m]$, and $\lambda \vdash n$.
Assume $L \ge n$.
Then the state
\begin{align}
  \phi^{\lambda}_{\bm{i}} = h^{-1}\left(\frac{\mathds{1}_{\mathrm{Sp}_{\lambda}}}{\mathrm{dim}({\lambda})}\otimes \frac{\mathds{1}_{\mathbb{S}_{\lambda}(\mathcal{H}_H)}}{\mathrm{dim}(\mathbb{S}_{\lambda}(\mathcal{H}_H))}\right)
\end{align}
is a permutation-invariant state singly occupying the sites $\bm{i}$ whose auxiliary irrep distribution (Cor.~\ref{cor:permutationinvariantauxiliarystates}) is $q_\mu = \delta_{\mu\lambda}$.
\label{prop:deltastate}
\end{proposition}
\begin{proof}
  Since $h$ (Def.~\ref{def:auxiliarystate}) is an isomorphism, $\phi_{\bm{i}}^\lambda$ is indeed a permutation-invariant state singly occupying $\bm{i}$. 
  Since $h(\phi^\lambda_{\bm{i}})$ is supported only on the $\lambda$-isotype and $h^{-1}$ is an intertwiner, $\phi_{\bm{i}}^\lambda$ is supported only on the $\lambda$-isotype by Schur's lemma.
  This concludes the proof.
\end{proof}
\begin{corollary}
  We take the same hypotheses as in Prop.~\ref{prop:deltastate}.
  Let $S \subseteq [m]$.
  Then,
  \begin{align}
    b(S|U, \phi_{\bm{i}}^\lambda) = \overline{\mathrm{Imm}}_{\lambda}(G(S|U, \bm{i})).
  \end{align}
  \label{cor:philambdagenbunchimm}
\end{corollary}
\begin{proof}
  Apply Prop.~\ref{prop:deltastate} and Lem.~\ref{lem:perminvgbisimm}.
\end{proof}

We also need the generalized bunching probability of the perfectly indistinguishable state in our proof of Thm.~\ref{thm:weakgenbunchequiv}.
\begin{proposition}
  Let $\bm{i}$ be a list of $n$ distinct indices from $[m]$, and let $\phi_{\bm{i}}$ be the perfectly indistinguishable state, defined in Eq.~\ref{eq:perfectlyindsitinguishable}.
  The generalized bunching probability for $\phi_{\bm{i}}$ is
  \begin{align}
    b(S|U, \phi_{\bm{i}}) = \overline{\mathrm{Imm}}_{(n)}(G(S|U, \bm{i})).
  \end{align}
  \label{prop:perfectlyindistinguishablegenbunchprob}
\end{proposition}
We note that $\overline{\mathrm{Imm}}_{(n)}$ is the permanent.
\begin{proof}
  We calculate the auxiliary state $h(\phi_{\bm{i}})$ using Def.~\ref{def:auxiliarystate}.
  Write $1^n$ for the string containing $n$ many $1$s.
  Then,
  \begin{align}
    h(\phi_{\bm{i}}) &= \sum_{\bm{j}, \bm{j}'}\Tr(\phi_{\bm{i}} a^\dagger(\bm{i},\bm{j})\ketbra{0}a(\bm{i}, \bm{j}'))\ketbra{\bm{j}}{\bm{j}'}\\
    &= \sum_{\bm{j}, \bm{j}'}\Tr(a^\dagger(\bm{i}, (1^n))\ketbra{0}a(\bm{i}, (1^n)) a^\dagger(\bm{i},\bm{j})\ketbra{0}a(\bm{i}, \bm{j}'))\ketbra{\bm{j}}{\bm{j}'}\\
    &= \ketbra{1^n}.
  \end{align}
  We apply Lem.~\ref{lem:genbunchfirstquantized} to obtain
  \begin{align}
    b(S|U, \rho) &= \sum_{\sigma \in \mathcal{S}_n}\bra{\bm{i}}(U^{\dagger})^{\otimes n}\Pi_S^{\otimes n}U^{\otimes n}P(\sigma)\ket{\bm{i}}\Tr(R(\sigma) h(\rho))\\
    &= \sum_{\sigma \in \mathcal{S}_n}\bra{\bm{i}}(U^{\dagger})^{\otimes n}\Pi_S^{\otimes n}U^{\otimes n}P(\sigma)\ket{\bm{i}}\\
&= \overline{\mathrm{Imm}}_{(n)}(G(S|U, \bm{i})),
  \end{align}
  since $\chi_{(n)}(\sigma) = 1$ for all $\sigma \in \mathcal{S}_n$.
\end{proof}

\section{Proof of equivalence of conjectures}
\label{sec:weakgenbunchproof}

We show in this section that the weak generalized-bunching conjecture is equivalent to Lieb's permanental-dominance conjecture for immanants.
\weakgenequivlieb*
\begin{proof}
  First, we show that Lieb's conjecture implies the weak generalized-bunching conjecture. Let $\bm{i}$ be a list of $n$ distinct indices from $[m]$, and let $\rho$ be a permutation-invariant state (Def.~\ref{def:permutationinvariantstates}) that singly occupies the sites $\bm{i}$.
  Let $S \subseteq [m]$.
  Let $U \in \mathrm{U}(\mathcal{H}_V)$.
By Lem.~\ref{lem:perminvgbisimm}, we can express the generalized bunching probability in terms of the Gram matrix $G$ (Eq.~\ref{eq:grammatrixdef}) and the auxiliary irrep distribution $\{q_\lambda\}_\lambda$ of $\rho$ according to
  \begin{align}
    b(S|U, \rho) &= \sum_{\lambda\in\mathrm{Part}_{n, d}}q_\lambda \overline{\operatorname{Imm}}_\lambda(G(S|U, \bm{i})).
  \end{align}
  By Prop.~\ref{prop:perfectlyindistinguishablegenbunchprob}, we have $\overline{\mathrm{Imm}}_{(n)}(G(S|U, \bm{i})) = b(S|U, \phi_{\bm{i}})$, where $\phi_{\bm{i}}$ is defined in Eq.~\ref{eq:perfectlyindsitinguishable}.
  Since it is a Gram matrix, $G(S|U, \bm{i})$ is positive semidefinite, and therefore, by Lieb's conjecture, we have
  \begin{align}
    b(S|U, \rho) \le \overline{\mathrm{Imm}}_{(n)}(G(S|U, \bm{i})) = b(S|U, \phi_{\bm{i}}).
  \end{align}

To prove the converse, let $M$ be a nonzero $n \times n$ positive-semidefinite matrix.
 Write $N = M / \norm{M}$, where $\norm{M}$ is the spectral norm of $M$.
 This ensures that \(\mathds{1}-N\) is positive semidefinite.
 Construct the isometry
 \begin{align}
   V = \begin{pmatrix} \sqrt{N} \\ \sqrt{\mathds{1}-N}\end{pmatrix},
 \end{align}
then append an orthonormal basis of the nullspace of $V$ as columns to get a unitary matrix $U$.
Since the choices of \(\mathcal{H}_{H}\) and $\mathcal{H}_V$ are arbitrary in the statement of the weak generalized-bunching conjecture, we can choose $\mathrm{dim}(\mathcal{H}_H) = n$ so that $\phi_{\bm{i}}^\lambda$ (Prop.~\ref{prop:deltastate}) is well-defined, and we can choose $\mathrm{dim}(\mathcal{H}_V) = 2n$, so the unitary $U\in \mathrm{U}(\mathcal{H}_V)$.
Now let $S = \left\{ 1, \ldots, n \right\}$, $\bm{i} = \left( 1, \ldots, n \right)$, and calculate the Gram matrix (Eq.~\ref{eq:grammatrixdef})
\begin{align}
  G(S|U, \bm{i})_{x, y} &= \sum_{s \in S}U^*_{s, i_x}U_{s, i_y}\\
  &= \sum_{s \in [n]}(\sqrt{N})^*_{s, x}(\sqrt{N})_{s, y}\\
  &= N_{x, y}
\end{align}
By Cor.~\ref{cor:philambdagenbunchimm}, we have that $b(S|U, \phi_{\bm{i}}^\lambda) = \overline{\mathrm{Imm}}_\lambda(G(S|U, \bm{i}))$.
Then by the weak generalized-bunching conjecture applied to $U$ and the subset $S$, 
\begin{align}
   \norm{M}^n \frac{1}{\chi_\lambda(e)}\sum_{\sigma \in \mathcal{S}_n}\chi_\lambda(\sigma)\prod_{x=1}^n N_{x, \sigma(x)} = 
   \norm{M}^nb(S|U, \phi^{\lambda}_{\bm{i}}) &\le \norm{M}^n b(S|U, \phi_{\bm{i}}) = \norm{M}^n \overline{\operatorname{Imm}}_{(n)}(N)\\
   \implies \frac{1}{\chi_\lambda(e)}\sum_{\sigma \in \mathcal{S}_n}\chi_\lambda(\sigma)\prod_{x=1}^n M_{x, \sigma(x)} &\le \overline{\operatorname{Imm}}_{(n)}(M).
 \end{align}
 \end{proof}

\section{Partial order on partially labelled states}
\label{sec:proofofpartiallylabelled}

We show that if the weak generalized-bunching conjecture (Conj.~\ref{conj:weakgenbunch}) holds, then the generalized bunching probability for partially labelled states (Def.~\ref{def:partiallylabelledstate}) is monotonic with respect to the refinement partial order (Def.~\ref{def:refinementofpartitions}).

Write $\mathcal{Q}_\mu(q)$ for the set of ordered partitions of the set $q$ into sets with sizes given by $\mu\vdash |q|$.
Writing $\mathrm{len}(\mu) = l$, elements of $\mathcal{Q}_\mu(q)$ are maps $R:[l]\rightarrow \mathcal{P}(q)$ to the powerset of $q$ such that $i \neq j \implies R(i)\cap R(j) = \emptyset$,  $\bigcup_{i \in [n]}R(i) = q$, and $|R(i)| = \mu_i$.
When we write $\mathcal{Q}_\mu(\bm{i})$ for a list of indices $\bm{i}$, we mean $\mathcal{Q}_\mu(\mathrm{Set}(\bm{i}))$, where $\mathrm{Set}(\bm{i}) = \left\{ i_x|x\in[n] \right\}$ is the set of elements of $\bm{i}$.
For an ordered partition $R \in \mathcal{Q}_\mu(q)$, we write $r \in R$ to mean $r$ is an element of the image of $R$.

Recall the definition of the perfectly indistinguishable state $\phi_{\bm{i}}$ from Eq.~\ref{eq:perfectlyindsitinguishable}.
Since $\phi_{\bm{i}}$ does not depend on the ordering of the elements of the list of indices $\bm{i}$, we also define it for when $\bm{i}$ is a set of indices, rather than a list.
\begin{proposition}
  Let $\bm{i}$ be a list of $n$ distinct indices from $[m]$, and  $\mu \vdash n$ be a partition of $n$ such that $\mathrm{len}(\mu) = l\le L$.
  Let $S \subseteq [m]$, and let $U\in \mathrm{U}(\mathcal{H}_V)$.
  Then the generalized bunching probability for the partially labelled state
  \(\rho_{\bm{i},\mu}\) (Def.~\ref{def:partiallylabelledstate}) is
  \begin{align}
    b(S|U, \rho_{\bm{i},\mu}) = \frac{\mu!}{n!}\sum_{R \in \mathcal{Q}_\mu(\bm{i})}\prod_{r \in R}b(S|U, \phi_{r})
  \end{align}
  \label{prop:buchpartiallylabelled}
\end{proposition}
\begin{proof} 
  First we define
   \begin{align}
     \rho_{\bm{i}, \mu, \sigma} = a^\dagger(\sigma^{-1}\cdot\bm{i}, \overline{\mu})\ketbra{0}{0}a(\sigma^{-1}\cdot\bm{i},\overline{\mu}),
   \end{align}
   so that $\rho_{\bm{i}, \mu} = \frac{1}{n!}\sum_{\sigma \in \mathcal{S}_n} \rho_{\bm{i}, \mu, \sigma}$.
From Lem.~\ref{lem:genbunchfirstquantized} we have
  \begin{align}
    b(S|U, \rho_{\bm{i}, \mu}) &= \Tr\left( n!\Pi_{\mathrm{Sym}^n}\left( \Pi_S^{\otimes n}U^{\otimes n}\ketbra{\bm{i}}{\bm{i}}(U^\dagger)^{\otimes n}\Pi_S^{\otimes n}\otimes h(\rho_{\bm{i}, \mu}) \right) \right)\\
    &= \frac{1}{n!}\sum_{\sigma\in\mathcal{S}_n}\Tr\left( n!\Pi_{\mathrm{Sym}^n}\left( \Pi_S^{\otimes n}U^{\otimes n}\ketbra{\bm{i}}{\bm{i}}(U^\dagger)^{\otimes n}\Pi_S^{\otimes n}\otimes h(\rho_{\bm{i}, \mu, \sigma}) \right) \right)\\
    &= \frac{1}{n!}\sum_{\sigma\in\mathcal{S}_n}b(S|U, \rho_{\bm{i}, \mu,\sigma}),
    \label{eq:partiallylabelledgenbunchfirst}
  \end{align}
  since $h$ (Def.~\ref{def:auxiliarystate}) is a linear isomorphism by Prop.~\ref{prop:hliniso}.
  Then,
  \begin{align}
    h(\rho_{\bm{i}, \mu,\sigma^{-1}}) &= \sum_{\bm{j}, \bm{j}'}\Tr(\rho_{\bm{i}, \mu,\sigma} a^\dagger(\bm{i},\bm{j})\ketbra{0}a(\bm{i}, \bm{j}'))\ketbra{\bm{j}}{\bm{j}'}\\
    &= \sum_{\bm{j}, \bm{j}'}\Tr(a^\dagger(\sigma\cdot\bm{i}, \overline{\mu})\ketbra{0}{0}a(\sigma\cdot\bm{i},\overline{\mu})a^\dagger(\bm{i},\bm{j})\ketbra{0}a(\bm{i}, \bm{j}'))\ketbra{\bm{j}}{\bm{j}'}.
  \end{align}
    Define $\delta(i, j) = \delta_{i,j}$ to make the arguments easier to read.
    Continuing the calculation and making use of Lem.~\ref{lem:productinnerproduct},
    \begin{align}
  \sum_{\bm{j}, \bm{j}'}&\Tr(a^\dagger(\sigma\cdot\bm{i}, \overline{\mu})\ketbra{0}{0}a(\sigma\cdot\bm{i},\overline{\mu})a^\dagger(\bm{i},\bm{j})\ketbra{0}a(\bm{i}, \bm{j}'))\ketbra{\bm{j}}{\bm{j}'}\\
  &=\sum_{\bm{j}, \bm{j}'}\sum_{\tau,\tau'\in \mathcal{S}_n}\delta(\sigma\cdot\bm{i},\tau\cdot\bm{i})\delta(\overline{\mu},\tau\cdot\bm{j})\delta(\sigma\cdot\bm{i},\tau'\cdot\bm{i})\delta(\overline{\mu},\tau'\cdot\bm{j})\ketbra{\bm{j}}{\bm{j}'}\\
      &= \sum_{\bm{j}, \bm{j}'}\sum_{\tau,\tau'\in \mathcal{S}_n}\delta(\sigma,\tau)\delta(\tau^{-1}\cdot\overline{\mu},\bm{j})\delta(\sigma,\tau')\delta((\tau')^{-1}\cdot\overline{\mu},\bm{j})\ketbra{\bm{j}}{\bm{j}'}\\
      &= \ketbra{\sigma^{-1}\cdot\overline{\mu}}{\sigma^{-1}\cdot\overline{\mu}}.
    \end{align}
    Therefore, by Lem.~\ref{lem:genbunchfirstquantized},
  \begin{align}
    b(S|U, \rho_{\bm{i}, \mu, \sigma}) 
    &= \Tr\left( n!\Pi_{\mathrm{Sym}^n}\left( \Pi_S^{\otimes n}U^{\otimes n}\ketbra{\bm{i}}{\bm{i}}(U^\dagger)^{\otimes n}\Pi_S^{\otimes n}\otimes \ketbra{\sigma\cdot\overline{\mu}}{\sigma\cdot\overline{\mu}} \right) \right)\\
    &= \sum_{\tau\in \mathcal{S}_n}\Tr\left(  \Pi_S^{\otimes n}U^{\otimes n}\ketbra{\tau\cdot\bm{i}}{\bm{i}}(U^\dagger)^{\otimes n}\Pi_S^{\otimes n}\otimes \ketbra{\tau\sigma\cdot\overline{\mu}}{\sigma\cdot\overline{\mu}} \right)\\
    &= \sum_{\tau\in \mathcal{S}_n} \bra{\bm{i}}(U^\dagger)^{\otimes n}\Pi_S^{\otimes n}U^{\otimes n}\ket{\tau\cdot\bm{i}} \delta(\tau\sigma\cdot\overline{\mu},\sigma\cdot\overline{\mu})\\
    &= \sum_{\tau\in S(\sigma\cdot\overline{\mu})} \bra{\bm{i}}(U^\dagger)^{\otimes n}\Pi_S^{\otimes n}U^{\otimes n}\ket{\tau\cdot\bm{i}},
    \label{eq:introductionofstabilizer}
  \end{align}
  where $S(\sigma\cdot\overline{\mu}) = \left\{ \tau\in \mathcal{S}_n|\tau\cdot(\sigma \cdot \overline{\mu}) = \sigma \cdot \overline{\mu} \right\}$.  
  Define the ordered partition $R_\sigma\in \mathcal{Q}_\mu(\bm{i})$ by $R_\sigma(x) = \left\{ i_v|v \in [n] \text{ s.t. }(\sigma\cdot \overline{\mu})_v = x \right\}$.
  Observe that the elements of the stabilizer $S(\sigma\cdot \overline{\mu})$ are permutations of $[n]$ that preserve the ordered partition $R_\sigma$.
 That is, if $\tau \in S(\sigma\cdot \overline{\mu})$, then $\tau$ can be expressed as $\tau_{R_\sigma(1)}\cdots\tau_{R_\sigma(l)}$, where each $\tau_{R_\sigma(x)}$ acts nontrivially only on $R_\sigma(x)$. 
 Now define $\ket{\bm{i}^{(x)}}$ to be the tensor product of the factors of $\ket{\bm{i}}$ corresponding to $R_\sigma(x)$, and define $\ket{1^k} = \ket{1}\otimes \cdots\otimes \ket{1}\in \mathcal{H}_H^{\otimes k}$. This shows that the right-hand side of Eq.~\ref{eq:introductionofstabilizer} can be written as
  \begin{align}
    b(S|U, \rho_{\bm{i}, \mu, \sigma})   
    &= \prod_{x\in[l]}\sum_{\tau_x\in \mathcal{S}_{\mu_x}}\bra{\bm{i}^{(x)}}(U^\dagger)^{\otimes \mu_x}\Pi_S^{\otimes \mu_x}U^{\otimes \mu_x}\ket{\tau_x\cdot\bm{i}^{(x)}}\\
 &=\prod_{x\in[l]}\sum_{\tau_x\in \mathcal{S}_{\mu_x}}\bra{\bm{i}^{(x)}}(U^\dagger)^{\otimes \mu_x}\Pi_S^{\otimes \mu_x}U^{\otimes \mu_x}\ket{\tau_x\cdot\bm{i}^{(x)}}\braket{1^{\mu_x}|\tau_x\cdot 1^{\mu_x}}\\
 &=\prod_{x\in[l]}\Tr\left( \mu_x!\Pi_{\mathrm{Sym}^{\mu_x}}\left( \Pi_S^{\otimes \mu_x}U^{\otimes \mu_x}\ketbra{\bm{i}^{(x)}}{\bm{i}^{(x)}}(U^\dagger)^{\otimes \mu_x}\Pi_S^{\otimes \mu_x}\otimes \ketbra{1^{\mu_x}}{1^{\mu_x}} \right) \right).
\end{align}
We have that $\ketbra{1^{\mu_x}}{1^{\mu_x}} = h(\phi_{R_\sigma(x)})$, so
\begin{align}
  b(S|U, \rho_{\bm{i}, \mu, \sigma})   
  &=\prod_{r \in R_\sigma}b(S|U, \phi_{r}).
\end{align}
Let $S(\overline{\mu})$ be the stabilizer of $\overline{\mu}$ under the action of $\mathcal{S}_n$, and let $A$ be a coset transversal of $\mathcal{S}_n/S(\overline{\mu})$.
If \(\sigma'\) and \(\sigma\) are in the same coset of \(S(\overline{\mu})\), then $\rho_{\bm{i},\mu,\sigma} = \rho_{\bm{i},\mu,\sigma'}$.
Therefore, returning to Eq.~\ref{eq:partiallylabelledgenbunchfirst},
\begin{align}
  b(S|U, \rho_{\bm{i}, \mu}) 
  &= \frac{1}{n!}\sum_{\sigma\in\mathcal{S}_n} b(S|U, \rho_{\bm{i}, \mu, \sigma})\\
  &= \frac{\mu!}{n!}\sum_{\pi\in A} b(S|U, \rho_{\bm{i}, \mu, \pi})\\
  &= \frac{\mu!}{n!}\sum_{\pi\in A} \prod_{r \in R_\pi}b(S|U, \phi_r)
\end{align}
We claim that the map $f:A\to \mathcal{Q}_\mu(\bm{i})$ given by $f(\pi) = R_\pi$ is a bijection between $A$ and $\mathcal{Q}_\mu(\bm{i})$.
First we show that $f$ is injective.
Suppose $\pi, \pi'\in A$ are such that $R_\pi = R_{\pi'}$.
Then $\pi\cdot\overline{\mu} = \pi'\cdot\overline{\mu}$, so since $A$ is a coset transversal, $\pi = \pi'$.
To show $f$ is surjective, let $R \in \mathcal{Q}_\mu(\bm{i})$.
Then let $c_R$ to be the list such that, for all $x \in [l]$, $v \in R(x) \implies (c_R)_v = x$.
Then $c_R$ is in the orbit of $\overline{\mu}$.
By the orbit-stabilizer theorem, the orbit of $\overline{\mu}$ is in bijection with $A$, and therefore there exists $\pi \in A$ such that $\pi\cdot \overline{\mu} = c_R$.
Therefore,
\begin{align}
  b(S|U, \rho_{\bm{i}, \mu}) &= \frac{\mu!}{n!}\sum_{R\in\mathcal{Q}_\mu(\bm{i})}\prod_{r\in R}b(S|U, \phi_{r}).
\end{align}
\end{proof}

We now state Thm.~\ref{thm:refinementmonotonicity} again and prove it.
\refinementmonotonicity*
\begin{proof}
  According to Prop.~\ref{prop:buchpartiallylabelled},
  \begin{align}
    b(S|U, \rho_{\bm{i},\lambda}) &=  \frac{\lambda!}{n!}\sum_{R \in \mathcal{Q}_\lambda(\bm{i})}\prod_{r \in R}b(S|U, \phi_r).
  \end{align}
  Let $N = \mathrm{len}(\lambda)$.
  Since $\lambda \trianglerighteq \mu$ (Def.~\ref{def:refinementofpartitions}), there exist partitions $\nu^{(x)}\vdash \lambda_x$ for $x \in [N]$ such that $\mu = \sum_{x \in [N]}\nu^{(x)}$.
  Since $\rho_{R(x), \nu^{(x)}}$ is a permutation-invariant state, by the weak generalized-bunching conjecture we have
  \begin{align}
    b(S|U, \phi_{R(x)}) &\ge b(S|U, \rho_{R(x), \nu^{(x)}}) \\
    &= \frac{\nu^{(x)}!}{|R(x)|!}\sum_{Q_x \in \mathcal{Q}_{\nu^{(x)}}(R(x))}\prod_{q_x \in Q_x}b(S|U, \phi_{q_x}).
  \end{align}
  Therefore,
  \begin{align}
    b(S|U, \rho_{\bm{i},\lambda})&\ge\underbrace{\frac{\lambda!}{n!}\sum_{R \in \mathcal{Q}_\lambda(\bm{i})}\prod_{x \in [N]}\frac{\nu^{(x)}!}{|R(x)|!}\sum_{Q_x \in \mathcal{Q}_{\nu^{(x)}}(R(x))}\prod_{q_x \in Q_x}b(S|U, \phi_{q_x})}_X
 \label{eq:refinementmonotonicity_ineq}
  \end{align}
  Write $X$ for the right-hand side of the above display.
  Now note that $\prod_{x\in [N]}|R(x)|! = \lambda!$, and $\prod_{x\in [N]}\nu^{(x)}! = \mu!$, so we have
  \begin{align}
    X &= \frac{\mu!}{n!}\sum_{R \in \mathcal{Q}_\lambda(\bm{i})}\underbrace{\prod_{x \in [N]}\sum_{Q_x \in \mathcal{Q}_{\nu^{(x)}}(R(x))}}\prod_{q_x \in Q_x}b(S|U, \phi_{q_x}).
    \label{eq:refinementmonotonicityproofmiddle}
  \end{align}
    We apply the distributive law to the underbraced pair of product and sum.
    To describe the resulting sum of product, let $\mathcal{F}$ be the set of functions \(f\) with domain \([N]\) for which there exists an $R \in \mathcal{Q}_\lambda(\bm{i})$ such that, for all \(x\in[N]\), \(f(x)\in\mathcal{Q}_{\nu^{(x)}}(R(x))\).
    For \(x\in[N]\), the ordered partition \(f(x)\) is indexed by \([\mathrm{len}(\nu^{(x)})]\) and the set \(f(x)(y)\) is of size \(\nu^{(x)}_{y}\).
    Let \(P\) be the set of pairs \((x,y)\) with \(x\in[N]\) and \(y\in[\mathrm{len}(\nu^{(x)})]\).
    We can now express \(X\) as
    \begin{align}
      X&= \frac{\mu!}{n!} \sum_{f\in\mathcal{F}}\prod_{x\in[N]}\prod_{q\in f(x)}b(S|U,\phi_{q}) \\
      &=\frac{\mu!}{n!} \sum_{f\in\mathcal{F}}\prod_{x\in[N]}\prod_{y\in[\mathrm{len}(\nu^{(x)})]}b(S|U,\phi_{f(x)(y)}) \\
       &=\frac{\mu!}{n!} \sum_{f\in\mathcal{F}}\prod_{(x,y)\in P}b(S|U,\phi_{f(x)(y)}) .
    \end{align}
    Because \(\mu=\sum_{x\in[N]}\nu^{(x)}\), \(|P|=\mathrm{len}(\mu)\).
    Let \(\{(x_{i},y_{i})\}_{i\in[\mathrm{len}(\mu)]}\) be an indexing of \(P\) such that \(i \mapsto \nu^{(x_{i})}_{y_{i}}\) is monotonically nonincreasing.
    By construction, for each \(f\in\mathcal{F}\), the map \(i\in[\mathrm{len}(\mu)] \mapsto f(x_{i})(y_{i})\) is a distinct ordered partition in \(\mathcal{Q}_{\mu}(\bm{i})\). 
    Conversely, for \(Q\in\mathcal{Q}_{\mu}(\bm{i})\), we can define \(f\) by \(f(x_{i})(y_{i})=Q(i)\) to obtain a member of \(\mathcal{F}\).
    It follows that the sum over \(\mathcal{F}\) can be replaced with a sum over \(\mathcal{Q}_{\mu}(\bm{i})\):
    \begin{align}
    X&= \frac{\mu!}{n!}\sum_{R \in \mathcal{Q}_\mu(\bm{i})}\prod_{r \in R}b(S|U, \phi_r)\\
    &= b(S|U, \rho_{\bm{i},\mu}).
    \end{align}
    The theorem follows by substituting \(X\) for the right-hand side of
    Eq.~\eqref{eq:refinementmonotonicity_ineq}.
\end{proof}

\section{More preliminaries}
\label{sec:moreprelims}
We need the following fact for the proof of Cor.~\ref{cor:schurconvexity}.
\begin{fact}
  Suppose $\lambda \vdash n$ has a box $\square$ with content $r < 0$.
Then for each $t \in \left\{ r, \ldots, 0 \right\}$, there exists some box $b_t$ in $\lambda$ such that $c(b_t) = t$. 
\label{fact:allcontents}
\end{fact}
\begin{proof}
Suppose the coordinates of $\square$ are $(i, j)$.
The boxes with coordinates
  \( (i', j)\) where \(i'=1,\ldots, i\) are in \([\lambda]\), and their contents are
\(r=j-i,\ldots, r+i-1=j-1\geq 0\).
\end{proof}

\subsection{Representation theory of the unitary group}

We make use of the Gelfand-Tsetlin (GZ) basis of the irreps of the unitary group.
See Ref.~\citenum{wright2016learn}, Sec. 2.4 for a good review of the GZ basis.
See Ref.~\citenum{harrowApplicationsCoherentClassical2005} for applications of the GZ basis in quantum information theory.
See Ref.~\citenum{goodmanSymmetryRepresentationsInvariants2009} for a good textbook on the representation theory of the unitary group.
The GZ basis of the irrep $\mathrm{V}_\lambda^d$ is an orthonormal basis enumerated by semistandard tableaux of shape $\lambda$ with entries from $[d]$.
  If $T$ is a semistandard tableau of shape $\lambda$ with entries from $[d]$, we write $T \vDash_d \lambda$.
  For a tableau $T$ of shape $\lambda \vdash k$ with entries from $[d]$, the occupation (Eq.~\ref{eq:occupations}) $\mathrm{wt}(T)\in \omega_{[d]}^k$ such that $\mathrm{wt}(T)(i)$ is the number of times that $i$ appears in $T$ is called its weight.
For $\alpha \in \mathbb{C}^d$ and $w\in \omega_{[d]}^n$ an occupation (Eq.~\ref{eq:occupations}), define the monomial of weight $w$ to be
\begin{align}
\alpha^w = \prod_{i \in [d]}\alpha_i^{w(i)}.
\end{align}
Then define
\begin{align}
  \alpha^T &= \alpha^{\mathrm{wt}(T)}.
  \label{eq:alphatoatableau}
\end{align}
The GZ basis $\left\{ \ket{T} \right\}_{T \vDash_d \lambda}$ is an orthonormal basis of $\mathrm{V}_\lambda^d$ that obeys
\begin{align}
  \pi_\lambda(\operatorname{diag}(\alpha))\ket{T} &= \alpha^T\ket{T}
  \label{eq:gzweights}
\end{align}
for all $T \vDash_d \lambda$ and for all $\alpha \in \mathrm{U}(1)^d$, where $\operatorname{diag}(\alpha)$ is the diagonal matrix such that $\operatorname{diag}(\alpha)_{ii} = \alpha_i$.
Therefore, we can express the character $\Tr(\pi_\lambda(\mathrm{diag}(\alpha)))$ as a sum over weight monomials.
This motivates the following definition.
\begin{definition}
  For a complex vector $\alpha \in \mathbb{C}^d$ and a partition $\lambda\vdash n$, let $s_\lambda(\alpha)$ be the Schur polynomial, defined by
\begin{align}
  s_\lambda(\alpha) &= \sum_{T \vDash_d \lambda}\alpha^T.
\end{align}
\label{def:schurpolynomial}
\end{definition}

Since for all $U, K\in \mathrm{U}(d)$, we have $\Tr(\pi_\lambda(K)\pi_\lambda(U)\pi_\lambda(K^\dagger)) = \Tr(\pi_\lambda(U))$, the character of $\pi_\lambda$ is equal to the Schur polynomial of the eigenvalues of $U$.
This observation, in conjunction with Schur-Weyl duality, can be used to derive a relationship between the power sum symmetric polynomials, defined below, and the Schur polynomials.

\begin{definition}[Power-sum symmetric polynomials]
For a complex vector $\alpha \in \mathbb{C}^d$ and a partition $\lambda\vdash n$, the power-sum symmetric polynomial of shape $\lambda$ is defined to be
\begin{align}
  p_\lambda(\alpha) &= \prod_{l \in \lambda}\left(\sum_{i=1}^d \alpha_i^l\right).
\end{align}
For a permutation $\sigma \in \mathcal{S}_n$, we define $p_\sigma(\alpha)$ to be the power-sum symmetric polynomial of shape given by the cycle type of $\sigma$.
  \label{def:powersum}
\end{definition}
The Schur polynomials are a linear combination of the power-sum symmetric polynomials, with coefficients proportional to the irreducible characters of the symmetric group.
\begin{corollary}[\cite{Stanley_2011} Cor. 7.17.5]
\begin{align}
  s_\lambda(\alpha) = \frac{1}{n!}\sum_{\sigma \in \mathcal{S}_n}\chi_\lambda(\sigma)p_\sigma(\alpha).
  \label{eq:schurpowersum}
\end{align}
\end{corollary}

The Schur polynomials also appear in the definition of a certain distribution that we use in the proof of Cor.~\ref{cor:schurconvexity}.
\begin{definition}[Schur-Weyl distribution (\cite{wright2016learn}, Def.~2.6.2)]
Let $\alpha \in \Delta^{d-1}$ be a distribution on $d$ letters, and let $n$ be a positive integer.
  If $d < n$, we append $\alpha$ with the appropriate number of zeros to interpret it as a distribution over $[n]$ instead. 
The Schur-Weyl distribution $\mathrm{SW}^n(\alpha)$ is a probability distribution over partitions $\lambda \vdash n$, where the probability of obtaining $\lambda$ is
\begin{align}
  \mathrm{SW}^n(\alpha)(\lambda) = \dim(\lambda)s_\lambda(\alpha).
\end{align}
\label{def:swdistn}
\end{definition}
It is shown in Ref.~\citenum{wright2016learn}, Eq. 2.16 that the Schur-Weyl distribution is in fact a probability distribution.

Define the rising factorial of a partition $\lambda \vdash n$,
\begin{align}
  d^{\uparrow \lambda} &= \prod_{\square \in [\lambda]}(d + c(\square)).
  \label{eq:risingfactorial}
\end{align}
By Fact.~\ref{fact:allcontents}, if some box of $\lambda$ has content less than or equal to $-d$, there must exist a box with content equal to $-d$, and therefore $d^{\uparrow \lambda} = 0$.
Therefore, the rising factorial is always nonnegative.
Then we have the following formula for the number of semistandard tableaux of shape $\lambda$ whose entries are elements of $[d]$~(\cite{wright2016learn} Def.~2.2.13 and \cite{Stanley_Fomin_1999} Cor.~7.21.4),
\begin{align}
  |\left\{ T|T\vDash_d \lambda \right\}| &=  \frac{\mathrm{dim}(\lambda)d^{\uparrow \lambda}}{|\lambda|!}\\
  &= \dim(\mathrm{V}^d_\lambda),
  \label{eq:nsst}
\end{align}
where the latter inequality is because GZ basis is a basis of $\mathrm{V}^d_\lambda$.

\begin{definition}[Weight space]
  Let $W$ be a Hilbert space with $\mathrm{dim}(W) = d$. For $w \in\omega_{[d]}^n$ (Eq.~\ref{eq:occupations}), the subspace spanned by $\left\{ \ket{T}\middle| T \vDash_d \lambda \text{ s.t. }\operatorname{wt}(T) = w \right\}$ is called the $w$-weight space of $\mathrm{V}_\lambda^d$.
The $w$-weight space of $\mathbb{S}_\lambda(W)$ is denoted $(\mathbb{S}_\lambda(W))_{\mathrm{wt} = w}$.
\label{def:weightspace}
\end{definition}
For $w \in \omega_{[d]}^n$ and $W$ a Hilbert space with $\mathrm{dim}(W) = d$, the dimension of the $w$-weight space $(\mathbb{S}_\lambda(W))_{\mathrm{wt} = w}$ is thus the number of semistandard tableaux of shape $\lambda$ with weight equal to $w$.
The weight spaces are simultaneous eigenspaces of all the diagonal unitaries, in the following sense.
\begin{proposition}
  Let $W$ be a $d$-dimensional Hilbert space, let $n$ be a positive integer, let $\lambda \vdash n$, and let $w \in \omega_{[d]}^n$. Define $J = \left\{ \ket{v}\in \mathbb{S}_\lambda(W)\middle|\pi_\lambda(\mathrm{diag}(\alpha))\ket{v} = \alpha^w \ket{v} \,\forall\alpha \in \mathrm{U}(1)^d \right\}$.
  Then $(\mathbb{S}_\lambda(W))_{\mathrm{wt} = w} = J$.
  \label{prop:weightspaces}
\end{proposition}

\subsection{Visible states}

The subalgebra $\left\langle \mathrm{U}(\mathcal{H}_V)\times \mathrm{U}(\mathcal{H}_H) \right\rangle \subseteq \mathrm{End}(\mathrm{Sym}^n(\mathcal{H}_V \otimes \mathcal{H}_H))$ is called the particle-type-preserving subalgebra.
Let $\mathrm{dim}(\mathcal{H}_V) = m, \mathrm{dim}(\mathcal{H}_H) = L$, and let $K = \mathrm{min}(m, L)$.
The $n$-particle Hilbert space $\mathrm{Sym}^n(\mathcal{H}_V \otimes \mathcal{H}_H)$ when restricted to the particle-type-preserving subalgebra has the decomposition
\begin{theorem}
\begin{align}
  \mathrm{Res}^{\mathrm{U}(\mathcal{H}_V \otimes \mathcal{H}_H)}_{\mathrm{U}(\mathcal{H}_V)\times \mathrm{U}(\mathcal{H}_H)}\left(\mathrm{Sym}^n(\mathcal{H}_V \otimes \mathcal{H}_H)\right) \cong \bigoplus_{\lambda\in\mathrm{Part}_{n, K}}\mathbb{S}_\lambda(\mathcal{H}_V) \otimes \mathbb{S}_\lambda(\mathcal{H}_H).
  \label{eq:hvtensdecomp}
\end{align}
\end{theorem}
A proof of this fact appears as \cite{goodmanSymmetryRepresentationsInvariants2009}, Thm.~5.6.7.
It also appears as \cite{fultonRepresentationTheory2004}, Exercise 6.11 (b), and in \cite{roweDualPairingSymmetry2012}.
Eq.~\ref{eq:hvtensdecomp} is called Cauchy's formula in~\cite{Procesi2006-on}.
We denote the subspace of $\mathrm{Res}^{\mathrm{U}(\mathcal{H}_V \otimes \mathcal{H}_H)}_{\mathrm{U}(\mathcal{H}_V)\times \mathrm{U}(\mathcal{H}_H)}\left(\mathrm{Sym}^n(\mathcal{H}_V \otimes \mathcal{H}_H)\right)$ that is isomorphic to $\mathbb{S}_\lambda(\mathcal{H}_V) \otimes \mathbb{S}_\lambda(\mathcal{H}_H)$ by $\mathcal{H}_\lambda$.
Let $\Pi^\lambda$ be the projection onto the subspace $\mathcal{H}_\lambda$.
Then for an operator $M\in \mathrm{End}(\mathrm{Sym}^n(\mathcal{H}_V \otimes \mathcal{H}_H))$ we can restrict the domain and image of the map $\Pi^\lambda M \Pi^\lambda$ to $\mathcal{H}_\lambda$.
\begin{definition}
  Let $M \in \mathrm{End}(\mathrm{Sym}^n(\mathcal{H}_V \otimes \mathcal{H}_H))$.
Define $M_{(\lambda)}:\mathcal{H}_\lambda\rightarrow \mathcal{H}_\lambda$, given by $M_{(\lambda)} v = \Pi^\lambda M v$.
  We also define the projection of $M$ to the particle-type-preserving subalgebra by
  \begin{align}
    M^{\mathrm{PTP}} = \sum_{\lambda \in \mathrm{Part}_{n, K}}\Pi^\lambda M \Pi^\lambda = \bigoplus_{\lambda\in \mathrm{Part}_{n, K}}M_{(\lambda)}.
  \end{align}
  \label{def:dofsepprojection}
\end{definition}
For any $M \in \mathrm{End}(\mathrm{Sym}^n(\mathcal{H}_V \otimes \mathcal{H}_H))$, $M^{\mathrm{PTP}} \in \left\langle \mathrm{U}(\mathcal{H}_V)\times \mathrm{U}(\mathcal{H}_H) \right\rangle$.
Furthermore, for any $M_0 \in \left\langle \mathrm{U}(\mathcal{H}_V)\times \mathrm{U}(\mathcal{H}_H) \right\rangle$, $M^{\mathrm{PTP}}_0 = M_0$.
In our applications, we are interested in the scenario where we perform a linear optical on the visible DOF, then measure in the number basis, not resolving the hidden DOF.
Therefore, it is useful to have a notion of partial trace over the hidden DOF, thus motivating the following definition.
\begin{definition}[Visible part of an operator]
  Let $M \in \mathrm{End}(\mathrm{Sym}^n(\mathcal{H}_V \otimes \mathcal{H}_H))$.
  Let $K = \mathrm{min}(\mathrm{dim}(\mathcal{H}_V),\mathrm{dim}(\mathcal{H}_H))$.
  We write $\Tr_H(M_{(\lambda)})$ as shorthand for $\Tr_{\mathbb{S}_\lambda(\mathcal{H}_H)}(M_{(\lambda)})$.  
  We now define the visible part of $M$ by
\begin{align}
  \Tr_H(M) &= \bigoplus_{\lambda\in\mathrm{Part}_{n, K}}\Tr_H(M_{(\lambda)}) \in \bigoplus_{\lambda\in\mathrm{Part}_{n, K}}\mathrm{End}\left(\mathbb{S}_\lambda(\mathcal{H}_V)\right).
  \label{eq:visiblepart}
\end{align}
When $M = \rho$ is a state, the state $\Tr_H(\rho)$ is called the visible state of $\rho$.
\end{definition}
In Eq.~\ref{eq:perfectlyindsitinguishable}, we defined the set of perfectly indistinguishable states with single-particle hidden state \(\ket{\psi}\). 
Not all states behaving bosonically are of this form. 
The set of states behaving bosonically is the set of perfectly distinguishable states as defined next.
\begin{definition}[Perfectly indistinguishable state]
  If $\rho\in \mathcal{D}(\mathrm{Sym}^n(\mathcal{H}_V \otimes \mathcal{H}_H))$ is a state such that its visible state $\mathrm{Tr}_H(\rho)$ (Def.~\ref{eq:visiblepart}) is supported only on the irrep $\mathbb{S}_{(n)}(\mathcal{H}_V)$, it is called a perfectly indistinguishable state.
  \label{def:perfectlyindistinguishablestategeneral}
\end{definition}
This definition accords with the definitions given in Refs.~\citenum{PhysRevA.91.013844,tichySamplingPartiallyDistinguishable2015,PhysRevA.91.063842,PhysRevResearch.4.043101}.

\section{Derivation of generalized bunching probability in terms of visible states}
\label{sec:derivationintermsofvisstate}

\begin{proposition}
  Let $W$ be a $d$-dimensional Hilbert space, and $w \in \omega_{[d]}^n$ be a weight.
  Suppose $N \in \bigoplus_{\lambda\in\mathrm{Part}_{n, d}}\mathcal{B}_+(\mathbb{S}_\lambda(W))$ satisfies
  \begin{align}
    \Tr(\mathrm{diag}(\alpha)\cdot N) = \alpha^w \Tr(N)
    \label{eq:sufficientcondition}
  \end{align}
  for all $\alpha \in \mathrm{U}(1)^d$.
  Then $\mathrm{supp}(N)\subseteq \bigoplus_{\lambda\in\mathrm{Part}_{n, d}}(\mathbb{S}_\lambda(W))_{\mathrm{wt} = w}$.
  \label{prop:operatorweightspaces}
\end{proposition}
\begin{proof}
  Express $N$ in the GZ basis according to
  \begin{align}
    N = \sum_{\substack{\lambda\in\mathrm{Part}_{n, d}\\T\vDash_d \lambda}}\sum_{\substack{\mu\in\mathrm{Part}_{n, d}\\T'\vDash_d \mu}}N_{T, T'}\ketbra{T}{T'}.
  \end{align}
  Calculate that
  \begin{align}
    \Tr(\mathrm{diag}(\alpha)\cdot N) &= \sum_{\substack{\lambda\in\mathrm{Part}_{n, d}\\T\vDash_d \lambda}}N_{T, T}\alpha^T.
  \end{align}
  Therefore, writing $\int d\alpha$ for the integral over the Haar measure of $\mathrm{U}(1)^d$,
  \begin{align}
    \sum_{\substack{\lambda\in\mathrm{Part}_{n, d}\\T\vDash_d \lambda}}N_{T, T}\alpha^T &= \alpha^w \Tr(N)\\
    \implies \int d\alpha    \sum_{\substack{\lambda\in\mathrm{Part}_{n, d}\\T\vDash_d \lambda}}N_{T, T}\alpha^T(\alpha^w)^* &=  \Tr(N)\\
      \iff\sum_{\substack{\lambda\in\mathrm{Part}_{n, d}\\T\vDash_d \lambda|\mathrm{wt}(T) = w}}N_{T, T} &= \Tr(N).
  \end{align}
  Therefore, since $N$ is positive semidefinite, $N$ must annihilate any GZ basis element $\ket{T}$ for which $\mathrm{wt}(T) \neq w$.
\end{proof}

\begin{lemma}
  Let $\bm{i} \in [m]^{n}$.
Let $M\in \mathcal{B}_+(\mathrm{Sym}^n(\mathcal{H}_V \otimes \mathcal{H}_H))$ be a positive-semidefinite matrix of the form
  \begin{align}
    \sum_{\bm{j}, \bm{j}' \in [L]^{n}} M_{\bm{j}, \bm{j}'} a^\dagger(\bm{i}, \bm{j})\ketbra{0}{0}a(\bm{i}, \bm{j}').
  \end{align}
  Then for all $\lambda \vdash n$ such that $\mathrm{len}(\lambda) \le K = \mathrm{min}(m, L)$, 
  \begin{align}
   \mathrm{supp}(\Tr_H(M_\lambda))\subseteq (\mathbb{S}_\lambda(\mathcal{H}_V))_{\mathrm{wt} = \xi(\bm{i})},
  \end{align}
   where $\Tr_H(M_\lambda)$ is defined in Eq.~\ref{eq:visiblepart}, and $\xi$ is defined in Eq.~\ref{eq:xidef}.
\label{lem:weightunderhiddentrace}
\end{lemma}

\begin{proof}
  We apply Prop.~\ref{prop:operatorweightspaces}.
  Let $\alpha\in \mathrm{U}(1)^{m}$.
  Recall the definition of $\xi$ from Eq.~\ref{eq:xidef}.
  Then by the definition of partial trace,
  \begin{align}
    \operatorname{Tr}\Big((\operatorname{diag}(\alpha))\cdot\Tr_H(M_{(\lambda)})\Big)
    &=\sum_{\lambda \in \mathrm{Part}_{n, K}}\Tr(\pi_\lambda(\operatorname{diag}(\alpha))\Tr_{H}(M_{(\lambda)}) )\\
    &= \sum_{\lambda \in \mathrm{Part}_{n, K}}\Tr(\pi_\lambda(\operatorname{diag}(\alpha)) \otimes \mathds{1}_{\mathbb{S}_\lambda(\mathcal{H}_H)} M_{(\lambda)} )\\
  &=\sum_{\lambda \in \mathrm{Part}_{n, K}} \Tr(\Pi^\lambda\mathrm{Sym}^n(\operatorname{diag}(\alpha) \otimes \mathds{1}_{\mathcal{H}_H})\Pi^\lambda M)\\
  &=\Tr(\mathrm{Sym}^n(\operatorname{diag}(\alpha) \otimes \mathds{1}_{\mathcal{H}_H})^{\mathrm{PTP}} M)
  \end{align}
  Now $\mathrm{Sym}^n(\mathrm{diag}(\alpha)\otimes \mathds{1}_{\mathcal{H}_H}) \in \left\langle \mathrm{U}(\mathcal{H}_V) \times \mathrm{U}(\mathcal{H}_H)\right\rangle$ so $\mathrm{Sym}^n(\mathrm{diag}(\alpha)\otimes \mathds{1}_{\mathcal{H}_H})^{\mathrm{PTP}} = \mathrm{Sym}^n(\mathrm{diag}(\alpha)\otimes \mathds{1}_{\mathcal{H}_H})$.
  Therefore,
  \begin{align}
    \operatorname{Tr}\Big((\operatorname{diag}(\alpha))\cdot\Tr_H(M_{(\lambda)})\Big)  &= \Tr((\operatorname{diag}(\alpha) \otimes \mathds{1}_{\mathcal{H}_H})\cdot M)\\
     &= \alpha^{\xi(\bm{i})}\Tr(M).
  \end{align}
  Therefore, by Prop.~\ref{prop:operatorweightspaces}, $\operatorname{supp}(\Tr_H(M)) \subseteq \bigoplus_{\lambda \in \mathrm{Part}_{n, K}}(\mathbb{S}_\lambda(\mathcal{H}_V))_{\mathrm{wt} = w}$.

  Now, because $\Tr_H(M)|_{\mathbb{S}_\lambda(\mathcal{H}_V)} = \Tr_H(M_{(\lambda)})$, we have 
  \begin{equation}
    \mathrm{supp}(\Tr_H(M_{(\lambda)})) = \mathrm{supp}(\Tr_H(M)|_{\mathbb{S}_\lambda(\mathcal{H}_V)}) \subseteq (\mathbb{S}_\lambda(\mathcal{H}_V))_{\mathrm{wt} = w}.
  \end{equation}
  Therefore, the lemma is proved.
\end{proof}

For an irrep $\mathbb{S}_\lambda(\mathcal{H})$, we write $\Pi^\lambda_{g}$ for the projection onto the $g$-weight space of $\mathbb{S}_\lambda(\mathcal{H})$. 
\begin{lemma}
Let $v \in \omega^n_{[m]}$ be an occupation (Eq.~\ref{eq:occupations}).
Let $N_v$ be as defined in Eq.~\ref{eq:nvdef}.
  Assume $m \ge n$ and $L \ge n$.
  Then 
  \begin{align}
    \sum_{g \in N_v}\ketbra{g}{g} &= \bigoplus_{\lambda \vdash n}\Pi^\lambda_v \otimes \mathds{1}_{\mathbb{S}_\lambda(\mathcal{H}_H)},
  \end{align}
  with $\ket{g}$ a state of definite occupation as defined in Eq.~\ref{eq:ketg}.
  \label{lem:projectorontoweightspace}
\end{lemma}

\begin{proof}
  Define $\Pi_v = \sum_{g \in N_v}\ketbra{g}{g}$.  We first show that
  $\Pi_v$ is invariant under the adjoint action of hidden unitaries.
  Let $\mathrm{zip}(\bm{i}, \bm{j}) = \left( (i_1, j_1), \ldots, (i_n, j_n) \right)$.  
  Recall the definition of $\xi$ from Eq.~\ref{eq:xidef} and $\zeta$ from Eq.~\ref{eq:zetadef}.
    We calculate
  \begin{align}
    \sum_{g \in N_v}\ketbra{g}{g} &= \sum_{l \in \omega^n_{[L]}}\frac{1}{\xi(\mathrm{zip}(\zeta(v),\zeta(l)))!}a^\dagger(\zeta(v),\zeta(l))\ketbra{0}{0}a(\zeta(v),\zeta(l))\\
    &= \sum_{\bm{j}\in [L]^{n}}\frac{\xi(\mathrm{zip}(\zeta(v),\bm{j}))!}{g!}\frac{1}{\xi(\mathrm{zip}(\zeta(v),\bm{j}))!}a^\dagger(\zeta(v),\bm{j})\ketbra{0}{0}a(\zeta(v),\bm{j})\\
    &= \frac{1}{v!}\sum_{\bm{j}\in [L]^{n}}a^\dagger(\zeta(v),\bm{j})\ketbra{0}{0}a(\zeta(v),\bm{j})
  \end{align}
  Let $U(\bm{k}|\bm{j}) = \prod_{x=1}^n U_{k_x,j_x}$.
  Then,
  \begin{align}
    (\mathds{1}_{\mathcal{H}_V}\otimes U)\cdot\Pi_v\cdot(\mathds{1}_{\mathcal{H}_V}\otimes U)^\dagger &= \frac{1}{v!}\sum_{\bm{j}\in [L]^{n}}a^\dagger(\zeta(v),U\cdot\bm{j})\ketbra{0}{0}a(\zeta(v),U\cdot\bm{j})\\
    &= \frac{1}{v!}\sum_{\bm{j},\bm{k},\bm{k}'\in [L]^{n}}U(\bm{k}|\bm{j})U^*(\bm{k}'|\bm{j})a^\dagger(\zeta(v),\bm{k})\ketbra{0}{0}a(\zeta(v),\bm{k}')\\
    &= \frac{1}{v!}\sum_{\bm{k},\bm{k}'\in [L]^{n}}\delta_{\bm{k},\bm{k}'}a^\dagger(\zeta(v),\bm{k})\ketbra{0}{0}a(\zeta(v),\bm{k}')\\
    &= \frac{1}{v!}\sum_{\bm{k}\in [L]^{n}}a^\dagger(\zeta(v),\bm{k})\ketbra{0}{0}a(\zeta(v),\bm{k})\\
    &= \Pi_v
  \end{align}
  Therefore, by Schur's lemma,
  \begin{align}
    \Pi_v &=  \bigoplus_{\lambda \vdash n}\Xi^\lambda_v\otimes \mathds{1}_{\mathbb{S}_\lambda(\mathcal{H}_H)}
  \end{align}
  for some projectors $\Xi^\lambda_v$ acting on $\mathbb{S}_\lambda(\mathcal{H}_V)$.
  By Lem.~\ref{lem:weightunderhiddentrace}, $\mathrm{supp}(\Xi_v^\lambda)\subset \left( \mathbb{S}_\lambda(\mathcal{H}_V) \right)_{\mathrm{wt} = v}$.

  We use dimension counting to show the reverse inclusion.
  It is clear that $\Tr(\Pi_v) = |N_v|$.
  Elements of $N_v$ are in bijection with matrices of nonnegative integers such that the $i$th row sums to $v(i)$.
  By the Robinson-Schensted-Knuth correspondence~(\cite{Stanley_Fomin_1999} Thm.~7.11.5), the set of such matrices is in bijection with pairs $(P, Q)$ of semistandard tableaux with $n$ boxes such that the shape of $P$ is equal to the shape of $Q$, and the weight of $P$ is equal to $v$. 
  Let $Y^\lambda_v$ be the number of semistandard tableaux of shape $\lambda$ with weight $v$.
  Then we have
  \begin{align}
    \Tr(\Pi_v) &= \sum_{\lambda \vdash n}Y^\lambda_v\sum_{l\in \omega^n_{[L]}}Y^\lambda_l\\
    &= \sum_{\lambda \vdash n}\mathrm{dim}\left(\left( \mathbb{S}_\lambda(\mathcal{H}_V) \right)_{\mathrm{wt} = v}\right)\sum_{l\in \omega^n_{[L]}}\mathrm{dim}\left(\left( \mathbb{S}_\lambda(\mathcal{H}_V) \right)_{\mathrm{wt} = l}\right)\\
    &= \sum_{\lambda \vdash n}\mathrm{dim}\left(\left( \mathbb{S}_\lambda(\mathcal{H}_V) \right)_{\mathrm{wt} = v}\right)\mathrm{dim}\left( \mathbb{S}_\lambda(\mathcal{H}_V)\right).
  \end{align}
  Therefore, the support of $\Pi_v$ must be the whole of $\bigoplus_{\lambda \vdash n}\left(\left( \mathbb{S}_\lambda(\mathcal{H}_V) \right)_{\mathrm{wt} = v}\right)\otimes\left( \mathbb{S}_\lambda(\mathcal{H}_V)\right)$, and therefore $\Xi_v^\lambda$ in fact is equal to $\Pi^\lambda_v$.
\end{proof}

\begin{lemma}
  Let $\rho \in \mathcal{D}(\mathrm{Sym}^n(\mathcal{H}_V \otimes \mathcal{H}_H))$.
  Assume $n \le m, n \le L$.
  Then the probability of observing an occupation of $v \in \omega_{[m]}^n$ after evolving by $U \in \mathrm{U}(\mathcal{H}_V)$ (Eq.~\ref{eq:visibleoccupationdistributions}) is given by
\begin{align}
  p(v|U, \rho) &= \sum_{\lambda \vdash n}\Tr(\pi_\lambda(U) \Tr_H(\rho_{(\lambda)}) \pi_\lambda(U)^\dagger \Pi_{v}^\lambda),
\end{align}
where $\rho_{(\lambda)}$ is defined in Def.~\ref{def:dofsepprojection}.
\label{lem:modelafterhiddentrace}
\end{lemma}
\begin{proof}
  The probability of observing outcome $v$ is given by
  \begin{align}
    p(v|U,\rho) &= \sum_{g\in N_v}\Tr\left( (U\otimes \mathds{1})\cdot \rho \cdot (U^\dagger\otimes \mathds{1}) \ketbra{g}{g}\right)\\
    &= \Tr\left( (U\otimes \mathds{1})\cdot \rho \cdot (U^\dagger\otimes \mathds{1}) \bigoplus_{\lambda \vdash n}\Pi_v^\lambda\otimes \mathds{1}_{\mathbb{S}_\lambda(\mathcal{H}_H)} \right)
  \end{align}
  by Lem.~\ref{lem:projectorontoweightspace}.
  Then,
  \begin{align}
    p(v|U,\rho) &= \sum_{\lambda \vdash n}\Tr\left( \pi_\lambda(U)\otimes \mathds{1}_{\mathbb{S}_\lambda(\mathcal{H}_H)}\, \rho\, \pi_\lambda(U^\dagger)\otimes \mathds{1}_{\mathbb{S}_\lambda(\mathcal{H}_H)} \Pi_v^\lambda\otimes \mathds{1}_{\mathbb{S}_\lambda(\mathcal{H}_H)} \right)\\
    &= \sum_{\lambda \vdash n}\Tr(\pi_\lambda(U) \Tr_H(\rho_{(\lambda)}) \pi_\lambda(U)^\dagger \Pi_{v}^\lambda).
  \end{align}
\end{proof}

Define $\Pi_S^\lambda = \sum_{v\in \omega^n_{[m]}(S)}\Pi^\lambda_{v}$, $\Pi^\lambda_S$ is the orthogonal projector onto the subspace of $\mathbb{S}_\lambda(\mathcal{H}_V)$ whose weights are nonzero only on $S$.
  An immediate consequence of the previous lemma is
\begin{corollary}
  Let $\rho \in \mathcal{D}(\mathrm{Sym}^n(\mathcal{H}_V \otimes \mathcal{H}_H))$. Assume that $n \le m, n \le L$. Then the probability that all $n$ particles arrive in the subset $S$ after applying the linear optical $U \in \mathrm{U}(\mathcal{H}_V)$ is
\begin{align}
  b(S|U, \rho) &= \sum_{\lambda \vdash n}\Tr\left( \pi_\lambda(U) \Tr_H(\rho_{(\lambda)}) (\pi_\lambda(U))^{\dagger} \Pi_S^\lambda \right).
\end{align}
\label{cor:genbunchafterhiddentrace}
\end{corollary}

\section{Derivation of the visible state for permutation-invariant states}
\label{sec:derivationforperminv}

In this section we calculate the visible state of permutation-invariant states.
Our analysis makes use of the decomposition of the irreps of the unitary group into irreps of the symmetric group, under restriction to the permutation matrices.
In general this problem is involved~\cite{heaton2021branching,orellana2021symmetric,harman2018representations}, but for singly occupied states the decomposition is relatively simple.

Recall the definition of $\xi$ from Eq.~\ref{eq:xidef}.
\begin{lemma}
  Let $\bm{i}$ be a list of $n$ distinct indices from $[m]$.
  Then for all $\lambda \vdash n$, the subspace $\mathbb{S}_\lambda(\mathcal{H}_V)_{\mathrm{wt} = \xi(\bm{i})} \subset \mathbb{S}_\lambda(\mathcal{H}_V)$ is closed under the action of $\mathcal{S}_n(\bm{i})\subset \mathrm{U}(\mathcal{H}_V)$ (Eq.~\ref{eq:snidef}).
  \label{lem:vissn}
\end{lemma}
\begin{proof}
  Let $\ket{T}$ be a GZ basis element for a standard tableau $T$ with entries given by $\bm{i}$.
  Then,
  \begin{align}
    \bra{T}F_{\bm{i}}(\sigma)^{\dagger}\operatorname{diag}(\alpha)F_{\bm{i}}(\sigma)\ket{T}
 &=    \bra{T}\operatorname{diag}(\sigma\cdot\alpha)\ket{T}\\
 &= (\sigma\cdot\alpha)^{T}\\
 &= \alpha^T,
  \end{align}
  where the last equality is because each nonzero weight in $T$ is equal to 1.
  Therefore, $F_{\bm{i}}(\sigma)\ket{T}$ is a weight vector with weight given by that of $\ket{T}$.
\end{proof}
Lemma~\ref{lem:vissn} implies that $\mathbb{S}_\lambda(\mathcal{H}_V)_{\mathrm{wt}=\xi(\bm{i})}$ is a representation of $\mathcal{S}_n$ under the action obtained from restricting the action of $\mathrm{U}(\mathcal{H}_V)$ to $\mathcal{S}_n(\bm{i})$.
\begin{lemma}[\cite{Procesi2006-on}, p.~272]
  Let $\bm{i}$ be a list of $n$ distinct indices from $[m]$.
  Then for all $\lambda \vdash n$, 
  \begin{align}
\mathbb{S}_\lambda(\mathcal{H}_V)_{\mathrm{wt} = \xi(\bm{i})} &\cong \mathrm{Sp}_\lambda
  \end{align}
as an $\mathcal{S}_n$-representation.
\label{lem:restrictionproblem}
\end{lemma}
\begin{proof}
  We can introduce any $\mathcal{H}_H$ of our choice, so let $\mathrm{dim}(\mathcal{H}_H) \ge n$. 
  Without loss of generality, assume that $\bm{i} = \left( 1, \ldots, n \right)$.
  By Cor.~\ref{cor:kiswdecomp}, $K_{\bm{i}}$ (Eq.~\ref{eq:singlyoccupiesdef}) is isomorphic to $\bigoplus_{\lambda\vdash n}\mathrm{Sp}_\lambda\otimes \mathbb{S}_\lambda(\mathcal{H}_H)$ under the action of $\mathcal{S}_n \times \mathrm{U}(\mathcal{H}_H)$.
  Define $(K_{\bm{i}})_\lambda$ to be the subspace of $K_{\bm{i}}$ that is isomorphic to the irrep $\mathrm{Sp}_\lambda \otimes \mathbb{S}_{\lambda}(\mathcal{H}_H)$.
  Define $\Pi_{(K_{\bm{i}})_\lambda}:\operatorname{Sym}^n(\mathcal{H}_V \otimes \mathcal{H}_H) \rightarrow \operatorname{Sym}^n(\mathcal{H}_V \otimes \mathcal{H}_H)$ to be the orthogonal projection onto $(K_{\bm{i}})_\lambda$.
  Observe that $\Pi_{(K_{\bm{i}})_\lambda}$ is an element of the particle-type-preserving subalgebra $\left\langle \mathrm{U}(\mathcal{H}_V)\times \mathrm{U}(\mathcal{H}_H) \right\rangle$.

Let
\begin{align}
  R_\lambda = \frac{1}{\mathrm{dim}(\mathrm{V}^d_\lambda)}\Tr_{H}(\Pi_{(K_{\bm{i}})_\lambda}).
\end{align}
The partial trace operation is an $\mathcal{S}_n$-intertwiner, so since $\Pi_{(K_{\bm{i}})_\lambda}$ commutes with $\sigma$ for all $\sigma \in \mathcal{S}_n$, $R_\lambda$ does as well.  
  Therefore, $R_\lambda$ is the direct sum of operators that are proportional to projectors onto irreps of $\mathcal{S}_n$.
Since $\Pi_{(K_{\bm{i}})_\lambda}$ has support only on the $\lambda$ isotype of $\mathcal{S}_n$, this means that $R_\lambda$ must only be supported on the $\lambda$-irrep.  
Therefore, $R_\lambda$ is proportional to a projector onto the $\lambda$-irrep of $\mathcal{S}_n$.

By Lem.~\ref{lem:weightunderhiddentrace}, we have
\begin{align}
  \mathrm{supp}\left(R_\lambda\right) \subseteq (\mathbb{S}_\lambda(\mathcal{H}_V))_{\mathrm{wt} = \xi(\bm{i})}.
\end{align}
Since the sites $\bm{i}$ are distinct, any semistandard tableau with weight $\xi(\bm{i})$ has distinct entries. 
Therefore, it must be a standard tableau, and all standard tableaux are semistandard, so $\mathrm{dim}( (\mathbb{S}_\lambda(\mathcal{H}_V))_{\mathrm{wt} = \xi(\bm{i})}) = \mathrm{dim}(\lambda) = \mathrm{dim}(\mathrm{Sp}_\lambda)$.
We conclude that
\begin{align}
  (\mathbb{S}_\lambda(\mathcal{H}_V))_{\mathrm{wt} = \xi(\bm{i})} &\cong \mathrm{Sp}_\lambda.
\end{align}
\end{proof}

\begin{lemma}
  Let $\rho \in \mathrm{Sym}^n(\mathcal{H}_V \otimes \mathcal{H}_H)$ be a permutation-invariant state that singly occupies the $n$ distinct sites $\bm{i}\in [m]^n$ (Def.~\ref{def:permutationinvariantstates}).
  Then with $\Tr_H(\rho)$ as defined in Eq.~\ref{eq:visiblepart}, 
  \begin{align}
    \Tr_H(\rho) &= \bigoplus_{\lambda\in\mathrm{Part}_{n, L}}q_\lambda \frac{\mathds{1}_{(\mathbb{S}_\lambda(\mathcal{H}_V))_{\mathrm{wt} = \xi(\bm{i})}}}{\mathrm{dim}(\lambda)},
  \end{align}
  where $\{q_\lambda\}_{\lambda\in\mathrm{Part}_{n, L}}$ is a probability distribution over partitions of $n$.
  \label{lem:mixturemodel}
\end{lemma}
\begin{proof}
  By Lem.~\ref{lem:weightunderhiddentrace}, $\mathrm{supp}(\Tr_H(\rho))\subseteq \bigoplus_{\lambda\in\mathrm{Part}_{n, L}}(\mathbb{S}_\lambda(\mathcal{H}_V))_{\mathrm{wt} = \xi(\bm{i})}$.
  Since $\rho$ is invariant under the visible permutation matrices that act nontrivially only on $\bm{i}$ and by Lem.~\ref{lem:restrictionproblem}
  $(\mathbb{S}_\lambda(\mathcal{H}_V))_{\mathrm{wt} = \xi(\bm{i})} \cong \mathrm{Sp}_\lambda$, we can apply Cor.~\ref{cor:schurcor} to conclude that
  \begin{align}
    \Tr_H(\rho) &= \bigoplus_{\lambda\in\mathrm{Part}_{n, L}}q_\lambda \frac{\mathds{1}_{(\mathbb{S}_\lambda(\mathcal{H}_V))_{\mathrm{wt} = \xi(\bm{i})}}}{\mathrm{dim}(\lambda)}.
  \end{align}
  Since $\Tr_H(\rho)$ is a density matrix, the $q_\lambda$ must form a probability distribution.
\end{proof}

\begin{lemma}
  Let $\alpha$ be a probability distribution on $[L]$, and let $\bm{i}$ be a list of $n$ distinct indices from $[m]$.
  Let $\alpha(\bm{j}) = \alpha_{j_1}\cdots\alpha_{j_n}$, and
\begin{align}
  \rho(\alpha) &= \sum_{\bm{j} \in [L]^{n}}\alpha(\bm{j})a^{\dagger}(\bm{i}, \bm{j})\ketbra{0}{0}a(\bm{i}, \bm{j}).
\end{align}
Then, with the definition of the Schur polynomial $s_\lambda(\alpha)$ from Def.~\ref{def:schurpolynomial},
  \begin{align}
    \Tr_H(\rho(\alpha)) &= \bigoplus_{\lambda\vdash n}\mathrm{dim}(\lambda)s_\lambda(\alpha)\frac{\mathds{1}_{(\mathbb{S}_\lambda(\mathcal{H}_V))_{\mathrm{wt} = \xi(\bm{i})}}}{\mathrm{dim}(\lambda)},
\end{align}
where if $\mathrm{len}(\lambda) > L$, then $\alpha_k$ is defined to be zero for $k > L$.
  \label{lem:modelderivation}
\end{lemma}
\begin{proof}
    By possibly extending $\alpha$ by the appropriate number of zeros, we can assume without loss of generality that $L \ge n$.
  For brevity, we write $\rho = \rho(\alpha)$ for this proof.
Because \(\alpha(\bm{j})\) is invariant under permutation of \(\bm{j}\), we apply Lem.~\ref{lem:mixturemodel}, to obtain $\Tr_H(\rho) = \bigoplus_{\lambda\vdash n}q_\lambda \frac{\mathds{1}_{(\mathbb{S}_\lambda(\mathcal{H}_V))_{\mathrm{wt} = \xi(\bm{i})}}}{\mathrm{dim}(\lambda)}$ for some probability distribution $q_\lambda$.
  It therefore suffices to calculate this distribution.

  We write $\delta(\bm{k}, \bm{l}) = \delta_{\bm{k}, \bm{l}}$ to make the arguments easier to read.
  Recall the definition of $F_{\bm{i}}$ from Def.~\ref{def:frepdef}.
  Then 
  \begin{align}
    \Tr( (F_{\bm{i}}(\sigma) \otimes \mathds{1}_{\mathcal{H}_H})\cdot\rho) &= \Tr((F_{\bm{i}}(\sigma) \otimes \mathds{1}_{\mathcal{H}_H})\cdot\sum_{\bm{j}}\alpha(\bm{j})a^{\dagger}(\bm{i}, \bm{j})\ketbra{0}{0}a(\bm{i}, \bm{j}))\\
&= \sum_{\bm{j}}\alpha(\bm{j})\Tr(a^{\dagger}(\bm{i}, \sigma\cdot\bm{j})\ketbra{0}{0}a(\bm{i}, \bm{j}))\\
    &= \sum_{\bm{j}}\alpha(\bm{j})\delta(\bm{j}, \sigma\cdot \bm{j}),
    \label{eq:middleofcalc}
  \end{align}
    by Prop.~\ref{prop:fintertwines}.
    Let $\{c_1, \ldots, c_t\}$ be the cycle decomposition of $\sigma$.
  Then the $\delta$ factor in the above expression is nonzero if and only if for all $a \in [t]$, for all $x, y \in c_a$, $j_x = j_y$.
  That is, the cycle decomposition of $\sigma$ ``matches'' the equality pattern of $\bm{j}$.
Recalling Def.~\ref{def:powersum},
  \begin{align}
    \Tr( (F_{\bm{i}}(\sigma) \otimes \mathds{1}_{\mathcal{H}_H})\cdot\rho) = \sum_{\bm{j}}\alpha(\bm{j})\delta(\bm{j}, \sigma\cdot \bm{j}) = 
    \sum_{j_1, \ldots, j_t}\alpha_{j_1}^{|c_1|}\cdots\alpha_{j_t}^{|c_t|} &= p_\sigma(\alpha).
    \label{eq:beforepluggingin}
  \end{align}
  
  On the other hand, by Lem.~\ref{lem:mixturemodel} applied to $\rho$, we have 
  \begin{align}
    \Tr( (F_{\bm{i}}(\sigma) \otimes \mathds{1}_{\mathcal{H}_H})\cdot \rho) &= \Tr( F_{\bm{i}}(\sigma) \Tr_H(\rho)) \\
    &= \Tr\left( \bigoplus_{\mu\vdash n}q_\mu \pi_\mu(F_{\bm{i}}(\sigma)) \frac{\mathds{1}_{(\mathbb{S}_\mu(\mathcal{H}_V))_{\mathrm{wt} = \xi(\bm{i})}}}{\mathrm{dim}(\mu)} \right)\\
    \end{align}
    By Lemma~\ref{lem:restrictionproblem}, $\Tr\left(\pi_\mu\left(F_{\bm{i}}\left(\sigma\right)\right)\right) = \chi_\mu(\sigma)$.
    Therefore, 
    \begin{align}
     \Tr( (F_{\bm{i}}(\sigma) \otimes \mathds{1}_{\mathcal{H}_H}) \rho)   &= \sum_{\mu \vdash n} q_\mu \frac{\chi_\mu(\sigma)}{\mathrm{dim}(\mu)}.
  \end{align}
Therefore, by Eq.~\ref{eq:beforepluggingin}, 
\begin{align}
    p_\sigma(\alpha) &= \sum_{\mu \vdash n} q_\mu \frac{\chi_\mu(\sigma)}{\mathrm{dim}(\mu)}.
\end{align}
By Lem.~\ref{lem:irrechorthonormality},
\begin{align}
\frac{1}{n!}\sum_{\sigma}p_{\sigma}(\alpha)\chi_{\lambda}(\sigma)^* 
    &= \frac{1}{n!}\sum_\sigma \sum_{\mu\vdash n}\chi_\lambda(\sigma)^*\chi_\mu(\sigma)\frac{q_\mu}{\mathrm{dim}(\mu)}\\
     \iff \frac{1}{n!}\sum_{\sigma}p_{\sigma}(\alpha)\chi_{\lambda}(\sigma)   &= \frac{q_\lambda}{\mathrm{dim}(\lambda)},
\end{align}
where we have used $\chi_{\lambda}(\sigma) = \chi_{\lambda}(\sigma)^*$~(\cite{fultonYoungTableauxApplications1996} $\mathsection$ 7.3, Cor. 4).
Recalling Def.~\ref{def:schurpolynomial}, and Eq.~\ref{eq:schurpowersum}, we recognize the left-hand side as a Schur polynomial, so
\begin{align}
    s_\lambda(\alpha) &= \frac{q_\lambda}{\mathrm{dim}(\lambda)},
  \end{align}
  and therefore $q_\lambda = s_\lambda(\alpha)\mathrm{dim}(\lambda)$, and we have proven the lemma.
\end{proof}

Next we show that a perfectly indistinguishable state with single-particle wavefunction $\psi \in \mathcal{H}_H$ as defined in Eq.~\ref{eq:perfectlyindsitinguishable} is in fact perfectly indistinguishable, as defined in Def.~\ref{def:perfectlyindistinguishablestategeneral}.
\begin{proposition}
  Let $\bm{i}$ be a list of $n$ distinct indices from $[m]$, let $\psi \in \mathcal{H}_H$, and let $\phi_{\bm{i}, \psi}$ be as defined in Eq.~\ref{eq:perfectlyindsitinguishable}.
  Then $\phi_{\bm{i}, \psi}$ is perfectly indistinguishable.
  \label{prop:pefectlyindistinguishableisperfectlyindistinguishable}
\end{proposition}
\begin{proof}
  We calculate that the auxiliary state (Def.~\ref{def:auxiliarystate}) is
  \begin{align}
    h(\phi_{\bm{i}, \psi}) = \ketbra{\psi}^{\otimes n}.
  \end{align}
  Therefore, for all $\sigma \in \mathcal{S}_n$, $P(\sigma)h(\phi_{\bm{i}, \psi}) = h(\phi_{\bm{i}, \psi})P(\sigma)^{-1} = h(\phi_{\bm{i}, \psi})$.
  Therefore, $h(\phi_{\bm{i}, \psi})$ is supported only on $\mathrm{Sp}_{(n)} \otimes \mathbb{S}_{(n)}(\mathcal{H}_H)$.
  So, by Lem.~\ref{lem:hintertwinesfandp}, $\phi_{\bm{i}, \psi}$ is supported only on $\mathbb{S}_{(n)}(\mathcal{H}_V)\otimes \mathbb{S}_{(n)}(\mathcal{H}_H)$, so $\phi_{\bm{i}, \psi}$ is perfectly indistinguishable.
\end{proof}

\section{Schur convexity of mean generalized bunching probability}
\label{sec:schurconvexity}
In this section we prove Thm.~\ref{thm:meangenbunchsuremajorization}, and its corollaries Cor.~\ref{cor:schurconvexity}, Cor.~\ref{cor:thermcor}.
Define $\Delta(\mathrm{Part}_n)$ be the set of probability distributions over partitions of $n$.
We have the following distribution over partitions, which is defined naturally for any $n$-particle state.
\begin{definition}
  Let $\rho \in \mathcal{D}(\mathrm{Sym}^n(\mathcal{H}_V \otimes \mathcal{H}_H))$ be a state. Recall the definition of $\rho_{(\lambda)}$ from Def.~\ref{def:dofsepprojection}. We define the irrep distribution of $\rho$, written $D_\rho \in \Delta(\mathrm{Part}_n)$, to be the distribution over partitions of $n$, given by
  \begin{align}
    D_\rho(\lambda) = \Tr(\rho_{(\lambda)}).
  \end{align}
  \label{def:irrepdistribution}
\end{definition}
\begin{lemma}
  Let $n \le m = \mathrm{dim}(\mathcal{H}_V)$ and $n \le L = \mathrm{dim}(\mathcal{H}_H)$.
  Let $\rho \in \mathcal{D}(\mathrm{Sym}^n(\mathcal{H}_V \otimes \mathcal{H}_H))$, and let $S\subseteq [m]$ be a subset of size $k$. 
Then 
  \begin{align}
    \mathbb{E}_{U \sim \mathrm{Haar}(\mathrm{U}(\mathcal{H}_V))}\left(b(S|U, \rho\right)) &= \mathbb{E}_{\lambda\sim D_\rho}\left(\frac{k^{\uparrow \lambda}}{m^{\uparrow \lambda}}\right).
  \end{align}
  \label{lem:meangenbunchgeneral}
\end{lemma}
\begin{proof}
By Cor.~\ref{cor:genbunchafterhiddentrace}, the probability that all $n$ particles arrive in the subset $S$ after applying the linear optical $U \in \mathrm{U}(\mathcal{H}_V)$ is
\begin{align}
  b(S|U, \rho) &= \sum_{\lambda \vdash n}\Tr\left( \pi_\lambda(U) \Tr_H(\rho_{(\lambda)}) (\pi_\lambda(U))^{\dagger} \Pi_S^\lambda \right),
  \label{eq:genbunchfirstintro}
\end{align}
where $\Pi^\lambda_S$ the orthogonal projector onto the subspace of $\mathbb{S}_\lambda(\mathcal{H}_V)$ whose weights are nonzero only on $S$.
The Haar average of the bunching probability is obtained by integrating this equation over unitaries according to
\begin{align}
  \mathbb{E}_{U \sim \text{Haar}(\mathrm{U}(\mathcal{H}_V))}\left(b(S|U, \rho\right)) &= \sum_{\lambda \vdash n}\int dU \Tr\left( \pi_\lambda(U) \Tr_H(\rho_{(\lambda)}) (\pi_\lambda(U))^{\dagger} \Pi_S^\lambda\right)\\
  &= \sum_{\lambda \vdash n}\Tr\left(  \Tr_H(\rho_{(\lambda)})\underbrace{ \int dU (\pi_\lambda(U))^{\dagger} \Pi_S^\lambda\pi_\lambda(U)}\right)
  \label{eq:gbgeneralmiddleofcalc}
\end{align}
where \(\int dU\) denotes integration over the Haar
  measure over unitaries on the visible Hilbert space.
Let $\Pi^\lambda$ be the orthogonal projector onto
$\mathbb{S}_\lambda(\mathcal{H}_V)$ in
$\bigoplus_{\mu\vdash n}\mathbb{S}_\mu(\mathcal{H}_V)$.
We calculate that the underbraced expression is
\begin{align}
  \int dU (\pi_\lambda(U))^{\dagger} \Pi_S^\lambda \pi_\lambda(U)  &= 
    \frac{\Tr(\Pi_S^\lambda)}{\mathrm{dim}(\mathrm{V}^m_\lambda)}\Pi^\lambda
\end{align}
by Schur's lemma. 
The subspace of $\mathbb{S}_\lambda(\mathcal{H}_V)$ with weights nonzero only on $S$ has dimension equal to the number of semistandard tableaux of shape $\lambda$ with entries from $[\left| S \right|] = [k]$.
Therefore, since $\Pi_S^\lambda$ is the orthogonal projector onto this subspace, $\Tr(\Pi_S^\lambda) = \mathrm{dim}(\mathrm{V}^k_\lambda)$.
Now we return to Eq.~\ref{eq:gbgeneralmiddleofcalc},
\begin{align}
  \mathbb{E}_{U \sim \text{Haar}(\mathrm{U}(\mathcal{H}_V))}\left(b(S|U, \rho\right)) &= \frac{\mathrm{dim}(\mathrm{V}^k_\lambda)}{\mathrm{dim}(\mathrm{V}^m_\lambda)}\sum_{\lambda \vdash n}\Tr\left(  \Tr_H(\rho_{(\lambda)}) \Pi^\lambda \right)\\
  \end{align}
  Now observe that
  \begin{align}
    \Tr\left(  \Tr_H(\rho_{(\lambda)}) \Pi^\lambda \right) = \Tr\left(  \Tr_H(\rho_{(\lambda)}) \right) = \Tr\left(\rho_{(\lambda)} \right) = D_\rho(\lambda).
  \end{align}
  Therefore,
  \begin{align}
    \mathbb{E}_{U \sim \text{Haar}(\mathrm{U}(\mathcal{H}_V))}\left(b(S|U, \rho\right)) &= \sum_{\lambda \vdash n}D_\rho(\lambda)\frac{\mathrm{dim}(\mathrm{V}^k_\lambda)}{\mathrm{dim}(\mathrm{V}^m_\lambda)}\\
    &= \sum_{\lambda \vdash n}D_\rho(\lambda)\frac{k^{\uparrow \lambda}}{m^{\uparrow \lambda}}\\
    &= \mathbb{E}_{\lambda \sim D_\rho}\left( \frac{k^{\uparrow \lambda}}{m^{\uparrow \lambda}} \right),
\end{align}
where the penultimate equality follows by Eq.~\ref{eq:nsst}.
\end{proof}
In Lem.~\ref{lem:fschurconvex}, we show that $f:\mathrm{Part}_n\to\mathbb{R}, f(\lambda) = \frac{k^{\uparrow \lambda}}{m^{\uparrow \lambda}}$ is monotonic with respect to the majorization partial order on partitions, which we define next.  
This definition is an extension of Def.~\ref{def:majorization} to partitions.
\begin{definition}[Majorization for partitions]
For $\lambda, \mu \vdash n$, we say that $\lambda$ majorizes $\mu$, written $\lambda \succ \mu$, if for all $k \in [n]$,
\begin{align}
  \sum_{i=1}^k \lambda_i \ge \sum_{i=1}^k \mu_i.
\end{align}
\label{def:majorizationforpartitions}
\end{definition}
\begin{lemma}
  Let $n, k \in [m]$. Define $f:\mathrm{Part}_n\rightarrow \mathbb{R}$ by $f(\kappa) = \frac{k^{\uparrow \kappa}}{m^{\uparrow \kappa}}$.
  Then $f$ is well-defined and $f\ge 0$ with equality exactly when $\kappa$ has a box with content equal to $-k$.
  Furthermore, if $\mu, \lambda \vdash n$ such that $\mu \neq \lambda$ and $f(\mu) \neq 0$, then
\begin{align}
\mu \succ \lambda \implies f(\mu)> f(\lambda).
  \label{eq:claim}
\end{align}
If $f(\mu) = 0$, then $\mu \succ \lambda \implies f(\mu) = f(\lambda) = 0$.
\label{lem:fschurconvex}
\end{lemma}
\begin{proof}
  With the definition of content from Def.~\ref{def:contentofabox}, we can write
\begin{align}
  f(\kappa) = \frac{k^{\uparrow \kappa}}{m^{\uparrow \kappa}} &= \prod_{\square \in [\kappa]}\frac{k+c(\square)}{m+c(\square)}.
\end{align}
Observe that since $n \le m$, we have $\forall\, \square \in [\kappa]$, $c(\square) > -m$, so we are never dividing by 0, and therefore $f$ is well-defined.
Since the rising factorial is always nonnegative, $f \ge 0$.

If \(\mu\succ\lambda\), then it is possible to obtain $\lambda$ from $\mu$ through a sequence of steps, where in each step a box is moved from the end of a longer row to the end of a shorter one.
Such a step is called a robin-hood transfer~\cite{arnoldMajorization2011}.
Therefore, it suffices to prove the statements of the lemma when $\lambda$ is obtained from $\mu$ by a robin-hood transfer.
The content of a box moved in a robin-hood transfer must decrease, because its column index must decrease, and its row index must increase.

Suppose that $f(\mu) = 0$.
Then
\begin{align}
  0 &=  \prod_{\square \in [\mu]}\frac{k + c(\square)}{m+c(\square)}\\
  \implies
  0 &=  \prod_{\square \in [\mu]}(k + c(\square))
\end{align}
and therefore there exists some box $\square \in [\mu]$ such that $c(\square) = -k$.
Since the content of the box being moved must decrease, by Fact~\ref{fact:allcontents}, there must be a box of $\lambda$ with content equal to $-k$, and we can conclude that $f(\mu) = 0 \implies f(\lambda) = 0$.

Next, suppose that $f(\mu) > 0$. 
If $f(\lambda) = 0$, we are done.
So assume $f(\lambda) \neq 0$.
Define $g(x) = \frac{k+x}{m+x}$, so $f(\kappa) = \prod_{\square \in [\kappa]}g(c(\square))$.  
Then because $f(\lambda) \neq 0$, the content of every box in $\lambda$ must be larger than $-k$, and therefore $g(c(\square)) > 0$ for all $\square \in [\lambda]$.
Suppose the box being moved has content $c_\mu$ in $\mu$ and $c_\lambda$ in $\lambda$, so that $c_\lambda < c_\mu$ because $\mu \neq \lambda$.
Then calculate that
\begin{align}
  \frac{f(\lambda)}{f(\mu)} &= \frac{g(c_\lambda)}{g(c_\mu)}
\end{align}
The above is positive and less than 1, since $g$ is positive and monotonically increasing on $(-k, m)$:
\begin{align}
  g'(x) &=  \frac{(m+x) - (k+x)}{(m+x)^2}\\
  &= \frac{m - k}{(m+x)^2} > 0.
\end{align}
Therefore, we have proved the claim Eq.~\ref{eq:claim}.
\end{proof}
We need the next two definitions to state Thm.~\ref{thm:meangenbunchsuremajorization}.
\begin{definition}[Coupling]
A coupling of two probability distributions $\alpha \in \Delta^{N-1}, \beta \in \Delta^{M-1}$ is a joint distribution $\gamma \in \Delta^{MN-1}$ such that the marginals are $\alpha$ and $\beta$.
That is, $\sum_x \gamma(x, y) = \beta(y)$ and $\sum_y \gamma(x, y) = \alpha(x)$.
If $X \sim \alpha$, $Y \sim \beta$, and $(X, Y) \sim \gamma$, we also say that $(X, Y)$ is a coupling of $X$ and $Y$.
\label{def:coupling}
\end{definition}
\begin{definition}[Sure Majorization]
  Let $p, q \in \Delta(\mathrm{Part}_n)$ be distributions over the set of partitions of $n$.
  If $(\bm{\mu},\bm{\lambda})$ (Def.~\ref{def:coupling}) is a coupling of $\bm{\mu}\sim p$ and $\bm{\lambda}\sim q$  for which $\bm{\mu}\succ \bm{\lambda}$ (Def.~\ref{def:majorizationforpartitions}) with probability \(1\), we say that $\bm{\lambda}$ is majorized by $\bm{\mu}$ surely in this coupling.
  When such a coupling exists, we write $p\gtrdot q$.
  Let $\rho, \tau \in \mathcal{D}(\mathrm{Sym}^n(\mathcal{H}_V \otimes \mathcal{H}_H))$.
  If $D_\rho \gtrdot D_\tau$ (Def.~\ref{def:irrepdistribution}), then we also write $\rho \gtrdot \tau$.
  \label{def:suremajorization}
\end{definition}
The next three propositions (Prop.~\ref{prop:suremajorizationisapartialorder}, Prop.~\ref{prop:suremajorizationgreatestandleast}, Prop~\ref{prop:suremajorizationpreorderonstates}) give basic properties about this relation.
\begin{proposition}
  The relation $\gtrdot$ is a partial order on $\Delta(\mathrm{Part}_n)$, the set of distributions over partitions of $n$.
  \label{prop:suremajorizationisapartialorder}
\end{proposition}
\begin{proof}
  First we show reflexivity.
  Let $p\in \Delta(\mathrm{Part}_n)$.
  Then we have the following coupling of $p$ with itself,
  \begin{align}
    q(\lambda, \lambda') &= \delta_{\lambda, \lambda'}p(\lambda).
  \end{align}
  This is indeed a coupling of $p$ with itself, since $\sum_{\lambda\vdash n}q(\lambda, \lambda') = \sum_{\lambda\vdash n}\delta_{\lambda, \lambda'}p(\lambda) = p(\lambda')$, and similarly in its other argument.
  Furthermore, for any sample from this distribution, we have $\lambda = \lambda'$, so in particular $\lambda\succ \lambda'$.

  We now show that $\gtrdot$ is transitive.
  Let $p_0, p_1, p_2 \in \Delta(\mathrm{Part}_n)$ be such that $p_0 \gtrdot p_1$ and $p_1 \gtrdot p_2$.
  Let $(\bm{\lambda}_0, \bm{\lambda}_1)\sim q_{0, 1}$ be a coupling of $\bm{\lambda}_0\sim p_0$ and $\bm{\lambda}_1 \sim p_1$ such that $\bm{\lambda}_0 \succ \bm{\lambda}_1$ with probability 1, and define $q_{1, 2}$ similarly.
  Then define the function
  \begin{align}
    q_{0, 1, 2}(\lambda_0, \lambda_1, \lambda_2) &= 
    \begin{cases}
      \frac{q_{0, 1}(\lambda_0, \lambda_1)q_{1, 2}(\lambda_1, \lambda_2)}{p_1(\lambda_1)} &\text{ if $p_1(\lambda_1) \neq 0$}\\
      0 & \text{otherwise}
    \end{cases}
  \end{align}
  Since $\sum_{\lambda_0\vdash n}q_{0,1, 2}(\lambda_0, \lambda_1, \lambda_2) = q_{1, 2}(\lambda_1, \lambda_2)$ and $\sum_{\lambda_2\vdash n}q_{0,1, 2}(\lambda_0, \lambda_1, \lambda_2) = q_{0, 1}(\lambda_0, \lambda_1)$, $q_{0, 1, 2}$ is a joint distribution such that its marginals are $p_0, p_1$ and $p_2$.
  Therefore, $q_{0, 2} := \sum_{\lambda_1 \vdash n}q_{0, 1, 2}(\lambda_0, \lambda_1, \lambda_2)$ is a coupling of $p_0$ and $p_2$.
  It remains to show that $(\bm{\lambda}_0, \bm{\lambda}_2)\sim q_{0, 2}$ obeys $\bm{\lambda}_0\succ \bm{\lambda}_2$ with probability 1.
  Observe that $q_{0, 1, 2}(\lambda_0, \lambda_2|\lambda_1) = q_{0, 1}(\lambda_0|\lambda_1)q_{1, 2}(\lambda_2|\lambda_1)$, so $\bm{\lambda}_0\to\bm{\lambda}_1\to\bm{\lambda}_2$ is a Markov chain.
  We now state the intuition for the calculation that follows.
  Conditional on a sample $\lambda_0$, it must be that $\lambda_1 \prec \lambda_0$, and conditional on $\lambda_1$, it must be that $\lambda_2 \prec \lambda_1$, and since majorization is transitive, we can conclude that $\lambda_0 \succ \lambda_2$.
  Formally, suppose that $(\lambda_0, \lambda_2)$ is in the support of $q_{0,2}$.
  Then
  \begin{align}
    q_{0, 2}(\lambda_0, \lambda_2) &= p_0(\lambda_0)\sum_{\lambda_1\vdash n}q_{1,2}(\lambda_2|\lambda_1)q_{0, 1}(\lambda_1|\lambda_0) \neq 0\\
    \implies \qquad&\exists \lambda_1\text{ s.t. }\left( q_{1,2}(\lambda_2|\lambda_1)\neq 0 \text{ and } q_{0,1}(\lambda_1|\lambda_0)\neq 0\right)\\
    \implies \qquad&\exists \lambda_1\text{ s.t. }\left( \lambda_2\prec\lambda_1 \text{ and } \lambda_1\prec\lambda_0\right)\\
    \implies \qquad & \lambda_2 \prec \lambda_0.
  \end{align}
  Therefore $\gtrdot$ is transitive.

  Finally, we show antisymmetry.
  Suppose $p_0, p_1 \in \Delta(\mathrm{Part}_n)$ such that $p_0 \gtrdot p_1$ and $p_1 \gtrdot p_0$.
  By Strassen's theorem~\cite{Roch_2024} Thm.~4.2.11, $p_0(A) = p_1(A)$ on every upward closed $A$.
  By the M\"obius inversion formula~\cite{Stanley2011-ld} Prop.~3.7.1, this implies that $p_0(\lambda) = p_1(\lambda)$ for all $\lambda \vdash n$.
\end{proof}
\begin{proposition}
  The distribution $k^{(n)} \in \Delta(\mathrm{Part}_n)$ on partitions of $n$, given by $k^{(n)}(\lambda) = \delta_{\lambda, (n)}$ is a greatest element of $\gtrdot$, and $k^{(1^n)}\in \Delta(\mathrm{Part}_n)$, given by $k^{(1^n)}(\lambda) = \delta_{\lambda, (1^n)}$ is a least element.
  \label{prop:suremajorizationgreatestandleast}
\end{proposition}
\begin{proof}
  Let $p \in \Delta(\mathrm{Part}_n)$. Then the coupling $q^{(n)}(\lambda, \mu) = k^{(n)}(\lambda)p(\mu)$ is supported only on the set $\lambda\succ \mu$.
  Similarly, the coupling $q^{(1^n)}(\lambda, \mu) = k^{(1^n)}(\lambda)p(\mu)$ is supported only on the set $\lambda\prec \mu$.
\end{proof}
\begin{proposition}
  The relation $\gtrdot$ is a preorder on $\mathcal{D}(\mathrm{Sym}^n(\mathcal{H}_V \otimes \mathcal{H}_H))$.
  \label{prop:suremajorizationpreorderonstates}
\end{proposition}
\begin{proof}
  This follows by reflexivity and transitivity of $\gtrdot$ on $\Delta(\mathrm{Part}_n)$.
\end{proof}
\begin{theorem}
  Let $\rho, \tau \in \mathcal{D}(\mathrm{Sym}^n(\mathcal{H}_V \otimes \mathcal{H}_H))$ such that $\rho \gtrdot \tau$ (Def.~\ref{def:suremajorization}).
  Let $S\subseteq [m]$.
  Then
  \begin{align}
    \mathbb{E}_{U \sim \mathrm{Haar}(\mathrm{U}(\mathcal{H}_V))}\left(b(S|U, \rho\right)) &\ge \mathbb{E}_{U \sim \mathrm{Haar}(\mathrm{U}(\mathcal{H}_V))}\left(b(S|U, \tau\right)).
  \end{align}
  \label{thm:meangenbunchsuremajorization}
\end{theorem}
\begin{proof}
  By definition, we obtain a coupling $(\bm{\mu}, \bm{\lambda})$ of $\bm{\mu}\sim D_\rho$ and $\bm{\lambda}\sim D_\tau$ such that $\bm{\mu} \succ \bm{\lambda}$ with probability 1. Therefore by Lem.~\ref{lem:fschurconvex},
\begin{align}
  \mathbb{E}_{(\bm{\mu}, \bm{\lambda})}\left(f(\mu) - f(\lambda)\right) &\ge 0\\
  \iff \mathbb{E}_{\bm{\mu} \sim D_\rho}\left(f(\mu)\right) &\ge \mathbb{E}_{\bm{\lambda} \sim D_\tau}\left(f(\lambda)\right)\\
  \iff \mathbb{E}_{U \sim \mathrm{Haar}(\mathrm{U}(\mathcal{H}_V))}\left(b(S|U, \rho\right)) &\ge \mathbb{E}_{U \sim \mathrm{Haar}(\mathrm{U}(\mathcal{H}_V))}\left(b(S|U, \tau\right)),
\end{align}
where the last equivalence is due to Lem.~\ref{lem:meangenbunchgeneral}.
\end{proof}
We now restate and prove the following corollary. Recall the definition of the perfectly indistinguishable state Eq.~\ref{eq:perfectlyindsitinguishable}.
\avglieb*
\begin{proof}
  Apply Thm.~\ref{thm:meangenbunchsuremajorization}, Prop.~\ref{prop:pefectlyindistinguishableisperfectlyindistinguishable}, and Prop.~\ref{prop:suremajorizationgreatestandleast}.
\end{proof}

We now turn to proving Cor.~\ref{cor:schurconvexity}, for which we need the following lemma.
\begin{lemma}
  Let $n \le m = \mathrm{dim}(\mathcal{H}_V)$ and $n \le L = \mathrm{dim}(\mathcal{H}_H)$.
  Let $\alpha\in \Delta^{L-1}$ be a distribution, and let $\rho(\alpha)$ be a uniform state (Eq.~\ref{eq:alphastate}).
  Then $D_{\rho(\alpha)} = \mathrm{SW}^n(\alpha)$.
  \label{lem:uniformstateirrepdistribution}
\end{lemma}
\begin{proof}
  Immediate from Lem.~\ref{lem:modelderivation}.
\end{proof}
Thm.~1.5.4 of~\cite{wright2016learn}, which we state next, provides a connection between the Schur-Weyl distribution $\mathrm{SW}^n(\alpha)$ and sure majorization (Def.~\ref{def:suremajorization}).
\begin{theorem}[Thm.~1.5.4 of \cite{wright2016learn}] Let
  $\alpha, \beta \in \Delta^{d-1}$ with $\beta \succ \alpha$.  Let
  $\bm{\mu}\sim \mathrm{SW}^n(\beta), \bm{\lambda} \sim
  \mathrm{SW}^n(\alpha)$.  Then for every 
  $n \in \mathbb{N}$, there is a coupling $(\bm{\mu}, \bm{\lambda})$
  of $\mathrm{SW}^n(\alpha)$ and $\mathrm{SW}^n(\beta)$ such that
  $\bm{\mu} \succ \bm{\lambda}$ with probability 1. 
  \label{thm:wright154}
\end{theorem}
By Lem.~\ref{lem:uniformstateirrepdistribution}, this implies that if $\alpha, \alpha' \in \Delta^{L-1}$ are such that $\alpha \succ \alpha'$, then $\rho(\alpha)\gtrdot \rho(\alpha')$.
We now restate and prove Cor.~\ref{cor:schurconvexity}.
\schurconvexity*
\begin{proof}
  By Thm.~\ref{thm:wright154}, we obtain a coupling $(\bm{\mu},\bm{\lambda})$ of $\bm{\mu}\sim \mathrm{SW}^n(\alpha)$ and $\bm{\lambda}\sim\mathrm{SW}^n(\alpha')$ for which $\bm{\mu}\succ \bm{\lambda}$ with probability \(1\).
  The corollary follows by Lem.~\ref{lem:meangenbunchgeneral}, Lem.~\ref{lem:uniformstateirrepdistribution}, and Thm.~\ref{thm:meangenbunchsuremajorization}.
\end{proof}

We now turn to proving Cor.~\ref{cor:thermcor}, for which we need the next two lemmas.
\begin{lemma}
  Let $\alpha_0 \in \Delta^{L-1}$ be a distribution such that no element occurs with probability 1, and let $\alpha_{\infty} = \left( 1, 0, \ldots, 0 \right) \in \Delta^{L-1}$.
  Let $S \subseteq [m]$.
  Then 
  \begin{align}
    \mathbb{E}_{U \sim \mathrm{Haar}(\mathrm{U}(\mathcal{H}_V))}\left(b(S|U, \rho(\alpha_{\infty})\right)) > \mathbb{E}_{U \sim \mathrm{Haar}(\mathrm{U}(\mathcal{H}_V))}\left(b(S|U, \rho(\alpha_0)\right)).
    \label{eq:strictinequality}
  \end{align}
  \label{lem:strictinequality}
\end{lemma}
\begin{proof}
  We use $f$ as defined in Lem.~\ref{lem:fschurconvex}.
  Observe that $\mathrm{SW}^n(\alpha_{\infty})(\lambda) = \delta_{\lambda, (n)}$.
  On the other hand, $\mathrm{SW}^n(\alpha)$ must have support off of $(n)$.
  If $\lambda$ is a partition of $n$ such that $\lambda \neq (n)$, then $(n)\succ \lambda$, so by Lem.~\ref{lem:fschurconvex}, $f(\lambda) < f( (n))$.
  Therefore, by Lem.~\ref{lem:uniformstateirrepdistribution} and Lem.~\ref{lem:meangenbunchgeneral}, we conclude the inequality Eq.~\ref{eq:strictinequality}.
\end{proof}

\begin{lemma}
Let $0=\epsilon_0\le \cdots\le\epsilon_d$.
  For $\beta \ge 0$, define
  \begin{align}
    (\alpha_\beta)_i = \frac{e^{-\beta\epsilon_i}}{\sum_{j}e^{-\beta\epsilon_j}}.
  \end{align}
  Then,
  \begin{align}
    \beta \ge \beta' \implies \alpha_\beta \succ \alpha_{\beta'}.
  \end{align}
  \label{lem:gibbsmaj}
\end{lemma}
\begin{proof}
  Let $k \in \left\{ 0, \ldots, d \right\}$.
  We want to show that
  \begin{align}
    f_k = \sum_{i=0}^k (\alpha_\beta)_i = \frac{\sum_{i=0}^ke^{-\beta\epsilon_i}}{\sum_{j=0}^d e^{-\beta\epsilon_j}}
  \end{align}
  is monotonically nondecreasing in \(\beta\).
  A sufficient condition is that the derivative of $\log(f_k)$ with respect to $\beta$ is nonnegative.
  Define the partition function (in the statistical mechanical sense, not the number theory sense) $Z_j = \sum_{i=0}^j e^{-\beta\epsilon_i}$ so that $\log f_k = \log Z_k - \log Z_d$.
  Then $Z_k$ is a partition function of a system $S_k$ with energy levels $\left\{ \epsilon_i \right\}_{i=1}^k$, and $Z_d$ is a partition function of a system $S_d$ with energy levels $\left\{ \epsilon_i \right\}_{i=1}^d$, which contains $S_k$ as a subsystem.
  Then
  \begin{align}
    \frac{\partial}{\partial \beta}\log f_k &=      \frac{\partial}{\partial \beta}\log Z_k - \frac{\partial}{\partial \beta}\log Z_d\\
    &= -\left\langle E \right\rangle_{S_k} + \left\langle E \right\rangle_{S_d},
    \label{eq:expectedenergy}
  \end{align}
  where $\left\langle E \right\rangle_S$ is the expected energy of the system $S$ at inverse temperature $\beta$.
  Since $S_d$ is just $S_k$ with some additional higher energy levels, and they are at the same temperature, it must have higher expected energy.
  Therefore, we should expect that Eq.~\ref{eq:expectedenergy} is nonnegative.
  Formally, we calculate
  \begin{align}
-\left\langle E \right\rangle_{S_k} + \left\langle E \right\rangle_{S_d}
&= \frac{\sum_{i=0}^k(-\epsilon_i)e^{-\beta\epsilon_i}}{(\sum_{i=0}^k e^{-\beta\epsilon_i})} - \frac{\sum_{j=0}^d(-\epsilon_j) e^{-\beta\epsilon_j}}{(\sum_{j=0}^d e^{-\beta\epsilon_j})}\\
&= \frac{(\sum_{j=0}^d e^{-\beta\epsilon_j})\sum_{i=0}^k(-\epsilon_i)e^{-\beta\epsilon_i} - (\sum_{i=0}^k e^{-\beta\epsilon_i})\sum_{j=0}^d(-\epsilon_j) e^{-\beta\epsilon_j}}{(\sum_{i=0}^k e^{-\beta\epsilon_i})(\sum_{j=0}^d e^{-\beta\epsilon_j})}
    \label{eq:majcalc0}
  \end{align}
  This is nonnegative exactly when its numerator is, so we calculate that the numerator of the right-hand side of Eq.~\ref{eq:majcalc0} is
  \begin{align}
    \sum_{j=0}^d \sum_{i=0}^k(-\epsilon_i + \epsilon_j)e^{-\beta(\epsilon_i+\epsilon_j)}
    \label{eq:majcalc2}
  \end{align}
Now we change summation variables to $a = i+j$ and $b = (i-j)$.
Define $\Delta_{a, b} = \epsilon_i - \epsilon_j$, and define $\tau_{a, b} = \epsilon_i + \epsilon_j$.
Then $\Delta_{a, b} = -\Delta_{a, -b}$ and $\tau_{a, b} = \tau_{a, -b}$.
Furthermore, $b \ge 0 \implies \Delta_{a, b}\ge 0$, and $b \le 0\implies \Delta_{a, b} \le 0$.
Then, the expression in Eq.~\ref{eq:majcalc2} is
\begin{align}
  \sum_{a = 0}^{d + k}\sum_{\substack{b = \mathrm{max}(-d, -a)\\ \text{step by 2}}}^{\mathrm{min}(k, a)}-\Delta_{a, b}e^{-\beta\tau_{a, b}}
  &= 
  \sum_{a = 0}^{d + k}\sum_{\substack{b = \mathrm{max}(-d, -a)\\ \text{step by 2}}}^{-\mathrm{min}(k, a)-2}-\Delta_{a, b}e^{-\beta\tau_{a, b}}
  \ge 0.
  \label{eq:majcalc3}
\end{align}
Therefore, $\frac{\partial}{\partial\beta}f_k \ge 0$, so $\log f_k$ is nondecreasing, and we have proven the lemma.
\end{proof}

We now restate and prove Cor.~\ref{cor:thermcor}.
\thermcor*
\begin{proof}
  Write $y(\beta) = \mathbb{E}_{U \sim \mathrm{Haar}(\mathrm{U}(\mathcal{H}_V))}\left(b(S|U, \rho(\alpha_\beta)\right))$ for the mean generalized bunching probability.
It is shown in Lem.~\ref{lem:gibbsmaj} that $\beta \ge \beta' \implies \alpha_\beta \succ \alpha_{\beta'}$.
Therefore, by Cor.~\ref{cor:schurconvexity}, the mean generalized bunching probability is monotonically nondecreasing as a function of inverse temperature $\beta$.
If not all the energies $\{\epsilon_i\}_i$ are equal, then $\alpha_0$ is equal to a uniform distribution over more than one letter, and therefore is not equal to a distribution for which some outcome occurs with probability one.
Observe also that $\lim_{\beta\to\infty}(\alpha_\beta)_k = \delta_{1, k}$.
Then from Lem.~\ref{lem:strictinequality}, $y(0)\neq\lim_{\beta\to\infty}y(\beta)$.
The mean generalized bunching probability is a polynomial as a function of the $\alpha_\beta$, and $\alpha_\beta$ is analytic in $\beta$.
Therefore, $y(\beta)$ cannot be constant on an open interval.
  Therefore, $y(\beta)$ is strictly monotonically increasing as a function of $\beta$.
\end{proof}

\section{Completeness of the partially labelled states}
\label{sec:completeness}

In this section we prove that the visible parts of the partially labelled states are a basis for the visible parts of permutation-invariant states.
A consequence of this is that the distribution over visible occupations for a permutation-invariant state is in the span of the corresponding distributions for partially labelled states.

We make use of the following inner product on class functions.
For $G$ a group, $f:G\rightarrow \mathbb{C}$ is called a class function on $G$ if $f(ghg^{-1}) = g$ for all $g, h \in G$.
\begin{definition}[\cite{fultonRepresentationTheory2004}, Eq.~2.11]
  Let $G$ be a finite group, and let $\phi_1,\phi_2$ be class functions on $G$.
  Then
  \begin{align}
    \left\langle \phi_1, \phi_2 \right\rangle_G &= \frac{1}{|G|}\sum_{g \in G}\phi_1(g)^*\phi_2(g)
    \label{eq:defip}
  \end{align}
  is an inner product on class functions on $G$.
  \label{def:classfnip}
\end{definition}
We drop the subscript $G$ on the inner product when it is clear which group is being used.
The irreducible characters of $G$ form an orthonormal basis of the class functions of $G$ under this inner product, see~\cite{fultonRepresentationTheory2004} Thm.~2.12.
If $V, W$ are representations of $G$ with characters $\chi_V, \chi_W$ respectively, we overload notation and write $\left\langle V, W \right\rangle$ to mean $\left\langle \chi_V, \chi_W \right\rangle$.
\begin{corollary}[\cite{fultonRepresentationTheory2004} Cor.~2.16]
  Let $G$ be a finite group, let $V$ be a representation of $G$ and let $W$ be an irrep of $G$.
  Then $\left\langle V, W \right\rangle$ is equal to the multiplicity of $W$ in $V$.
\end{corollary}

For $\lambda\vdash n$, it is naturally associated to the weight $g_\lambda:[n]\to \left\{ 0,\ldots, n \right\}$ given by $g_\lambda(i) = \lambda_i$, where we have extended $\lambda$ by the appropriate number of zeros.
We say that a tableau has weight $\lambda$ if it has weight $g_\lambda$.
We make use of the following numbers in Prop.~\ref{prop:partiallylabeledvisiblestate}.
\begin{definition}[Kostka numbers~\cite{fultonRepresentationTheory2004}, p. 56]
  Let $\mu, \lambda \vdash n$.
The Kostka number $K_{\mu, \lambda}$ is the number of semistandard tableaux of shape $\mu$ with weight $\lambda$.
\end{definition}
The Kostka numbers arise as a certain multiplicity in the representation theory of the symmetric groups.
\begin{lemma}[\cite{fultonYoungTableauxApplications1996}~$\mathsection7.3$ Cor.~1 and \cite{fultonRepresentationTheory2004}~Cor.~3.20]
Let $\mu, \lambda \vdash n$.
  Let $\mathrm{Triv}$ be the trivial representation of $\mathcal{S}_\mu = \mathcal{S}_{\mu_1}\times \cdots \times \mathcal{S}_{\mu_{\mathrm{len}(\mu)}}$.
  Then the Kostka numbers give the multiplicity
  \begin{align}
 K_{\lambda,\mu}&=    \left\langle \mathrm{Res}^{\mathcal{S}_n}_{\mathcal{S}_\mu}\mathrm{Sp}_\lambda, \mathrm{Triv}\right\rangle_{\mathcal{S}_\mu}.
  \end{align}
  \label{lem:kostkamultiplicity}
\end{lemma}

  We are now prepared to prove the main result of this section.
\begin{proposition}
  Let $L = \mathrm{dim}(\mathcal{H}_H) \ge n$.
  Let $\rho_{\bm{i}, \mu}$ be a partially labelled state (Def.~\ref{def:partiallylabelledstate}) with label pattern $\mu$, occupying the $n$ distinct sites $\bm{i}$.
  Then the visible state for $\rho_{\bm{i}, \mu}$ is
  \begin{align}
    \mathrm{Tr}_H(\rho_{\bm{i}, \mu}) &= \frac{\mu!}{n!}\bigoplus_{\lambda\vdash n}K_{\lambda, \mu}\frac{\mathds{1}_{(\mathbb{S}_\lambda(\mathcal{H}_H))_{\mathrm{wt} = \xi(\bm{i})}}}{\mathrm{dim}(\lambda)}
  \end{align}
  \label{prop:partiallylabeledvisiblestate}
\end{proposition}
\begin{proof}
  The state $\rho_{\bm{i}, \mu}$ is permutation invariant (Def.~\ref{def:permutationinvariantstates}), so by Lem.~\ref{lem:mixturemodel} it takes the form $\bigoplus_{\substack{\lambda \vdash n}}q_\lambda \frac{\mathds{1}_{(\mathbb{S}_\lambda(\mathcal{H}_V))_{\mathrm{wt} = \xi(\bm{i})}}}{\mathrm{dim}(\lambda)}$ for some probability distribution $\left\{ q_\lambda \right\}_{\lambda\vdash n}$.
We claim that
\begin{align}
  q_\lambda = \Tr(\Theta_\lambda \rho_{\bm{i}, \mu}),
\end{align}
where $\Theta_\lambda$ is the $\lambda$-isotypic projector (Lem.~\ref{lem:isotypicprojection}) of $K_{\bm{i}}$ (Eq.~\ref{eq:singlyoccupiesdef}).
Let $\Pi_\lambda$ be the projector onto the $\lambda$-subspace of $\bigoplus_{\mu\vdash n}(\mathbb{S}_\mu(\mathcal{H}_V))_{\mathrm{wt} = \xi(\bm{i})}$.
Then $\Pi_\lambda\otimes \mathds{1}_{\mathbb{S}_\lambda(\mathcal{H}_H)} = \Theta_\lambda$.
Therefore,
\begin{align}
  q_\lambda &=  \Tr(\Tr_H(\rho_{\bm{i}, \mu})\Pi_\lambda)\\
  &= \Tr\left( \rho_{\bm{i}, \mu} \Theta_\lambda\right).
\end{align}
We now calculate $\Tr\left( \rho_{\bm{i}, \mu} \Theta_\lambda\right)$.
The $\lambda$-isotypic projector of $K_{\bm{i}}$ is given by
  \begin{align}
    \Theta_\lambda = \frac{1}{n!}\sum_{\sigma\in\mathcal{S}_n}\chi_\lambda(\sigma)F_{\bm{i}}(\sigma^{-1})\otimes \mathds{1}_{\mathcal{H}_H}.
  \end{align}
  Therefore,
  \begin{align}
    \Tr(\Theta_\lambda\rho_{\bm{i},\mu}) &= \Tr\left( \frac{1}{n!}\sum_{\sigma\in\mathcal{S}_n}\chi_\lambda(\sigma)F_{\bm{i}}(\sigma^{-1})\otimes \mathds{1}_{\mathcal{H}_H}\frac{1}{n!}\sum_{\tau\in \mathcal{S}_n} a^\dagger(\tau\cdot\bm{i}, \overline{\mu})\ketbra{0}{0}a(\tau\cdot\bm{i}, \overline{\mu}),
 \right)\\
 &= \frac{1}{(n!)^2}\sum_{\sigma, \tau\in \mathcal{S}_n}\chi_\lambda(\sigma)\Tr(a^\dagger(\bm{i}, \sigma^{-1}\tau^{-1}\cdot\overline{\mu})\ketbra{0}{0}a(\bm{i}, \tau^{-1}\cdot\overline{\mu}))
  \end{align}
  The summands are nonzero exactly when $\sigma$ is in the stabilizer group of $\tau^{-1}\cdot \overline{\mu}$.
  This stabilizer group is $\mathcal{S}_\mu = \mathcal{S}_{\mu_1}\times \cdots \times \mathcal{S}_{\mu_{\mathrm{len}(\mu)}}$.
  Thus,
  \begin{align}
    \Tr(\Theta_\lambda\rho_{\bm{i},\mu}) &= \frac{1}{(n!)^2}\sum_{\sigma\in\mathcal{S}_\mu, \tau\in\mathcal{S}_n}\chi_\lambda(\sigma)\\
    &= \frac{\mu!}{n!}\left\langle\mathrm{Res}^{\mathcal{S}_n}_{\mathcal{S}_\mu}\kappa_\lambda, \mathrm{Triv} \right\rangle_{\mathcal{S}_\mu}\\
    &= \frac{\mu!}{n!}K_{\lambda,\mu},
  \end{align}
  by Lem.~\ref{lem:kostkamultiplicity}.
\end{proof}

\begin{corollary}
  With the assumptions of Prop.~\ref{prop:partiallylabeledvisiblestate}, the states $\{\Tr_H(\rho_{\bm{i}, \mu})\}_{\mu\vdash n}$ are linearly independent.
  \label{cor:perminvvislinind}
\end{corollary}
\begin{proof}
  Observe that $\{\Tr_H(\rho_{\bm{i}, \mu})\}_{\mu\vdash n}$ is linearly independent exactly when $\{(n!/\mu!)\Tr_H(\rho_{\bm{i}, \mu})\}_{\mu\vdash n}$ is.
  Since the Kostka matrix is full rank~(\cite{Stanley_Fomin_1999} Sec.~7.16) and the states $\left\{\frac{\mathds{1}_{(\mathbb{S}_\lambda(\mathcal{H}_H))_{\mathrm{wt} = \xi(\bm{i})}}}{\mathrm{dim}(\lambda)}\right\}_{\lambda\vdash n}$ are linearly independent, so too are the $\{\Tr_H(\rho_{\bm{i}, \mu})\}_{\mu\vdash n}$.
\end{proof}

\begin{corollary}
  Let $\bm{i}$ be a list of $n$ distinct indices from $[m]$, and let $\rho$ be a permutation-invariant state that singly occupies $\bm{i}$. Let $L = \mathrm{dim}(\mathcal{H}_H) \ge n$. 
  Then there exists a set of real coefficients $\left\{ c_\mu \right\}_{\mu\vdash n}$ such that, for all $U \in \mathrm{U}(\mathcal{H}_V)$ and all $v \in \omega_{[m]}^n$, the visible occupation distribution (Eq.~\ref{eq:visibleoccupationdistributions}) can be written as the linear combination
\begin{align}
  p(v|U, \rho) &= \sum_{\mu\vdash n}c_\mu p(v|U, \rho_{\bm{i}, \mu}).
\end{align}
\label{cor:perminvpdistlinind}
\end{corollary}
\begin{proof}
  By Lem.~\ref{lem:modelafterhiddentrace}, the visible occupation distribution for $\rho$ is
  \begin{align}
    p(v|U, \rho) &= \sum_{\lambda \vdash n}\Tr(\pi_\lambda(U) \Tr_H(\rho_{(\lambda)}) \pi_\lambda(U)^\dagger \Pi_{v}^\lambda).
  \end{align}
  By Cor.~\ref{cor:perminvvislinind}, $\Tr_H(\rho)\in \mathrm{span}(\left\{  \Tr_H(\rho_{\bm{i}, \mu})|\mu \vdash n\right\})$, and therefore
  \begin{align}
    \Tr_H(\rho) &= \sum_{\mu\vdash n}c_\mu\Tr_H(\rho_{\bm{i}, \mu}).
  \end{align}
  Therefore,
  \begin{align}
    p(v|U, \rho) &= \sum_{\lambda \vdash n}\Tr(\pi_\lambda(U) \sum_{\mu\vdash n}c_\mu\Tr_H((\rho_{\bm{i}, \mu})_{(\lambda)}) \pi_\lambda(U)^\dagger \Pi_{v}^\lambda)\\
    &= \sum_{\mu\vdash n}c_\mu p(v|U, \rho_{\bm{i}, \mu}).
  \end{align}
\end{proof}

\section{Comments on the Auxiliary Space}
\label{sec:commentsonentanglement}
We now offer some comments on the auxiliary state space that is the image of the map $h$ (Def.~\ref{def:auxiliarystate}).
Let $\bm{i}$ be a list of $n$ distinct indices from $[m]$.
Let $\rho$ be a perfectly indistinguishable state (Def.~\ref{def:perfectlyindistinguishablestategeneral}) that singly occupies the sites $\bm{i}$ (Eq.~\ref{eq:singlyoccupiesdef}).
Then by Lem.~\ref{lem:hintertwinesfandp} and Thm.~3 of Ref.~\citenum{harrowChurchSymmetricSubspace2013}, $\rho$ satisfies
\begin{align}
  (F_{\bm{i}}(\sigma)\otimes \mathds{1})\cdot \rho &=  \rho \cdot (F_{\bm{i}}(\sigma)\otimes \mathds{1})^{-1} = \rho\qquad \forall\,\sigma \in \mathcal{S}_n.
\end{align}
Therefore, any perfectly indistinguishable state that singly-occupies $\bm{i}$ is also permutation invariant (Def.~\ref{def:permutationinvariantstates}).

The wavefunction corresponding to the perfectly indistinguishable state with single-particle hidden state \(\ket{\psi}\in\mathcal{H}_{N}\) (Eq.~\ref{eq:perfectlyindsitinguishable}) is given by
\begin{align}
  \ket{f_{\bm{i}, \psi}} &= a^{\dagger}_{i_1, \psi}\cdots a^\dagger_{i_n, \psi}\ket{0},
\end{align}
so that $\phi_{\bm{i}, \psi} = \ketbra{f_{\bm{i}, \psi}}$.
By Thm.~3 of Ref.~\citenum{harrowChurchSymmetricSubspace2013} and Lem.~\ref{lem:hintertwinesfandp}, the set of perfectly indistinguishable states that singly occupy $\bm{i}$ is spanned by the states $\left\{ \ket{f_{\bm{i}, \psi}} \right\}_{\ket{\psi} \in \mathcal{H}_H}$.

We now turn to a notion of entanglement.
\begin{definition}[Auxiliary-entangled state]
Let $\bm{i}$ be a list of $n$ distinct indices from $[m]$.
Let $\rho\in \mathcal{D}(K_{\bm{i}})$ be a state that singly occupies $\bm{i}$ (Def.~\ref{eq:singlyoccupiesdef}).
If the auxiliary state $h(\rho)$ (Def.~\ref{def:auxiliarystate}) is entangled, then $\rho$ is said to be auxiliary entangled.
If $\rho$ is not auxiliary entangled, it is called auxiliary separable.
\label{def:auxiliaryentangled}
\end{definition}
Auxiliary-separable states are sometimes called ``classically correlated'' states in the literature~\cite{shchesnovichUniversalityGeneralizedBunching2016}.
Auxiliary-entangled states are typically more difficult to prepare in the photonics platform, so we now offer some comments on the role of these states in our work.
Since there are states of the symmetric subspace of $(\mathbb{C}^d)^{\otimes n}$ that are entangled~\cite{harrowChurchSymmetricSubspace2013}, there are perfectly indistinguishable states that singly occupy $\bm{i}$ that are auxiliary entangled, by Lem.~\ref{lem:hintertwinesfandp}.

The set of permutation-invariant states that singly occupy $\bm{i}$ includes other states that are auxiliary entangled.
For example, consider the following state, given as a Slater determinant,
\begin{align}
  \ket{\text{ferm}} = \frac{1}{n!}\sum_{\sigma \in \mathcal{S}_n}\chi_{(1^n)}(\sigma)\ket{\sigma \cdot (1, \ldots, n)} \in \mathcal{H}_H^{\otimes n},
\end{align}
where $\chi_{(1^n)}(\sigma)$ is the sign of the permutation $\sigma$.
The state $\ket{\text{ferm}}$ is an entangled state, and $P(\sigma) \ketbra{\text{ferm}}P(\sigma)^{-1} = \ketbra{\text{ferm}}$ for all $\sigma \in \mathcal{S}_n$, where $P(\sigma)$ is defined in Eq.~\ref{eq:permutationtensorfactoraction}.
Therefore, by Lem.~\ref{lem:hintertwinesfandp}, $h^{-1}(\ketbra{\text{ferm}})$ is a permutation-invariant state that is auxiliary entangled.
This means that the weak generalized-bunching conjecture (Conj.~\ref{conj:weakgenbunch}) applies to some states that are auxiliary entangled.
We did not check whether the states $\phi_{\bm{i}}^\lambda$ (Prop.~\ref{prop:deltastate}) are auxiliary entangled for a given $\lambda\vdash n$.
These states are mixed in general, so checking whether their auxiliary states are entangled requires a more detailed analysis.

\end{document}

%% file: defs.tex
\providecommand{\ignore}[1]{}

\newif\ifcmnt
\cmnttrue
\ifdefined\cmntsoff\cmntfalse\fi

\ifcmnt
    \providecommand{\aucmnt}[1]{#1}

\else
    \providecommand{\aucmnt}[1]{}

\fi

\numberwithin{equation}{section}
\newtheorem{theorem}{Theorem}[section]
\newtheorem*{theorem*}{Theorem}
\newtheorem{lemma}[theorem]{Lemma}
\newtheorem*{lemma*}{Lemma}
\newtheorem{corollary}[theorem]{Corollary}
\newtheorem{fact}[theorem]{Fact}
\newtheorem{proposition}[theorem]{Proposition}

\theoremstyle{definition}
\newtheorem{definition}[theorem]{Definition}
\newtheorem{conjecture}[theorem]{Conjecture}

\theoremstyle{remark}

\newcommand{\R}{\mathbb{R}}
\newcommand{\C}{\mathbb{C}}

\let\originalleft\left
\let\originalright\right
\renewcommand{\left}{\mathopen{}\mathclose\bgroup\originalleft}
\renewcommand{\right}{\aftergroup\egroup\originalright}
